\documentclass[
notitlepage,
floatfix,
aps,pra,
reprint,
twocolumn,
superscriptaddress
]{revtex4-2}
\usepackage[lmargin=.7in,rmargin=.7in,tmargin=.6in,bmargin=1in]{geometry}
\usepackage{amsmath,amssymb,amsfonts,enumitem}
\usepackage{amsthm}
\usepackage{adjustbox}
\usepackage{tabularx}
\usepackage{booktabs}
\usepackage{minitoc}
\usepackage[normalem]{ulem}
\usepackage{tikz}
\usepackage{bbm}
\usepackage{graphicx}
\usepackage{textcomp}
\usepackage{braket}
\usepackage{mathtools}
\usepackage{bm}
\usepackage{times}
\usepackage[ruled,vlined]{algorithm2e}
\usepackage{algcompatible}

\renewcommand{\Re}{\operatorname{Re}}

\usepackage[colorlinks]{hyperref}
\hypersetup{
    linkcolor=blue,
    citecolor=blue,
    urlcolor=blue
}

\usepackage{amsmath,amssymb,amsfonts}
\usepackage{comment}
\newtheorem{theorem}{Theorem}

\newtheorem{lemma}{Lemma}

\newtheorem{definition}{Definition}
\newtheorem{prop}{Proposition}

\usepackage{braket}

\usepackage[normalem]{ulem}
\newcommand\erase{\bgroup\markoverwith{\textcolor{red}{\rule[.5ex]{2pt}{0.4pt}}}\ULon}
\newcommand\eraseblue{\bgroup\markoverwith{\textcolor{blue}{\rule[.5ex]{2pt}{0.8pt}}}\ULon}

\newcommand{\EPFLQSE}{Centre for Quantum Science and Engineering, \'{E}cole Polytechnique F\'{e}d\'{e}rale de Lausanne (EPFL), Lausanne, Switzerland}
\newcommand{\EPFL}{Institute of Physics, \'{E}cole Polytechnique Fédérale de
Lausanne (EPFL), Lausanne, Switzerland}
\newcommand{\KEIO}{Quantum Computing Center, Keio University, Hiyoshi 3-14-1, Kohoku-ku, Yokohama 223-8522, Japan}
\newcommand{\CHULA}{Chula Intelligent and Complex Systems Lab, Department of Physics, Faculty of Science, Chulalongkorn University, Bangkok, Thailand}

\newcommand{\hiddenlink}[2]{{\hypersetup{hidelinks}\hyperlink{#1}{#2}}}
\usepackage{tocloft}

\usepackage{dsfont}

\begin{document}

\doparttoc 
\faketableofcontents
\part{}

\title{
Double-bracket quantum algorithms for thermal state preparation
}

\author{Andrew Wright}
\affiliation{\EPFL} \affiliation{\EPFLQSE} 
\author{Reyhaneh Aghaei Saem}
\affiliation{\EPFL} \affiliation{\EPFLQSE} 
\author{Supanut Thanasilp}
\affiliation{\CHULA}
\affiliation{\EPFL} \affiliation{\EPFLQSE} 

\author{Yudai Suzuki}
\email{yudai.suzuki@epfl.ch}
\affiliation{\EPFL} \affiliation{\EPFLQSE} \affiliation{\KEIO}
\author{Zo\"{e} Holmes}
\affiliation{\EPFL} \affiliation{\EPFLQSE}

\date{\today}

\begin{abstract}

We propose quantum algorithms for preparing thermal states via the simulation of the thermofield double states.
The key idea is to leverage double-bracket quantum algorithms to implement imaginary-time evolution on thermofield double states, whose reduced state realizes the Gibbs state.
Our method, termed double-bracket thermofield double (\texttt{DB-TFD}), introduces two variants. The first, the \texttt{vanilla} \texttt{DB-TFD} algorithm, directly implements imaginary-time evolution using double-bracket quantum imaginary-time evolution. The second, \texttt{poly} \texttt{DB-TFD}, employs double-bracket quantum signal processing to approximate the imaginary-time evolution operator via a polynomial transformation.
We demonstrate that the complexity of the \texttt{poly} \texttt{DB-TFD} algorithm scales exponentially with the inverse temperature in a broad practical regime.
This scaling is consistent with existing methods, and 
numerical simulations support the corresponding theoretical bound. 
We further demonstrate the utility of \texttt{DB-TFD} in quantum Boltzmann machines for generative modeling, achieving improved performance compared with variational imaginary-time evolution approaches. 
These results establish \texttt{DB-TFD} as a promising route for thermal state preparation in the near-term and early-fault-tolerant regimes.
\end{abstract}

\maketitle

\section{Introduction}

Thermal state preparation is a fundamental task in quantum physics, with broad relevance across quantum many-body physics~\cite{poulin2009preparing, lee2004dopingmottinsulatorphysics} and high-energy physics~\cite{Maldacena_2013}. 
Beyond their physical significance, thermal states have also found applications in purely computational settings, including optimization~\cite{brandao2017quantumspeedupssemidefiniteprogramming, ingber2000adaptive} and machine learning~\cite{Biamonte_2017, Amin_2018}. 
Despite this wide-ranging importance, preparing Gibbs states of general many-body Hamiltonians is computationally challenging: the problem is QMA-hard in the low-temperature limit for quantum approaches~\cite{Aharonov_2009} and NP-hard for classical spin glasses~\cite{Barahona_1982}.
Nevertheless, a substantial body of work has been devoted to developing more efficient algorithms for thermal state preparation~\cite{chen2023quantumthermalstatepreparation, Rall_2023, ding2025endtoendefficientquantumthermal,Gily_n_2019, Holmes_2022, chowdhury2016quantumalgorithmsgibbssampling, Poulin_2009, wu2019variational, Zhu_2020, sewell2022thermalmultiscaleentanglementrenormalization, Sagastizabal_2021, Zoufal_2021}

Quantum computing, as a physics-oriented computational platform, provides a natural framework for simulating thermal properties of quantum systems. 
With the focus of implementation on fault-tolerant quantum computers, a variety of algorithms for Gibbs state preparation have been proposed, including approaches based on Lindbladian simulations~\cite{chen2023quantumthermalstatepreparation, Rall_2023, ding2025endtoendefficientquantumthermal} and imaginary-time evolution-based frameworks~\cite{Gily_n_2019, Holmes_2022}.
These approaches provide rigorous complexity guarantees and asymptotic efficiency. 
However, despite their strong theoretical performance, many of these methods remain impractical for currently available quantum hardware because sophisticated primitives such as block-encoding and linear combination of unitaries~\cite{Gily_n_2019, chowdhury2016quantumalgorithmsgibbssampling}, and other resource-intensive subroutines~\cite{chen2023quantumthermalstatepreparation,Poulin_2009}, are required.

On the other hand, many near-term quantum algorithms for the task have also been developed.
A large class of these methods is variational and relies on hybrid quantum-classical optimization loops~\cite{Motta_2019, wu2019variational, Zhu_2020, sewell2022thermalmultiscaleentanglementrenormalization, Sagastizabal_2021, Zoufal_2021}. 
While feasible for small systems, variational thermal state preparation often encounters fundamental difficulties, such as barren plateaus where gradients vanish exponentially with the system size~\cite{McClean2018barrenplateaus, Larocca_2025}, local minima~\cite{anschuetz2022critical}, and the quantum measurement problem~\cite{gonthier2022measurement}. 
As a result, the practical utility of these strategies on currently available noisy intermediate-scale quantum devices remains uncertain.
These challenges motivate the search for non-variational methods that can operate in an intermediate regime.

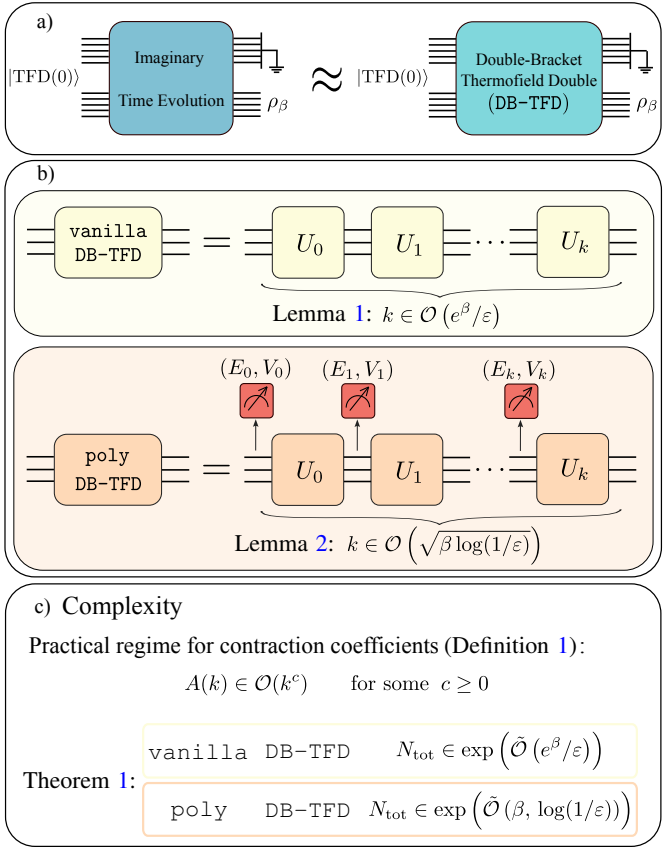
\begin{figure}[t]
    \centering
    \begin{tikzpicture}
    \definecolor{lightgray}{HTML}{F4F4F4}
    \definecolor{littlelightgray}{HTML}{ecececff}
    \definecolor{pale}{HTML}{7ca3d4ff}
    \definecolor{lightred}{HTML}{d8a2a2}
    
    \node[anchor=center]
    {\centering\includegraphics[width=1\linewidth]{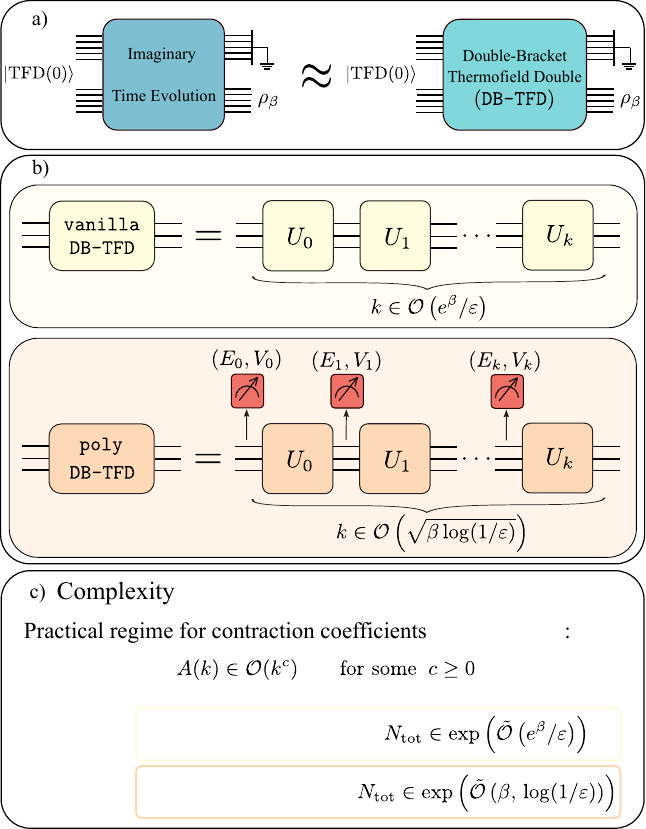}};
    \node[text width=2cm] at (0.2, 1.47)
    {Lemma \ref{theorem:naive_error}:};
    \node[text width=2cm] at (-0.3, -1.57)
    {Lemma \ref{theorem:poly_error}:};
    \node[text width=2cm] at (2.47, -2.94 )
    {(Definition \ref{def:contraction-coeff})};
    \node[text width=2.5cm] at (-1.2, -4.32)
    {\texttt{vanilla}};
    \node[text width=2cm] at (0.1, -4.32)
    {\texttt{DB-TFD}};
    \node[text width=2.5cm] at (-0.9, -5.1)
    {\texttt{poly}};
    \node[text width=2cm] at (0.1, -5.1)
    {\texttt{DB-TFD}};
    \node[text width=2cm] at (-3.1, -4.7)
    {Theorem \ref{prop:PF-better-error}:};

  \end{tikzpicture}
    \caption{\textbf{ Overview of Double-bracket Thermofield Double (DB-TFD) approaches for thermal state preparation.} (a) The reduced state of imaginary-time evolution applied to a thermofield double state gives the target thermal state. To implement imaginary-time evolution, we employ the double-bracket quantum algorithm framework. (b) We introduce two variants: the \texttt{vanilla} \texttt{DB-TFD} approach and the \texttt{poly} \texttt{DB-TFD} approach. The number of iterations required by the \texttt{vanilla} \texttt{DB-TFD} method scales as $\mathcal{O}(e^\beta/\varepsilon)$, as established in Lemma~\ref{theorem:naive_error}. In contrast, Lemma~\ref{theorem:poly_error} shows an exponential improvement of the number of iterations $\mathcal{O}(\sqrt{\beta\log(1/\varepsilon)})$, at the cost of needing to estimate the energy and variance at each iteration. 
    (c) Theorem~\ref{prop:PF-better-error} demonstrates that \texttt{poly} \texttt{DB-TFD} achieves a query complexity exponential in $\beta$, $N_{\rm{tot}} \in \exp (\tilde{\mathcal{O}}(\beta))$ (inline with prior state-of-the-art~\cite{Gily_n_2019, chowdhury2016quantumalgorithmsgibbssampling, Poulin_2009}), for regimes where the contraction coefficient $A(k)$ (Definition~\ref{def:contraction-coeff}) scales sub-polynomially with $k$.
    Here, $f = \tilde{\mathcal{O}}(g)$ indicates that $f = \mathcal{O}\left( g \, \mathrm{poly log}(g)\right)$.
    Numerical simulations support these scaling behaviors for several physical models.}
    \label{fig:mainFig}
\end{figure}

In this work, we propose an algorithm for thermal state preparation based on the simulation of thermofield double states~\cite{Cottrell_2019}. 
The thermofield double state is a purification of a thermal state defined on a doubled Hilbert space, such that tracing out one subsystem yields the desired Gibbs state.
Importantly, applying imaginary-time evolution to a thermofield double state can prepare thermofield double states at different temperatures.
We thus employ the double-bracket quantum algorithm framework, which was originally introduced for synthesizing quantum circuits for diagonalization tasks~\cite{Gluza_2024} and has since been extended to imaginary-time evolution~\cite{gluza2024double} and general polynomial transformations of Hamiltonians~\cite{suzuki2025doublebracketalgorithmquantumsignal}.
We call our proposed method the double-bracket thermofield double (\texttt{DB-TFD}) algorithm and introduce two variants: the \hiddenlink{vanilla-DBTFD}{\texttt{vanilla} \texttt{DB-TFD}} and the \hiddenlink{poly-DBTFD}{\texttt{poly} \texttt{DB-TFD}}.
The former directly employs the double-bracket quantum imaginary-time evolution (\texttt{DB-QITE}) framework~\cite{gluza2024double} to implement imaginary-time evolution, whereas the latter leverages double-bracket quantum signal processing (\texttt{DB-QSP})~\cite{suzuki2025doublebracketalgorithmquantumsignal} to realize a polynomial approximation of the exponential function appearing in imaginary-time evolution. Fig.~\ref{fig:mainFig} provides a summary.

We show that our \hiddenlink{poly-DBTFD}{\texttt{poly} \texttt{DB-TFD}} can achieve exponential query complexity with respect to the inverse temperature $\beta$ within a broad practical regime characterized through a contraction coefficient introduced in Definition~\ref{def:contraction-coeff}.
This scaling is comparable to that of existing methods~\cite{ Gily_n_2019, chowdhury2016quantumalgorithmsgibbssampling}.
Numerical simulations further support both the validity of this regime and the predicted query complexity. The agreement between the analytical bounds in this regime and empirical performance further strengthens the practical relevance of the proposed framework.

Foreseeing a range of applications, including optimization and machine learning, we focus here on applying \texttt{DB-TFD} to generative modeling tasks using a quantum Boltzmann machine~\cite{Amin_2018, Coopmans_2024}, where the Gibbs state is prepared using our approach.
In the task, \texttt{DB-TFD} demonstrates improved performance compared with variational imaginary-time evolution algorithms~\cite{Zoufal_2021}. 
By avoiding variational optimization and reducing reliance on deep fault-tolerant primitives, these results highlight \texttt{DB-TFD} as a feasible approach for thermal state preparation on late near-term and early fault-tolerant devices.

This paper is organized as follows. 
Sec.~\ref{sec:DBTFD} introduces thermofield double states and the double-bracket quantum algorithms framework, followed by our proposal, \texttt{DB-TFD}.
In Sec.~\ref{sec:analytics}, we present theoretical and numerical analysis on its implementation costs.
We then evaluate the performance of \texttt{DB-TFD} in generative modeling tasks in Sec.~\ref{sec:QBM-numerics}.
Finally, Sec.~\ref{sec:conclusion} summarizes the main findings and discusses open questions as well as directions for future work.

\section{Double-bracket thermofield double (DB-TFD) approach}
\label{sec:DBTFD}

The main idea of our approach is to exploit the purification of thermal states, known as thermofield double states~\cite{Cottrell_2019}, in combination with double-bracket quantum algorithms~\cite{Gluza_2024} for implementing imaginary-time evolution~\cite{gluza2024double} and quantum signal processing~\cite{suzuki2025doublebracketalgorithmquantumsignal}. 
We therefore provide a brief overview of these techniques before introducing our method for thermal state preparation, termed the double-bracket thermofield double (\texttt{DB-TFD}).

\subsection{Thermofield double states for quantum thermal simulation}
The goal of thermal state preparation is to generate the Gibbs state  $\rho(\beta) = e^{-\beta H_{s}} / Z(\beta)$ associated with an $n$-qubit system Hamiltonian $H_{s} \in\mathbb{C}^{D\times D}$ acting on the system Hilbert space $\mathcal{H}_{sys}$, where $D=2^{n}$.
Here, $\beta=1/T$ denotes the inverse temperature and $Z(\beta) = {\rm Tr}[e^{-\beta H_{s}}]$ is the partition function of the system Hamiltonian.

Recently, substantial effort has been devoted to this task using a variety of quantum techniques tailored to different computational settings: variational approaches have been proposed with an emphasis on near-term quantum devices~\cite{wu2019variational, shtanko2023preparingthermalstatesnoiseless}, whereas Lindbladian simulation~\cite{chen2023quantumthermalstatepreparation, ding2025endtoendefficientquantumthermal} and imaginary-time evolution-based frameworks~\cite{chowdhury2016quantumalgorithmsgibbssampling, Gily_n_2019, Motta_2019}, have been developed for fault-tolerant quantum computers.
A detailed overview of related work is provided in App.~\ref{app:related_work}.

One strategy for thermal state preparation is to trace out one subsystem of the thermofield double states defined as
\begin{equation}
    \ket{\rm{TFD}(\beta)} \coloneqq \frac{1}{\sqrt{Z(\beta)}} \sum_{i} e^{-\beta E_{i}/2}\ket{i}_{\rm{sys}}\ket{\Tilde{i}}_{\rm{a}},
\end{equation}
where  $E_{i}$ and $\ket{i}_{\rm{sys}}$ are the eigenvalues and eigenvectors of $H_{s}$, respectively.
The auxiliary states $\ket{\Tilde{i}}_{\rm{a}}$ form an orthonormal basis of the auxiliary Hilbert space $\mathcal{H}_{a}$ and the notation $\tilde{\cdot}$ is used to describe the unitary freedom inherent in purification.
A key property of the thermofield double state is that tracing out the auxiliary subsystem yields the thermal state,
\begin{equation}
    \mathrm{Tr}_{\rm{a}}\left[\ket{\rm{TFD}(\beta)} \bra{\rm{TFD}(\beta)}  \right] = \frac{e^{-\beta H_{s}}}{Z(\beta)} =:\rho(\beta).
\end{equation}
This property provides a direct route to thermal state preparation via purification, where the problem reduces to preparing an appropriate entangled state in an enlarged Hilbert space.
Thermofield double states are also of independent interest in quantum information theory~\cite{Swingle_2016}, quantum gravity~\cite{Maldacena_2003,Maldacena_2013,balasubramanian2014multiboundary}, and holography~\cite{maldacena2017diving}.

The thermofield double state can be equivalently expressed as
\begin{equation} \label{eq:tfd_ite_form}
    \ket{\rm{TFD}(\beta)} = \frac{1}{\sqrt{\mathcal{N}}}e^{-\beta \left(H_{s} \otimes \mathds{1}\right)/2} \ket{\rm{TFD}(0)},
\end{equation}
where $\mathcal{N} = \|e^{-\beta \left(H_{s} \otimes \mathds{1}\right)/2} \ket{\rm{TFD}(0)}\|^2 = Z(\beta)/Z(0) = Z(\beta)/D$ is the normalization factor and $\ket{\rm{TFD}(0)}$ is the maximally entangled state between $\mathcal{H}_{\rm{sys}}$ and $\mathcal{H}_{\rm{a}}$.
Eq.~\eqref{eq:tfd_ite_form} thus indicates that preparation of the target thermal state for $H_{s}$ at inverse temperature $\beta$ can be achieved by applying $e^{-\beta \left(H_{s} \otimes \mathds{1}\right)/2}$ with normalization to the maximally entangled state, followed by the partial trace over the auxiliary register.
In other words, implementing imaginary-time evolution for time duration $\beta/2$ directly leads to the generation of the target thermal state.

\subsection{Overview of double-bracket quantum algorithms}

Double-bracket quantum algorithms provide a framework for synthesizing quantum circuits that implement the so-called double-bracket flow~\cite{Helmke}. 
The double-bracket flow is a continuous differential equation defined on symmetric matrices and has been applied to a wide range of computational tasks, including matrix diagonalization, QR decomposition, eigenvalue sorting, and related problems~\cite{BROCKETT199179}.

The first quantum algorithmic realization of the double-bracket flow was proposed in Ref.~\cite{Gluza_2024}, which introduced a quantum approach to matrix diagonalization by implementing the double-bracket flow on quantum computers.
This approach has since been extended to enable two quantum algorithms: double-bracket quantum imaginary-time evolution (\texttt{DB-QITE})~\cite{gluza2024double} and double-bracket algorithm for quantum signal processing (\texttt{DB-QSP})~\cite{suzuki2025doublebracketalgorithmquantumsignal}.

\medskip
\paragraph*{DB-QITE framework --}
\texttt{DB-QITE} is a method to implement imaginary-time evolution for a Hamiltonian $H$ defined as 
\begin{equation} \label{eq:ite_def}
    |\Phi(\tau)\rangle = \frac{e^{-\tau H} |\Phi(0)\rangle}{\|e^{-\tau H} |\Phi(0)\rangle\|},
\end{equation}
where $\tau$ denotes the imaginary time and $\|e^{-\tau H} |\Phi(0)\rangle\|=\sqrt{\braket{\Phi(0)|e^{-2\tau H} |\Psi(0)}}$. 
Taking the derivative of the density matrix representation $\Phi(\tau)=\ket{\Phi(\tau)}\bra{\Phi(\tau)}$ in Eq.~\eqref{eq:ite_def} with respect to $\tau$ yields
\begin{equation} \label{eq:dbf_ite_form}
    \frac{\partial \Phi(\tau)}{\partial \tau} = \left[[\Phi(\tau),H\right], \Phi(\tau)],
\end{equation}
which is an exact form of double-bracket flow.
This suggests that imaginary-time evolution is a solution of the double-bracket flow.

To implement imaginary-time evolution from the viewpoint of double-bracket flow, the double-bracket quantum algorithm framework proceeds in two steps. First, discretizing the double-bracket flow in  Eq.~\eqref{eq:dbf_ite_form} yields 
\begin{equation}
    |\Psi{(s)}\rangle =e^{s[\Psi,H]}|\Psi \rangle,
    \label{eq:dbqite_iter}
\end{equation}
where $\Psi = \ket{\Psi}\bra{\Psi}$ and $s$ denotes a step size.
Eq.~\eqref{eq:dbqite_iter} follows from the exponentiation of $[\Psi,H]$, by regarding the commutator, which is anti-Hermitian, as a generator.

Second, the exponential of the commutator is approximated using product formulas~\cite{Childs_2013}. 
In general, an $m$-order product formula takes the form
\begin{equation}
    f_m{(t)} := e^{i\alpha_1tH}e^{i\alpha_2t\Psi}e^{i\alpha_3tH}\dots e^{i\alpha_p t \Psi} = e^{t^2[\Psi,H]}+\mathcal{O}(t^{m+1}),
    \label{eq:product-formula_}
\end{equation}
where the real coefficients, $\{\alpha_i\}_{i=1}^{p}$, and the number of terms depend on the specific scheme. 
A simple example is the group commutator,
\begin{equation}
    e^{s[\Psi,H]} = e^{i\sqrt{s}H}e^{i\sqrt{s}\Psi}e^{-i\sqrt{s}H}e^{-i\sqrt{s}\Psi} + \mathcal{O}(s^{3/2}),
    \label{eq:GCI}
\end{equation}
which provides a recursive approach to approximate imaginary-time evolution.
In particular, this leads to the iteration
\begin{equation} \label{eq:db_qite_orig}
    U_{k+1} = e^{i\sqrt{s_k}H}U_{k}e^{i\sqrt{s_k}\Psi_0} U_{k}^{\dagger} e^{-i\sqrt{s_k}H}U_{k},
\end{equation}
where $U_{k}\ket{\Psi_0} = \ket{\Psi_k}$ with the initial state $\ket{\Psi_0}$.
The recursive scheme defined in Eq.~\eqref{eq:db_qite_orig} constitutes the \texttt{DB-QITE}.
See App.~\ref{App:DB-QITE} for more details.

\medskip
\paragraph*{DB-QSP framework --}
\texttt{DB-QSP} enables the implementation of any matrix-valued polynomial transformations $p(H)$ without ancilla-qubits or post-selection. 
Such polynomial transformations are central to the systematic design of quantum algorithms. 
For example, by setting $p(H) \approx e^{- \tau H}$, using existing low-degree approximations~\cite{sachdeva2013approximationtheorydesignfast}, we can realize imaginary-time evolution.

The key observation is that a linear polynomial with a real root, $H-\tau \mathds{1}$ with $\tau \in \mathbb{R}$, can be realized through a first-order approximation of the double-bracket flow in Eq.~\eqref{eq:dbqite_iter}.
Concretely, 
\begin{equation}
    e^{s_{\tau}[\Psi,H]}|\Psi\rangle = \frac{\left(H-\tau \mathds{1}\right)|\Psi\rangle}{\|(H-\tau \mathds{1})|\Psi\rangle\|},
    \label{eq:poly_form_DBI}
\end{equation}
where the step size $s_\tau=  -\frac{1}{\sqrt{V_{\Psi}}} \arccos \left( \frac{E_{\Psi}-\tau}{\sqrt{V_{\Psi}+(E_{\Psi}-\tau)^2}} \right)$ is determined by the energy $E_{\Psi}=\braket{\Psi|H|\Psi}$ and the variance $V_{\Psi}=\braket{\Psi|(H-E_{\Psi})^2|\Psi}$ of the Hamiltonian with respect to the state~$\ket{\Psi}$.
This construction extends to the case for a complex root, $H-z\mathds{1}$ with $z\in \mathbb{C}$, by incorporating reflections about the state $\ket{\Psi}$, i.e.,  $e^{i\theta_z\Psi}$ with $\theta_z\in\mathbb{R}$: further details are provided in App.~\ref{app:Double-bracket quantum signal processing (DB-QSP)}.

By utilizing the identity in Eq.~\eqref{eq:poly_form_DBI}, a degree-$K$ polynomial $p(H)=c \prod_{k=1}^{K} (H-z_{k}\mathds{1})$ for $c,z_{k} \in \mathbb{C}$ applied to an initial state $\ket{\Psi_{0}}$ after normalization results in
\begin{equation}
    \frac{p(H)\ket{\Psi_{0}}}{\|p(H)\ket{\Psi_{0}}\|} = \prod_{k=0}^{K} e^{i\theta_k \Psi_{k}}e^{is_{k} [\Psi_{k}, H]} \ket{\Psi_{0}},
    \label{eq:DB-QSP-implementation}
\end{equation}
where $\ket{\Psi_{k'}}= \prod_{k=0}^{k'} e^{i\theta_k \Psi_{k}}e^{is_{k} [\Psi_{k}, H]} \ket{\Psi_{0}}$. The exponential of the commutator in Eq.~\eqref{eq:DB-QSP-implementation} is approximated using product formulas.
This procedure constitutes the \texttt{DB-QSP} implementation of matrix-valued polynomial transformations: see more details in App.~\ref{app:Double-bracket quantum signal processing (DB-QSP)}.
We note that to find the time step to implement the polynomial, the energy and variance must be computed at each iteration. 
Due to shot noise, only the estimations of these values can be obtained, inducing an additional statistical error.

\medskip
\medskip

An advantage of both \texttt{DB-QITE} and \texttt{DB-QSP} is that their implementations do not require ancilla qubits or post-selection, in contrast to many existing methods that rely on these resources~\cite{Gily_n_2019, chen2023quantumthermalstatepreparation, chowdhury2016quantumalgorithmsgibbssampling}.
This feature makes them potentially more suitable for the implementations under realistic device constraints such as limited qubit connectivity and restricted budgets for two-qubit gates~\cite{AbuGhanem_2025}, compared with fault-tolerant quantum algorithms. 

On the other hand, this advantage comes at the cost of circuit depth. 
The commutator appearing in the exponent depends explicitly on the evolving quantum state, which necessitates the implementation of state-dependent operations.
As a result, the recursive structure of the algorithms leads to exponential growth in circuit depth with the number of steps~$k$.
When the group commutator construction is used for \texttt{DB-QITE}, the circuit depth scales as $\mathcal{O}(3^k)$ for $k$ steps. 
This worst case hardness is consistent with the QMA-hardness of ground state preperation~\cite{kempe2005complexitylocalhamiltonianproblem, Barahona_1982}.

\texttt{DB-QITE} and \texttt{DB-QSP} are therefore only effective for short-time imaginary-time evolution or low-degree polynomial transformations, respectively. Thus we expect these algorithms to find use only when a small number of iterations suffices to reach the target, for example, as a warm start or in situations where the initial state is already close to the ground state. 
Indeed, in such cases, the absence of block encodings in these methods can reduce the required number of two-qubit gates relative to alternative algorithms~\cite{gluza2024double}.

\subsection{Our proposal: DB-TFD} \label{subsec:Our proposal}
Building on techniques outlined above, we propose a double-bracket quantum algorithm approach for thermal state preparation that applies imaginary-time evolution to thermofield double states: we refer to the method as \texttt{DB-TFD}.
Specifically, we introduce two variants: \hiddenlink{naive-DBTFD}{\texttt{vanilla} \texttt{DB-TFD}} and \hiddenlink{poly-DBTFD}{\texttt{poly} \texttt{DB-TFD}}.

\begin{description}[leftmargin=0cm]
    \item[\hypertarget{naive-DBTFD}{Vanilla DB-TFD}] This approach directly exploits the first-order approximation of the double-bracket flow in Eq.~\eqref{eq:dbqite_iter}.
    Define 
    \begin{equation}
    \begin{split} \label{eq:naive_DB_TFD}
        &|\psi_{k}^{\texttt{(vanilla)}} \rangle := \prod_{i=0}^{k-1} U_{\rm{DB}, i}^{\texttt{(vanilla)}}  |\rm{TFD}(0)\rangle \\
        & U_{\rm{DB}, i}^{\texttt{(vanilla)}} := e^{\frac{\beta}{2k} [\psi_i^{\texttt{(vanilla)}}, H]} \,  ,   
    \end{split} 
    \end{equation}
    with $\psi_0^{\texttt{(vanilla)}} = |\rm{TFD}(0)\rangle \langle\rm{TFD(0)}|$ and $k\ge 1$.
    Since taking the limit $k\to \infty$ yields exact imaginary-time evolution, i.e., $\lim_{k\to \infty} |\psi_{k}^{\texttt{(vanilla)}} \rangle = |\rm{TFD}(\beta)\rangle$ (as shown in App.~\ref{Naive_DBI_perfomance}), the \texttt{DB-QITE} framework is straightforwardly employed to implement this procedure. 
    Moreover, the product-formula approximation of the unitary in Eq.~\eqref{eq:naive_DB_TFD}, which is used for actual implementation, is denoted by $U_{\rm{PF}, k}^{\texttt{(vanilla)}}$. 
    This method does not require any knowledge of intermediate-state properties once $\beta$ and the number of iterations $k$ are fixed.
    
    \item[\hypertarget{poly-DBTFD}{Poly DB-TFD}] 
    This approach exploits the \texttt{DB-QSP} framework to approximate the exponential function of imaginary-time evolution using a degree-$k$ polynomial.
    Consider 
    \begin{equation}
    \begin{split} \label{eq:poly_DB_TFD}
        |\psi_k^{\texttt{(poly)}}\rangle &:=\frac{ \prod_{i=0}^{k-1}\left(H-z_i \mathds{1})  \right)|\rm{TFD}(0)\rangle }{\| \prod_{i=0}^{k-1}\left(H-z_i \mathds{1} )  \right)|\rm{TFD}(0)\rangle \|} \\
        &= \prod_{i=0}^{k-1}U_{\rm{DB}, i}^{\texttt{(poly)}}|\rm{TFD}(0)\rangle \,, \\
        U_{\rm{DB}, i}^{\texttt{(poly)}} &:=e^{i\theta_i \psi_i^{\texttt{(poly)}}}e^{s_i [\psi_i^{\texttt{(poly)}},H]} \, ,
    \end{split} 
    \end{equation}
    with $\psi_0^{\texttt{(poly)}} = |\rm{TFD}(0)\rangle \langle\rm{TFD(0)}|$.
    Choosing the angles $\{(\theta_{i},s_{i})\}_i$ is equivalent to selecting a polynomial root $z_i\in\mathbb{C}$ and is therefore determined by the target polynomial to be implemented.
    Moreover, since the angles depend on the energy and variance, this method (in contrast to \hiddenlink{naive-DBTFD}{\texttt{vanilla} \texttt{DB-TFD}}) requires additional estimations of these quantities at each step.
    
    For the polynomial approximations of the exponential function, low-degree approximations have already been developed in, e.g., Refs.~\cite{sachdeva2013approximationtheorydesignfast}; see App.~\ref{exponential_poly_approx} for the details. These results show that the exponential function on the interval $[-1,1]$ can be approximated with precision $\varepsilon$  by a polynomial of degree $\mathcal{O}\left (\sqrt{\beta\log\left({1}/{\varepsilon}\right)}\right )$, yielding sublinear scaling in the inverse temperature $\beta$. This construction relies on the Taylor series of the exponential, together with approximation of each monomial via Chebyshev polynomials.
    Consequently, this is known to be optimal with respect to $\beta$.
    
    Similar to \hiddenlink{naive-DBTFD}{\texttt{vanilla} \texttt{DB-TFD}}, the product-formula approximation of the unitary in Eq.~\eqref{eq:poly_DB_TFD} is denoted by $U_{\rm{PF}, k}^{\texttt{(poly)}}$.
\end{description}

\section{Scaling analysis}
\label{sec:analytics}

We aim to analyze the implementation costs of preparing the thermofield double state $\ket{\rm{TFD}(\beta)}$ at inverse temperature~$\beta$. 
Specifically, we quantify the number of queries to the Hamiltonian evolution $\exp({itH})$ and the reflection operator with respect to the initial state $\exp(i\theta \psi_{0})$ within a target precision~$\varepsilon$.
To estimate the implementation costs, we divide the total preparation error into three sources: the intrinsic algorithmic errors, the product-formula errors, and the statistical shot noise errors.

More concretely, in Sec.~\ref{subsec:perfect-implementation}, we start by analyzing the intrinsic algorithmic errors of \hiddenlink{naive-DBTFD}{\texttt{vanilla}} and \hiddenlink{poly-DBTFD}{\texttt{poly} \texttt{DB-TFD}} states, which are associated with $|\psi_{k}^{\rm{(\texttt{vanilla})}} \rangle$ and  $|\psi_{k}^{{(\texttt{poly})}} \rangle$ in Eqs.~\eqref{eq:naive_DB_TFD} and~\eqref{eq:poly_DB_TFD}. These states assume a perfect implementation of the commutator exponentials $\exp(s_k[\psi_k^{(\rm{v})},H])$ with $\rm{v}=\{\texttt{vanilla},\,\texttt{poly}\}$, and allow us to determine the resources (i.e., the number of recursive steps for the \hiddenlink{naive-DBTFD}{\texttt{vanilla}} approach and the polynomial degree for the \texttt{poly} approach) to achieve an error $\varepsilon_{\rm{DB}}$.

Second, in Sec.~\ref{subsec:realistic-implementation}, we evaluate the states, $|\phi_{k}^{\rm{(v)}} \rangle$ with $\rm{v}=\{\texttt{vanilla},\,\texttt{poly}\}$, which are prepared when the exponential of the commutators are approximated via product formulas with unitaries $U_{\rm{PF}, k}^{(\rm{v})}$.
This approximation introduces an additional error $\varepsilon_{\rm{PF}}$.

Third, we note that the \hiddenlink{poly-DBTFD}{\texttt{poly} \texttt{DB-TFD}} requires estimating both the energy and the variance to determine the time step, whereas the \hiddenlink{naive-DBTFD}{\texttt{vanilla} \texttt{DB-TFD}} does not.
We therefore analyze the realistically prepared state using the estimated values of the energy and variance, which we denote $|\varphi_{k}^{(\rm{v})} \rangle $ with $\rm{v}=\{\texttt{vanilla},\,\texttt{poly}\}$, and estimate the query complexity to prepare the state within a tolerance of $\varepsilon_{\rm{noise}}$. For the \texttt{vanilla} \texttt{DB-TFD}, we have $|\varphi_{k}^{\rm{(\texttt{vanilla})}} \rangle = |\phi_{k}^{\texttt{(vanilla)}} \rangle$.

Owing to the three sources of error, the total error in preparing the target thermofield double state, $| \rm{TFD}(\beta) \rangle$, is bounded by the triangle inequality as follows:
\begin{align}
    \|  | \varphi_k^{(\rm{v})} \rangle - | \rm{TFD}(\beta) \rangle \| & \leq   \varepsilon_{\rm{DB}} +\varepsilon_{\rm{PF}}   + \varepsilon_{\rm{noise}}
     \leq \varepsilon,
\end{align}
where we introduce $\| | \psi_k^{(\rm{v})} \rangle - | \rm{TFD}(\beta) \rangle \| \leq \varepsilon_{\rm{DB}}$, $\| | \phi_k^{(\rm{v})} \rangle - | \psi_k^{(\rm{v})} \rangle\| \leq \varepsilon_{\rm{PF}}$, and $\| | \varphi_k^{(\rm{v})} \rangle - | \phi_k^{(\rm{v})} \rangle\| \leq \varepsilon_{\rm{noise}} $. 
We therefore aim to bound each term such that the sum is bounded by~$\varepsilon$.
In the analysis hereafter, we assume without loss of generality that the Hamiltonian is normalized, i.e., $\|H\|_{\text{op}} \le 1$, where $\|\cdot  \|_{\rm{op}}$ denotes the operator norm. 

\subsection{Intrinsic Algorithmic Errors: Lower bounds on the number of recursion steps and polynomial degrees}
\label{subsec:perfect-implementation}
We begin by considering thermal state preparation using the states $|\psi_{k}^{\rm{(v)}} \rangle$ defined in Eqs.~\eqref{eq:naive_DB_TFD} and~\eqref{eq:poly_DB_TFD}, where $\rm{v} \in \{\texttt{vanilla},\,\texttt{poly}\}$. In this context, we formally define $k$ as the number of recursive steps for the \hiddenlink{naive-DBTFD}{\texttt{vanilla}} approach and the polynomial degree for the \hiddenlink{poly-DBTFD}{\texttt{poly}} approach.
The goal here is to approximate the thermofield double state $\ket{\rm{TFD}(\beta)}$ at inverse temperature $\beta$ within a specified error tolerance, 
\begin{equation}
    \| |\psi_{k}^{\rm{(v)}}\rangle - |\rm{TFD}(\beta)\rangle \| \leq \varepsilon_{\rm{DB}} . 
\end{equation}

We first analyze the number of recursive steps required for the \hiddenlink{naive-DBTFD}{\texttt{vanilla} \texttt{DB-TFD}} scheme. 
The detailed derivation is deferred to App.~\ref{Naive_DBI_perfomance} and follows a procedure similar to that in Ref.~\cite{mcmahon2025equatingquantumimaginarytime}.

\begin{lemma}[Recursion steps required for \hiddenlink{naive-DBTFD}{\texttt{vanilla} \texttt{DB-TFD}}] \label{theorem:naive_error}
To achieve an $\varepsilon_{\rm{DB}}$-precision approximation of the thermofield double state at inverse temperature $\beta$ using the \hiddenlink{naive-DBTFD}{\texttt{vanilla} \texttt{DB-TFD}}, it suffices to choose the number of iterations $k$, such that
\begin{equation} \label{eq:k_naive}
    k \geq \frac{2\beta\|H\|_{\rm{op}} e^{\beta \|H\|_{\rm{op}}}}{3\varepsilon_{\rm{DB}}}.
\end{equation}
\end{lemma}

A sketch of the proof is as follows.
The bound in Lemma~\ref{theorem:naive_error} follows from analyzing the truncation error incurred by the first-order discretization in the \hiddenlink{naive-DBTFD}{\texttt{vanilla} \texttt{DB-TFD}} construction.
This scheme can be viewed as approximating imaginary-time evolution through repeated applications of an infinitesimal generator, in close analogy with the standard approximation of an exponential operator via
\begin{equation}  \label{eq:exp_approximation}
    e^X = \lim_{k  \rightarrow\infty} \left(\mathds{1}+\frac{X}{k}\right)^k.
\end{equation}
Each recursion step in \hiddenlink{naive-DBTFD}{\texttt{vanilla} \texttt{DB-TFD}} corresponds to a finite step in this discretized evolution with step size $\beta / (2k)$. 
By bounding the accumulated truncation error of this approximation and requiring it to be no larger than $\varepsilon_{\rm{DB}}$, the stated lower bound on the number of required recursion steps can be obtained.

Lemma~\ref{theorem:naive_error} implies that the number of recursion steps $k$ for the \hiddenlink{naive-DBTFD}{\texttt{vanilla} \texttt{DB-TFD}} approach exhibits an exponential dependence on the inverse temperature. 
Specifically, the required recursion steps scale as $\mathcal{O}\left({e^{\beta \|H\|_{\rm{op}}}}/{\varepsilon_{\rm{DB}}}\right)$.
This indicates that the approach does not yield an efficient method for preparing $\varepsilon_{\rm{DB}}$-approximate thermofield double states at low temperatures, even under the idealized assumption that the exponential of the commutator can be implemented perfectly and at negligible cost. 

It is intriguing to contrast this behavior with \texttt{DB-QITE}, which aims to prepare the state corresponding to the zero-temperature limit $\beta \to \infty$ and may not exhibit exponential scaling in the recursion depth.
The difference can be attributed to the distinct objectives of the two methods.
For ground-state preparation, as long as the energy is minimized (or equivalently, the fidelity to the ground state is improved), how well the target operation is realized does not matter. 
However, preparing a thermal state requires precise implementation of the imaginary-time evolution operator. Since the \hiddenlink{naive-DBTFD}{\texttt{vanilla} \texttt{DB-TFD}} relies on a first-order discretization, the truncation errors can scale exponentially with $\beta$. \\

Next, we consider the \hiddenlink{poly-DBTFD}{\texttt{poly} \texttt{DB-TFD}} scheme, which enables tailored polynomial constructions for approximating general functions.
The large number of recursion steps required by \hiddenlink{naive-DBTFD}{\texttt{vanilla} \texttt{DB-TFD}} arises from the inefficiency of the approximation of the exponential in Eq.~\eqref{eq:exp_approximation}. 
In contrast, as mentioned in Sec.~\ref{subsec:Our proposal}, low-degree polynomial approximations of the exponential is known to be optimal with respect to the inverse temperature, in the sense that any polynomial approximation of the exponential requires a degree at least $\mathcal{O}\left(\sqrt{\beta}\right)$~\cite{sachdeva2013approximationtheorydesignfast}.
Since the \texttt{DB-QSP} framework allows for flexible polynomial constructions, we leverage these results to estimate the polynomial degree for \hiddenlink{poly-DBTFD}{\texttt{poly} \texttt{DB-TFD}}.
The detailed proof is provided in App.~\ref{QSP_perfomance}.

\begin{lemma}[Polynomial degrees required for \hiddenlink{poly-DBTFD}{\texttt{poly} \texttt{DB-TFD}}] \label{theorem:poly_error}
To achieve an $\varepsilon_{\rm{DB}}$-precision approximation of the thermofield double state at inverse temperature $\beta$ using the \hiddenlink{poly-DBTFD}{\texttt{poly} \texttt{DB-TFD}}, it suffices to choose the polynomial degree $k$, such that
    \begin{equation}
        k \geq e\sqrt{\frac{\beta}{2}\log\left(\frac{8}{\varepsilon_{\rm{DB}}} \sqrt{\frac{D}{Z(\beta)}}\right)},
    \end{equation}
    where $D$ denotes the system dimension and $Z(\beta)$ is the partition function associated with the shifted Hamiltonian $H + \mathds{1}$. 

\end{lemma}

Note that the scaling of polynomial degree in Lemma~\ref{theorem:poly_error} can be further simplified in standard settings. 
For a normalized Hamiltonian, the partition function admits the lower bound $Z(\beta)\ge e^{-\beta}$.
Incorporating this consideration yields the required polynomial degree scaling as $\mathcal{O}(\sqrt{\beta^2 + \beta \log(1/\varepsilon_{\rm{DB}})})$.
Lemma~\ref{theorem:poly_error}, together with Lemma~\ref{theorem:naive_error}, indicates that the number of recursive steps required by \hiddenlink{poly-DBTFD}{\texttt{poly} \texttt{DB-TFD}} is exponentially smaller than that of \hiddenlink{naive-DBTFD}{\texttt{vanilla} \texttt{DB-TFD}}.

In practice, the exponential of the commutators must be approximated using product formulas, which introduces additional errors depending on the order of the chosen product formula. 
Moreover, in the \hiddenlink{poly-DBTFD}{\texttt{poly} \texttt{DB-TFD}} scheme, accurate estimation of the energy and variance is required to determine the step sizes associated with the polynomial roots. 
This requirement not only increases the resource overhead but also introduces additional sources of error due to the finite number of measurement shots.
In the following section, we turn to the analysis of these actual implementation costs.

\subsection{Resource counts for the full algorithm}
\label{subsec:realistic-implementation}
Lemma~\ref{theorem:naive_error} and Lemma~\ref{theorem:poly_error} do not yet account for the depth required for each iteration, nor the algorithmic error that arises from product formula approximation of $\exp(s_k[\psi_k^{(\rm{v})},H])$.
In the \hiddenlink{poly-DBTFD}{\texttt{poly} \texttt{DB-TFD}} scheme, additional errors are also incurred due to the finite precision in estimating the energy and variance, which stems from a limited number of measurement shots.
In the following, we analytically quantify these errors and assess their impact on the overall accuracy of the thermal state preparation in realistic scenarios.

\medskip
\paragraph*{Product formula error --}
We first analyze the errors arising from the use of an $m$-th order product formula, $f_m(\sqrt{s})$, instead of the exact exponential of the commutator.
By definition, an $m$-th order product formula can be expressed as
\begin{equation}
\begin{split}
    f_m(\sqrt{s}) &:=e^{i\alpha_1\sqrt{s}H}e^{i\alpha_2\sqrt{s}\psi}\ldots e^{i\alpha_{N_{m}}\sqrt{s} \psi} \\
    &= e^{s[\psi,H]}+Cs^{\frac{m+1}{2}},
\end{split}
\label{eq:product-formula}
\end{equation}
where the constant $C$, the number of factors $N_m$, and the phases $\{\alpha_{i}\}_{i=1}^{N_{m}}$ depend on the specific product formula employed.
Explicit constructions up to sixth order have been presented~\cite{Casas_2025}, and recursive constructions have also been developed~\cite{Chen_2022}. 
Theoretical results suggest that an $m$-th order product formula requires $N_m \in \Omega(2^m/m)$ queries to approximate the exponential of the exact commutator~\cite{Casas_2025}, although achieving this scaling in practice is nontrivial due to the substantial overhead associated with recursive constructions of high-order product formulas.
\medskip

Here, we compare the ideal unitaries $U_{\rm{DB}, k}^{\rm{(v)}}$ of Eqs.~\eqref{eq:naive_DB_TFD} and~\eqref{eq:poly_DB_TFD}, which implement exponentials of commutators exactly across all $k$ iterations, with its product-formula approximation $U_{\rm{PF}, k}^{\rm{(v)}}$, where $\rm{v} \in \{\texttt{vanilla},\,\texttt{poly}\}$.
To quantify their difference, we consider the distance between the states generated by those unitaries (from the same initial state), i.e., 
\begin{equation}
\label{eq:Delta_k}
    \Delta_k = \||\psi_k^{(\rm{v})} \rangle - | \phi_k^{(\rm{v})} \rangle \| \leq \varepsilon_{\rm{PF}},
\end{equation}
where $|\psi_k^{(\rm{v})} \rangle = U_{\rm{DB}, k}^{(\rm{v})}|\psi_{0} \rangle$ and $| \phi_k^{(\rm{v})} \rangle = U_{\rm{PF},k}^{(\rm{v})}|\psi_{0}\rangle$.
In particular, the derivation relies on the recursive error analysis where we can write;
\begin{equation} \label{eq:error_accumulation}
    \Delta_{j+1} \leq a_j \Delta_j + b_j
\end{equation}
with step-dependent coefficient $a_j$ and bias term $b_j$. By setting $b=\max_{i} b_{i}$, solving Eq.~\eqref{eq:error_accumulation} for $j=k$ leads to
\begin{equation}\label{eq:contraction-coefficients-with-A}
    \Delta_{k} \leq a_{k-1} a_{k-2} \dots a_2 a_1 \Delta_{0} +  b  \sum_{i = 1}^{k-1} A_i, 
\end{equation}
where $A_i = \prod_{j = i}^{k-1} a_j$ for each $i\in\{1, 2, 3, \dots, k-1\}$.
Eq.~\eqref{eq:contraction-coefficients-with-A} shows that the prefactor $A_{i}$ governs the accumulation of error as $k$ increases.
Concretely, when $a_{j}\leq 1$, i.e., the error contracts at step $k$ in 2-norm, the prefactors $A_{i}$ are correspondingly suppressed, leading to reduced error accumulation.
On the other hand, large values of $A_{i}$ lead to significant increases in the error.
This observation motivates the following definition.

\begin{definition}[Contraction Coefficient] \label{def:contraction-coeff}
   For $k$ steps of the \texttt{DB-TFD} algorithm, we define the contraction coefficient
    \begin{equation}  \label{eq:contraction-coefficients}
        A(k) = \sum_{i = 1}^{k-1} A_i = \sum_{i = 1}^{k-1} \prod_{j = i}^{k-1}  a_j \, ,
    \end{equation}
    where $a_j$ is implicitly defined in Eq.~\eqref{eq:error_accumulation}.
\end{definition}

Using the contraction coefficient, we obtain the upper bound on the error in Eq.~\eqref{eq:Delta_k}.
Since the product-formula errors affect both \hiddenlink{naive-DBTFD}{\texttt{vanilla}} and \hiddenlink{poly-DBTFD}{\texttt{poly} \texttt{DB-TFD}} in the same manner, we can establish a common upper bound on the error introduced by this approximation.
The detailed proof of this bound is provided in App.~\ref{proof_PF_error_general}.

\begin{lemma}[Error incurred by an $m$-th order product formula]
Let $U_{DB,k}^{\rm{(v)}}$ denote the $k$-step \texttt{DB-TFD} using the exact exponential of the commutators, and let $U_{PF, k}^{\rm{(v)}}$ denote the corresponding unitary using an $m$-th order product formula with $N_m$ factors, for $\rm{v}=\{\texttt{vanilla},\,\texttt{poly}\}$.
Both constructions use the same Hamiltonian $H$ and employ an identical sequence of time steps.
Let $s_{(\rm max)}$ denote the largest step size. Then, for any initial state $|\psi_{0}\rangle$, the error after $k$ iterations is bounded by
\begin{equation} \label{eq:pf_bound}
    \Delta_{k} \le C s_{(\rm{max})}^{\frac{m+1}{2}} A(k),
\end{equation}
where $A(k)$ is the contraction coefficient given by Definition~\ref{def:contraction-coeff}, and $C$ is a constant dependent on the choice of product formula.
\label{lemma:PF_error}
\end{lemma}

Lemma \ref{lemma:PF_error} provides a criterion for selecting the order $m$ of the product formula to ensure that $U_{DB,k}^{\rm{(v)}}$ can be approximated with accuracy $\Delta_k \leq \varepsilon_{\rm{PF}}$, in terms of $k$ and $s_{(\rm{max})}$. 
We recall that, $N_m$ operators are applied per iteration in the product-formula approximation, and hence the total number of operator applications over $k$ scales as ${\mathcal{O}}(N_m^k)$.
Thus, Lemma~\ref{lemma:PF_error} shows that achieving an accuracy of $\varepsilon_{\rm{PF}}$ requires $m \in \Omega\left( \log(A(k)/\varepsilon_{\rm{PF}})\right)$; see App.~\ref{app:product_formula_order_general} for the detailed derivation.

For physically motivated Hamiltonians, we find that this contraction coefficient does not scale substantially with system size. 
In Fig.~\ref{fig:A_K_different_sizes}, we plot $A(k)$ for the one dimensional XXZ~\cite{gu2003entanglement}, Transverse-Field Ising Model (TFIM)~\cite{pfeuty1970one} and Heisenberg models, respectively defined as follows:
\begin{align}
        &H_{\rm{XXZ}} = \sum_{k=1}^{n-1}  J (X_kX_{k+1}+Y_kY_{k+1}) +J_z Z_{k}Z_{k+1}+ h\sum_{k=1}^n Z_k, \label{eq:XXZ} \\
        &H_{\rm{TFIM}} = J\sum_{k=1}^{n-1}Z_{k}Z_{k+1} + h\sum_{k=1}^n X_k,
        \label{eq:TFIM} \\
        &H_{\rm{Heis}} = \sum_{k=1}^{n-1}  X_kX_{k+1}+Y_kY_{k+1}+Z_{k}Z_{k+1},
        \label{eq:heisenberg}
\end{align}
where $X_k, Y_k$ and $Z_k$ are Pauli-X, -Y and -Z operators on qubit $k$.
Here, $J$ and $J_z$ denote coupling strengths and $h$ is the constant for the external magnetic field.
The numerical evaluation is performed for a system of $n= 4, 6, 8$ qubits (corresponding to a thermofield double states of $n=8,12,16$) at different $\beta$ ranging from $0.3$ to $1$. We utilize a fourth-order product formula for the approximation of the exponential of commutators. We additionally set all the couplings $J$ and external fields $h$ to 1 for the models.
Note that $k$ is chosen such that \hiddenlink{poly-DBTFD}{\texttt{poly} \texttt{DB-TFD}} achieves an error of $\varepsilon = 0.01$ in preparing the thermal state under the perfect implementation of the exponential  of the commutator.

As illustrated in Fig~\ref{fig:A_K_different_sizes}, the coefficient $A(k)$ is nearly unchanged across different system sizes $D$ for a fixed polynomial order $k$ and increases sublinearly in~$k$. 
This observation suggests that the product-formula error is largely independent of system size, while the inverse temperature remains the dominant factor affecting its magnitude. Such size-independence supports the practical scalability of the approach and indicates its applicability to larger many-body systems. 

We therefore consider the scaling of \texttt{DB-TFD} in the regime where the contraction coefficient scales at most polynomially, i.e.,
\begin{equation}\label{eq:contraction_conj}
    A(k) \in \mathcal{O}(k^{c})
\end{equation}
for some constant $c \geq 0$. This is potentially a rather conservative assumption to make given that Fig.~\ref{fig:A_K_different_sizes} shows much better scaling, i.e., $A_k < k$, indicating that our resulting bounds may be overly pessimistic. 
Under this assumption, the upper bound on the error $\Delta_k$ scales at worst polynomially in $k$.
The detailed proof of the following lemma is provided in App.~\ref{app:proof_PF_error_linear}.

\begin{figure}
    \centering
    \includegraphics[width=\linewidth]{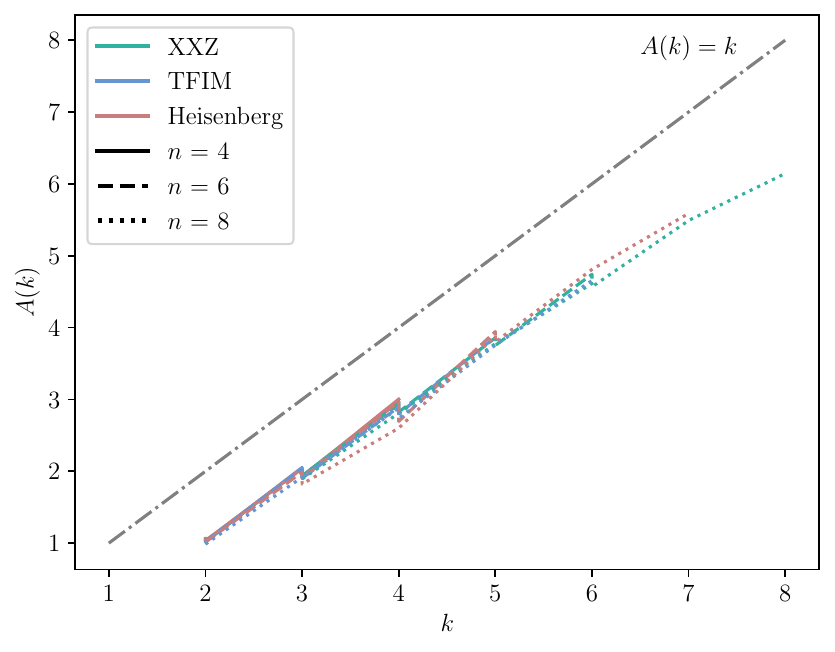}
    \caption{\textbf{Scaling of the contraction coefficients across different systems and system sizes.}
    The contraction coefficient $A(k)$ is evaluated for different systems (XXZ, TFIM, and the Heisenberg model), system sizes $n$, and polynomial order $k$. In particular, we observe $A(k) < k$, suggesting that the growth of the contraction coefficient is slower than the assumption in Eq.~\eqref{eq:contraction_conj}. For a given $k$, the values of $A(k)$ 
    remain similar across the considered systems and system sizes. This observation suggests that the contraction coefficient exhibits only a weak dependence on the specific system and its size.}
    \label{fig:A_K_different_sizes}
\end{figure}

\begin{lemma}[Error incurred by an $m$-th order product formula in practical cases]
    \label{lemma:PF_error_linear}
    Let $U_{DB,k}^{\rm{(v)}}$ denote the $k$-step \texttt{DB-TFD} using the exact exponential of the commutators, and let $U_{PF, k}^{\rm{(v)}}$ denote the corresponding unitary using an $m$-th order product formula, with $N_m$ factors, for $\rm{v}=\{\texttt{vanilla},\,\texttt{poly}\}$.
    Both constructions use the same Hamiltonian $H$ and employ an identical sequence of time steps.
    Let $s_{(\rm max)}$ denote the largest step size. Then, in the regime where $A_i \in \mathcal{O}(k^c)$ holds for the initial state $|\psi_{0}\rangle$, the error after $k$ iterations is bounded by
    \begin{equation} \label{eq:pf_bound_a(k)}
        \||\psi_k^{(\rm{v})} \rangle -|\phi_k^{(\rm{v})} \rangle\| \leq  \mathcal{O} \left( C k^{c} s_{(\rm{max})}^{\frac{m+1}{2}}\right),
    \end{equation}
    where $C$ is a constant depending on the choice of product formula.
\end{lemma}

Using Lemma \ref{lemma:PF_error_linear}, we can show that achieving an accuracy of $\varepsilon_{\rm{PF}}$ to approximate $U_{DB,k}^{\rm{(v)}}$ requires $m \in {\mathcal{O}}(\log(k/\varepsilon_{\rm{PF}}))$; see App.~\ref{app:product_formula_order_efficient} for the detailed derivation, and this allows us to compute the query complexity $N_{\rm tot}$ for \hiddenlink{naive-DBTFD}{\texttt{vanilla}} and \hiddenlink{poly-DBTFD}{\texttt{poly} \texttt{DB-TFD}} respectively. Here, query complexity refers to the number of oracle calls to Hamiltonian evolutions and reflection operators required to prepare a thermal state at inverse temperature $\beta$, subject to a total error bound $\varepsilon_{\rm{PF}} + \varepsilon_{\rm{DB}} \leq \varepsilon$. 
The detailed analysis is provided in App.~\ref{app_number_operations_efficient}.

\begin{theorem}[Oracle calls to Hamiltonian evolutions and reflection operators required for preparing $| \rm{TFD}(\beta)\rangle$ to error $\varepsilon$ in \textit{practical} cases]
    \label{prop:PF-better-error}
    To achieve an $\varepsilon$-precision approximation of the thermofield double state at inverse temperature $\beta$ using the \hiddenlink{naive-DBTFD}{\texttt{vanilla}} or \hiddenlink{poly-DBTFD}{\texttt{poly} \texttt{DB-TFD}} at least $N_{\rm{tot}}$ oracle calls to Hamiltonian evolutions and reflection operators are required. For the $A(k) \in \mathcal{O}(k^{c})$ regime, the scalings of the \hiddenlink{naive-DBTFD}{\texttt{vanilla}} and \hiddenlink{poly-DBTFD}{\texttt{poly}} approaches are respectively given by
    \begin{itemize}
        \item \hiddenlink{naive-DBTFD}{\texttt{vanilla} \texttt{DB-TFD}}: 
        \begin{equation}
            N_{\rm{tot}} \in \exp\left( \tilde{\mathcal{O}} ( \exp(\beta) / \varepsilon)\right),
        \end{equation}
        \item \hiddenlink{poly-DBTFD}{\texttt{poly} \texttt{DB-TFD}}: 
        \begin{equation}
            N_{\rm{tot}} \in \exp\left( \tilde{\mathcal{O}}\left(\beta, \; \log(1/\varepsilon)\right)\right),
        \end{equation}
        where $f = \tilde{\mathcal{O}}(g)$ denotes that $f = {\mathcal{O}}\left(g \, \rm{poly}\log(g)\right)$. 
    \end{itemize}
    
\end{theorem}

A sketch of the proof is as follows. 
Assuming access to the best-known product-formula constructions as outlined above, we obtain the bound $N_m^k = \frac{2^{mk}}{m^k} \leq e^{km}$.
This yields an overall scaling of $e^{\mathcal{O}(mk)}$, and hence $e^{\mathcal{O}(k \log k)}$ upon substituting the scaling of $m$.
We then substitute the corresponding values of $k$ from Lemmas~\ref{theorem:naive_error} and~\ref{theorem:poly_error}.
Theorem~\ref{prop:PF-better-error} indicates that, in the practical regime for which Eq.~\eqref{eq:contraction_conj} holds, the query complexity of the \hiddenlink{poly-DBTFD}{\texttt{poly} \texttt{DB-TFD}} is comparable to the existing methods.

Finally, while Theorem~\ref{prop:PF-better-error} holds when the contraction coefficients grow at most polynomially (a regime that we have found to be relevant for a variety of physically motivated Hamiltonians), we cannot rule out the possibility that, in the worst case, $A(k)$ could grow exponentially in $k$. 
In such cases, as shown in App.~\ref{app_number_operations}, we find that \texttt{DB-TFD} can scale doubly exponentially in $\beta$. We leave it an open question whether there exist physically interesting Hamiltonians for which this is the case.

\medskip
\paragraph*{Error due to the statistical noise --}

Here, we discuss the error arising from statistical uncertainty in the estimation of the energy and the variance in the \hiddenlink{poly-DBTFD}{\texttt{poly} \texttt{DB-TFD}} scheme, due to a finite number of measurement shots. 
We recall that, since the \hiddenlink{naive-DBTFD}{\texttt{vanilla} \texttt{DB-TFD}} approach requires no intermediate-state measurements, it is free from statistical noise. 
For the \hiddenlink{poly-DBTFD}{\texttt{poly} \texttt{DB-TFD}}, however, the parameters $\{(\theta_{i},s_{i})\}$ depend explicitly on the energy and the variance.
As such, the empirical estimates of the energy and variance can deviate from
the true values of the energy and variance, 
leading to a deviation between the implemented evolution $| \varphi_k^{(\texttt{poly})} \rangle$ and the ideal one $| \phi_k^{(\texttt{poly})} \rangle$ at each step.
Moreover, these discrepancies accumulate over $k$ iterations.

As shown in Proposition~\ref{prop:statistical_error}, in the worst case, suppressing the error induced by statistical noise requires the precision of the energy and variance estimates to scale exponentially with the number of recursion steps $k$.
Hence, the \hiddenlink{poly-DBTFD}{\texttt{poly} \texttt{DB-TFD}} must incur additional exponential resources to suppress statistical error.
Nevertheless, as shown in Lemma~\ref{theorem:poly_error}, the recursion depth scales as $\mathcal{O}(\sqrt{\beta})$, implying that the query complexity associated with suppressing statistical errors is subdominant compared to that arising from the product formula approximation.
That is, the overall complexity is dominated by the product-formula error.
Moreover, the bound in Proposition~\ref{prop:statistical_error} is generally overly pessimistic.
The analysis relies on repeated applications of the triangle inequality, which in general substantially overestimates the error accumulation. 
In practice, the actual implementation cost is significantly lower. This is shown in Fig.~\ref{fig:shotNoiseError} for the Heisenberg model Eq.~\eqref{eq:heisenberg} of $n =8$ qubits, where the error does not grow exponentially despite the use of a constant number of measurement shots. 
Details of the shot-noise model are provided in App.~\ref{App:shotNoise}. In particular, $N_{\rm{shots}}$ denotes the total measurement-shot budget, with shots distributed uniformly across the Pauli operators composing the system Hamiltonian and its square.
 \\

\paragraph*{Numerical Implementation.}
Finally, we numerically examine the implementation costs of \hiddenlink{poly-DBTFD}{\texttt{poly} \texttt{DB-TFD}} for physically relevant quantum systems in practical scenarios using the previously defined one-dimensional models defined in Eqs.~\eqref{eq:XXZ}-\eqref{eq:heisenberg}.
In our numerical simulations, we study systems up to $n=10$ qubits, corresponding to a doubled system with $20$ qubits, where imaginary-time evolution is applied.
To make it consistent with our theoretical analysis, the Hamiltonian is normalized and shifted; further details are provided in App.~\ref{app:numerics}.
\begin{figure}
    \centering
    \includegraphics[width=\linewidth]{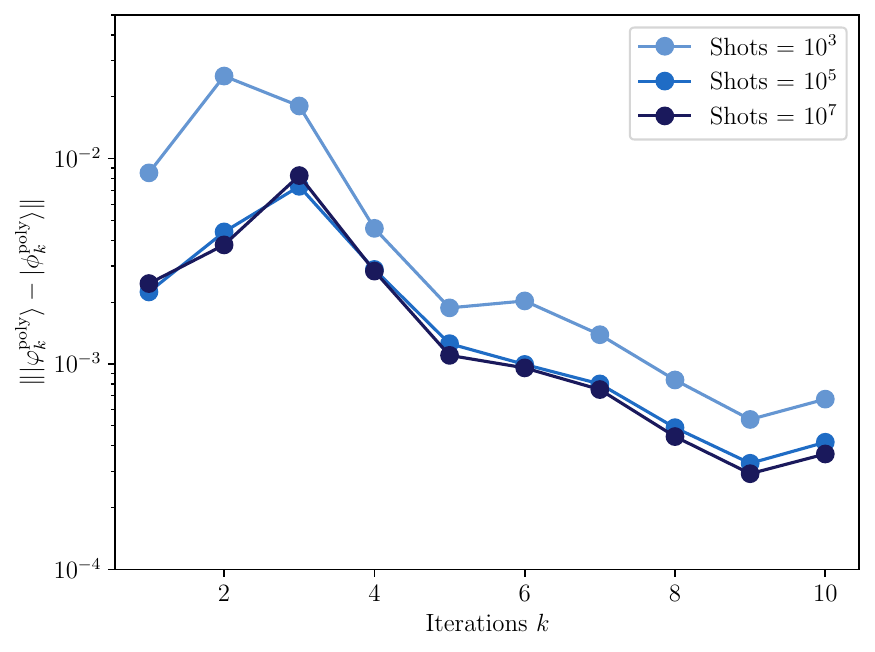}
    \caption{\textbf{Error between poly DB-TFD with infinite and finite measurement shots.} 
    We compare the impact of statistical noise in the estimation of energy and variance required to determine the parameters $(s,\theta)$.
    In particular, we compare results obtained using \texttt{poly} \texttt{DB-TFD} with infinite measurement shots and with a total of $N_{\rm{shots}}=10^3, 10^5, 10^7$ (divided uniformly between the Pauli terms that need to be measured) per iteration. In contrast to the analytical results in App.~\ref{prop:statistical_error}, no exponential growth with respect to $k$ is observed even for $N_{\rm{shots}}=10^3$. 
    Since the results obtained with $10^5$ measurements shots are comparable to those for $10^7$, we infer that the overall error is dominated by the product-formula approximation rather than the statistical measurement noise.}
    \label{fig:shotNoiseError}
\end{figure}

Fig.~\ref{fig:nQueries} shows the number of required queries to the Hamiltonian evolution and reflection operators required by the \hiddenlink{poly-DBTFD}{\texttt{poly} \texttt{DB-TFD}} to achieve the precision $\varepsilon=0.01$, as a function of inverse temperature $\beta$.
We use a total of $10^6$ measurement shots (divided uniformly between the Pauli terms that need to be measured) to estimate the energy and variance required for updating the state parameters.
Note that, in the figure, we define the number of queries as the circuit depth (i.e., $N_{m}^{k}$) and exclude the number of accesses to these operators required for the estimation of the energy and variance.
This choice reflects our focus on the circuit-depth scaling derived in the analysis above.

The observed scaling is consistent with the prediction of Theorem~\ref{prop:PF-better-error}, and comparable to the scaling achieved by existing quantum algorithms~\cite{Gily_n_2019, chowdhury2016quantumalgorithmsgibbssampling, Poulin_2009}.
These results indicate the practicality of our approach for physically relevant systems.

\begin{figure}[t]
    \centering
    \includegraphics[width= \linewidth]{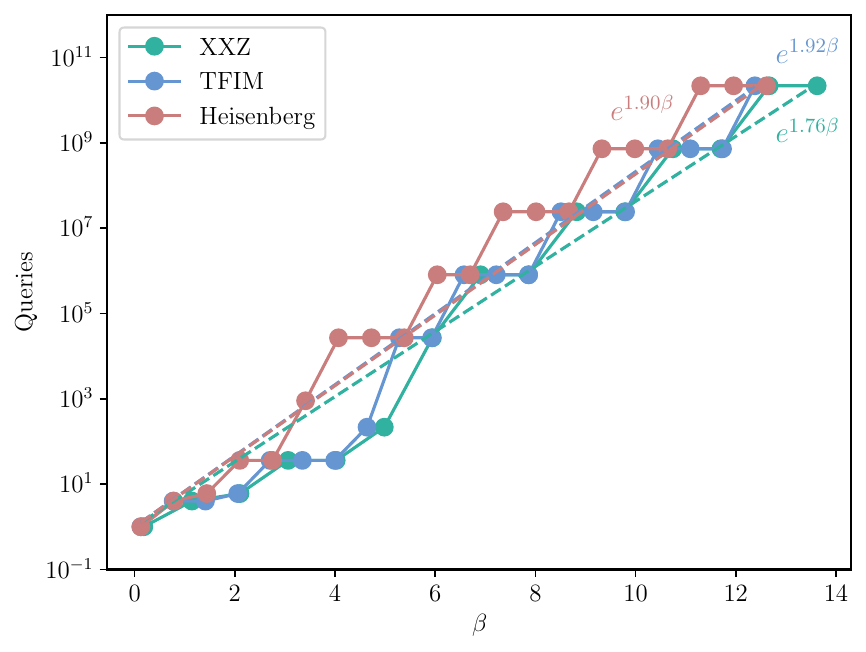}
    \caption{ \textbf{Number of queries for poly DB-TFD to obtain the thermofield double state at different inverse temperatures $\beta$.}
    We numerically evaluate the number of queries required by \texttt{poly} \texttt{DB-TFD} to achieve a thermal state with precision $\varepsilon = 0.01$.
    Here, we consider three models: one-dimensional XXZ, TFIM, and Heisenberg model. For comparison, we include reference scalings proportional to $\exp({\beta})$. The fitting suggests an overall scaling of $\exp(\mathcal{O}(\beta))$ in line with the analysis based on the assumed practical regime of contraction coefficient in Eq.~\eqref{eq:contraction_conj}.}
    \label{fig:nQueries}
\end{figure}

\section{Implementation of Quantum Boltzmann machines} \label{sec:QBM-numerics}

In this section, we demonstrate an application of our framework to generative modeling using quantum Boltzmann machines, where the Gibbs state is prepared via \texttt{DB-TFD}.

In such machine learning tasks, exact Gibbs-state preparation is not necessarily required, since performance is evaluated through the loss function rather than the fidelity of the prepared state itself. This tolerance to moderate approximation errors makes \texttt{DB-TFD} well suited to learning applications.
Moreover, as a non-variational algorithm, \texttt{DB-TFD} avoids trainability challenges, which commonly arise in variational approaches. These properties motivate the use of \texttt{DB-TFD} for generative modeling.

The objective of this task is to generate a target density matrix $\eta$ using a model parametrized as a Gibbs state:
\begin{equation}
    \rho_{\bm{\theta}} = \frac{e^{-H_{\bm{\theta}}}}{Z(\bm{\theta})},
\end{equation}
where the Hamiltonian is given by $H_{\bm{\theta}} = \sum_{i=1}^{M} \theta_i G_i$ with Hermitian operators ${G_i}$, such as Pauli operators, and the parameters $\bm{\theta}~=~(\theta_1, \dots, \theta_M)$ to be optimized.
The partition function is denoted as $Z(\bm{\theta}) = {\rm Tr}[e^{-H_{\bm{\theta}}}]$.

To optimize the parameters, we utilize the quantum relative entropy
\begin{equation}
    S(\eta||\rho_{\bm{\theta}}) = \text{Tr}(\eta\log\eta)-\text{Tr}(\eta\log\rho_{\bm{\theta}})
\end{equation}
as the loss function.
Importantly, this loss is a strictly convex function over the parameters $\bm{\theta}$, ensuring the existence of a unique minimizer $\bm{\theta^*}$~\cite{Coopmans_2024}.
The optimization can be performed using gradient descent, where the gradient is given by
\begin{equation}
    \frac{\partial S(\eta||\rho_\theta)}{\partial \theta_i} = \langle G_i \rangle_{\rho_{\bm{\theta}}} -\langle G_i \rangle_{\eta}
    \label{eq:QBM-quantum-relative-entropy}
\end{equation}
with $\braket{G}_{\Psi}:=\mathrm{Tr}[G\Psi]$.
Consequently, parameter updates require estimating the expectation values of all generators $\{G_{i}\}$ with respect to the model state $\rho_{\bm{\theta}}$.
In this setting, we apply the \texttt{DB-TFD} algorithm to prepare $\rho_{\bm{\theta}}$ to evaluate the gradient.

\begin{figure}[t]
    \centering
    \includegraphics[width=\linewidth]{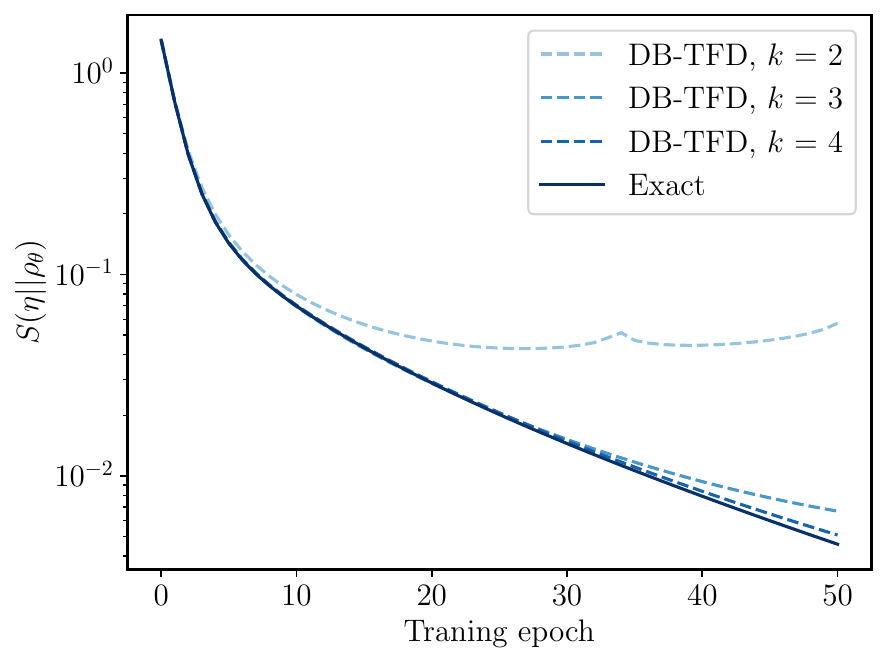}
    \caption{\textbf{ Training loss trajectories for poly DB-TFD with different polynomial degrees $k$.} 
    We consider a generative modeling task for the XXZ model to assess how effectively \texttt{poly} \texttt{DB-TFD} with polynomial degrees $k=2,3,4$ reduces the loss. At each training epoch, the thermal state required for computing the loss gradient is prepared using the \texttt{DB-TFD} approaches. Even a small number of iterations yields a training-loss trajectory that closely matches that obtained using the exact thermal state.}
    \label{fig:trainingExample}
\end{figure}

\medskip
We first demonstrate the performance of the \hiddenlink{poly-DBTFD}{\texttt{poly} \texttt{DB-TFD}} scheme on the task of preparing thermal states of the XXZ model. 
Fig.~\ref{fig:trainingExample} shows the training loss trajectories obtained with \hiddenlink{poly-DBTFD}{\texttt{poly} \texttt{DB-TFD}} for different $k=2, 3, 4$, together with the trajectory corresponding to exact thermal state preparation. 
As $k$ increases, the loss trajectory produced by \hiddenlink{poly-DBTFD}{\texttt{poly} \texttt{DB-TFD}} approaches that of the exact method. 
Notably, even with four recursive steps, the loss closely matches the exact result. 
This observation suggests that \hiddenlink{poly-DBTFD}{\texttt{poly} \texttt{DB-TFD}} can achieve high-quality thermal state preparation with a small recursion depth, indicating potential feasibility on near-term devices and early-stage fault-tolerant quantum computers.

Motivated by this observation, we next compare the performance of \hiddenlink{poly-DBTFD}{\texttt{poly} \texttt{DB-TFD}} to a near-term approach: i.e., the variational quantum imaginary-time evolution (VarQITE) algorithm~\cite{McArdle_2019, Zoufal_2021}, which performs imaginary-time evolution variationally to obtain the thermal state.
Specifically, we consider three benchmarks: two classical datasets --the Bernoulli distribution and the bar-and-stripes dataset of size $2\times 2$ ~\cite{bowles2024subtledata}-- and one quantum target given by the XXZ model~\cite{gu2003entanglement}.
Details of the numerical setup and the VarQITE implementation are provided in App.~\ref{app:numerics-generative-models}.

\begin{table}[t]
\centering
\caption{\textbf{Comparison of final training losses for different generative modeling tasks.}
We consider the XXZ model, the Bernoulli distribution with $p=1/2$, and the Bar-and-Stripes dataset of size $2\times 2$ to compare the performance of the exact thermal state, \texttt{poly} \texttt{DB-TFD}, and VarQITE.
We show the best performance over multiple trials in both the noiseless setting and two noisy settings: one with only shot noise and the other with only depolarization noise. We consider two different total numbers of shots and two different error rates.          
For the two noisy scenarios, the parameters of \texttt{poly} \texttt{DB-TFD} are set to $k=4$ and $m=4$, whereas the circuit depth for VarQITE is set to $l=2$ as this is the best among $l=1,\ldots,5$.}
\label{tab:QBMTable2}
\small
\begin{tabularx}{\linewidth}{lXXXX}
\toprule
 & XXZ & Bernoulli \newline ($p=\tfrac12$) & Bar Stripes \newline $2\times2$  \\
\midrule
Exact           & 0.009 & 0.014 & 0.097  \\
\hline
DB-TFD ($k=5,m=5 $)          & \textbf{0.009} & \textbf{0.021}  & \textbf{0.101} \\
DB-TFD ($k=4,m=4$)          & 0.013 & 0.023  & 0.191 \\
VarQITE ($l=2$)         & 0.011 & 0.028 & 0.108  \\
\hline
Shot noise \\
DB-TFD ($N_{\rm{shots}}=10^5$)    & 0.013 & 0.023& 0.191 \\
DB-TFD ($N_{\rm{shots}}=10^{11}$)    & 0.013 & 0.023& 0.191 \\
VarQITE ($N_{\rm{shots}}=10^5$)    & 1.35 & 1.24 & 1.62  \\
VarQITE ($N_{\rm{shots}}=10^{11}$)    & 0.019 &  0.035  &  0.141 \\
\hline
Depolarizing noise \\
DB-TFD ($q=10^{-5}$)   &  0.063 & 0.197 & 0.337  \\
DB-TFD ($q=10^{-6}$)   &   0.015 & 0.036 & 0.204  \\
VarQITE  ($q=10^{-5}$)  & 0.011 & 0.028 & 0.108 \\
VarQITE  ($q=10^{-6}$)  & 0.011 & 0.028 & 0.108 \\
\bottomrule

    \end{tabularx}
\end{table}

We first consider an idealized setting in which the only source of error arises from the product-formula approximation, while other noise sources are neglected. 
As summarized in Table~\ref{tab:QBMTable2}, the \hiddenlink{poly-DBTFD}{\texttt{poly} \texttt{DB-TFD}} with the recursion depth fixed to $k=5$ and a product-formula order $m=5$ outperforms VarQITE across all tasks.
We also consider the case $k=4$ and $m=4$, where the approximation to the thermal state is less accurate.
In this case, the \hiddenlink{poly-DBTFD}{\texttt{poly} \texttt{DB-TFD}} exhibits better performance than VarQITE for the Bernoulli distribution task and
 VarQITE (slightly) outperforms \texttt{DB-TFD} for the XXZ model and the bar-and-stripes dataset (but the performance of the two algorithms on these tasks remains comparable).
Note that we observe VarQITE with depth $l=2$ performs best among $l=1,\ldots,5$; this is possibly because trainability issues or instability in the matrix inversion lead to larger errors as the depth, and consequently the number of parameters, increases. We expect the same to hold with shot and depolarizing noise, given that the noise only worsens with the number of layers.

\medskip
We further examine a more interesting setting in which two types of realistic noise are present during execution: finite measurement statistics, which introduce statistical noise in the estimation of energy expectations, and depolarizing noise acting on the quantum circuit. To illustrate the distinct effects of these noise models clearly, we apply them separately, i.e., firstly only shot noise is included and then only depolarizing noise. 
For the former scenario, a \textit{total} of $N_{\rm{shots}}=10^{5}, 10^{11}$ measurement shots (divided uniformly across all terms that need to be measured) per iteration are considered.
In the latter, depolarizing noise with the error rates $q=10^{-5},10^{-6}$ acts each time an elementary gate is applied in the circuit. That is, the noise acts on each gate layer in VarQITE, whereas it acts on either a Hamiltonian evolution or a reflection gate for \texttt{DB-TFD}.
For the implementation, we set $k=4$ and $m=4$ for the \hiddenlink{poly-DBTFD}{\texttt{poly} \texttt{DB-TFD}}.
See App,~\ref{App:shotNoise} for the details.
Table~\ref{tab:QBMTable2} shows the performance of both methods under these noise models.

In the presence of shot noise, \hiddenlink{poly-DB-TFD}{\texttt{poly} \texttt{DB-TFD}} consistently outperforms VarQITE. 
The performance of VarQITE deteriorates substantially when only a limited number of measurement shots is available, whereas \hiddenlink{poly-DB-TFD}{\texttt{poly} \texttt{DB-TFD}} maintains stable performance across all shot budgets considered.
Indeed, as the number of measurement shots is increased to $N_{\rm{shots}}=10^{11}$, the performance of VarQITE approaches that observed in the noiseless setting. 
This observation suggests that the proposed method is substantially more measurement-efficient through the implementation, requiring fewer measurement shots throughout the computation to achieve a comparable level of performance.

On the other hand, VarQITE exhibits greater robustness in the presence of depolarizing noise. For a depolarizing error rate of $q=10^{-6}$, \hiddenlink{poly-DB-TFD}{\texttt{poly} \texttt{DB-TFD}} attains performance comparable to that of VarQITE on two of the considered datasets, namely the XXZ model and the Bernoulli distribution.
However, when the error rate is increased to $q=10^{-5}$, the performance of \hiddenlink{poly-DB-TFD}{\texttt{poly} \texttt{DB-TFD}} degrades noticeably, while VarQITE remains comparatively stable. These results indicate that VarQITE is more resilient to depolarizing noise than the proposed approach.

The observed performance gap between the \hiddenlink{poly-DBTFD}{\texttt{poly} \texttt{DB-TFD}} and VarQITE under noisy settings can be attributed to their different resource requirements.
We note that both methods use $N_{\rm{shots}}$ number of measurements shots per observable to estimate the set of expectation values $ \{ \langle G_i \rangle_{\rho_{\bm{\theta}}} \}_{i=1}^{M}$, which are used to update the parameters $\bm{\theta}$.
This suggests that the difference ultimately reduces to the cost of the thermal state preparation required at each step.

In VarQITE, a parametrized quantum circuit with $L$ parameters is iteratively applied for $k^{(\rm{V})} $ steps to prepare an approximation of the thermal state. 
Therefore, obtaining the information required to update the parameters in parametrized quantum circuits entails $\mathcal{O}\left( (L^2 + LM) k^{(\rm{V})} \right)$ circuit execution at each optimization step. 
In contrast, \hiddenlink{poly-DBTFD}{\texttt{poly} \texttt{DB-TFD}} only requires estimating the energy and its variance at each step. 
The corresponding measurement cost therefore scales as $\mathcal{O}\left((M^2 + M) k^{(\texttt{poly})}\right)$, where $k^{(\texttt{poly})}$ denotes the required polynomial degree.
In practice, the number of parameters $L$ in VarQITE may exceed the number of Hamiltonian terms $M$ to faithfully reproduce the imaginary-time evolution dynamics.
Under this condition, and for a fixed iteration count, \hiddenlink{poly-DBTFD}{\texttt{poly} \texttt{DB-TFD}} requires fewer measurements. 

In contrast, the circuit depth of VarQITE remains fixed, whereas that of \texttt{DB-TFD} grows exponentially with the repetition number in theory. This distinction may contribute to the greater robustness of VarQITE in the presence of depolarizing noise.
Nevertheless, our numerical results indicate that only very small values of $k^{(\texttt{poly})}$ are sufficient in practice, often much smaller than the iteration count required by VarQITE. This observation may mitigate the overhead associated with increased circuit depth.

Overall, these findings suggest that \hiddenlink{poly-DBTFD}{\texttt{poly} \texttt{DB-TFD}} constitutes a promising and implementable approach for generative modeling tasks on currently available quantum hardware, well before fully scalable fault-tolerant quantum computers become available.

\medskip
\medskip

\section{Discussion \& Conclusion}
 \label{sec:conclusion}

We propose \texttt{DB-TFD}, a quantum algorithm for thermal state preparation which applies imaginary-time evolution to thermofield double states using double-bracket quantum algorithms. 
In regimes where the contraction coefficient scales sub-polynomially, \hiddenlink{poly-DBTFD}{\texttt{poly} \texttt{DB-TFD}} achieves an exponential query complexity in the inverse temperature $\beta$ (matching the best-known scalings achieved by existing quantum algorithms~\cite{Gily_n_2019, chowdhury2016quantumalgorithmsgibbssampling, Poulin_2009}). 
Our numerical simulations further support the corresponding theoretical result.
Moreover, applying our \texttt{DB-TFD} to generative modeling with quantum Boltzmann machines shows that even a small number of recursive steps suffices to outperform existing variational approaches, particularly in shot-noise-dominated settings. These results highlight the practical utility of \texttt{DB-TFD} for late near-term experiments and early fault-tolerant quantum devices.

An important direction for future work is to further investigate the regimes in which our algorithm remains efficient. In particular, we have provided numerical evidence that the contraction coefficients scale favourably for a number of physically motivated models, but whether we can rigorously prove these scalings in interesting regimes remains an open question.
Although our numerical simulation supports the validity of the regime, establishing how broadly this applies both analytically and numerically will be key to assessing the actual power of the algorithm. 

Another important question is how the complexity of our algorithm compares with that of prior methods. 
A key difficulty is that the notion of query complexity differs across approaches: in our setting, queries correspond to Hamiltonian evolutions and reflection operators, whereas many fault-tolerant algorithms rely on block-encoding of Hamiltonians or Lindbladian jump operators. 
To be meaningful, any quantitative comparison must compare our query complexity against the ancilla qubit overheads, the implementation costs of complex jump operators, and potentially the exponential mixing times of the dissipative methods.
As a result, direct quantitative comparisons are nontrivial. 
Nevertheless, establishing a comparison of query complexities within a common elementary gate set remains an important and nontrivial goal.

One promising direction to improve the efficiency of our approach may be to tailor alternative Hamiltonians. 
Specifically, one can reformulate the thermal state preparation problem in terms of a Hamiltonian whose ground state encodes the desired thermal state, thereby reducing the task to ground-state preparation. Ground-state preparation can be substantially more efficient in practice than suggested by worst-case analytical bounds.
This is because the task does not require an accurate implementation of the exponential operators; rather, it suffices that each operation reduces the energy, even if the operator itself is implemented imperfectly. 
Similar ideas have been explored using different Hamiltonian constructions in related contexts~\cite{Holmes_2022,Cottrell_2019}. 
Leveraging such Hamiltonian design within the double-bracket quantum algorithm framework may  help prove that the contraction coefficient scales polynomially as assumed for the practical regime.

From a broader perspective, the proposed approach offers a potential interface between quantum algorithms and holography-inspired questions of complexity. 
In the AdS/CFT correspondence~\cite{israel1976thermo,Maldacena_2003}, thermofield double states play a central role as dual of the two-sided eternal black holes.
Since the present algorithm prepares thermofield double states through imaginary-time evolution, the required circuit depth and algorithmic cost provide a concrete handle on how increasingly nonlocal correlations are assembled as the temperature is lowered.
This perspective might suggest a possible connection between quantum algorithmic complexity and holographic complexity~\cite{susskind2016computational,brown2016holographic,harlow2013quantum,brown2020python}.
Our \texttt{DB-TFD} operates in a doubled Hilbert space without ancilla qubits and therefore describes an isolated composite system, rather than a theory with an explicit gravitational dual. Nevertheless, analyzing the preparation complexity may provide insight into how entanglement structure and algorithmic complexity jointly contribute to the entanglement–geometry correspondence underlying black hole physics.

\medskip
\medskip
\textit{Acknowledgments.}
RAS acknowledges support from the Swiss National Science Foundation [grant number 200021-219329].
ST acknowledges Exchange Faculty Travel Grant: TG168033 from Chulalongkorn University, as well as funding from National Research Council of Thailand (NRCT) [grant number N42A680126]. ST further acknowledges Thailand Science research and Innovation Fund Chulalongkorn University (IND\_FF\_69\_258\_2300\_062).
ZH acknowledges support from the Sandoz Family Foundation-Monique de Meuron program for Academic Promotion. 
AW acknowledges insightful discussions with Michele Minervini.\\
\bibliographystyle{naturemag}
\bibliography{apssamp}

\clearpage
\newpage
\onecolumngrid
\part{Appendix}

\parttoc

\appendix

\renewcommand{\thetheorem}{\Alph{section}.\arabic{theorem}}
\renewcommand{\thelemma}{\Alph{section}.\arabic{lemma}}
\renewcommand{\theprop}{\Alph{section}.\arabic{prop}}

\newtheorem{theoremA}{Theorem}[section]
\newtheorem{propositionA}[theoremA]{Proposition}
\newtheorem{lemmaA}[theoremA]{Lemma}
\newtheorem{definitionA}[theoremA]{Definition}
\newtheorem{corollaryA}[theoremA]{Corollary}
\newtheorem{remarkA}[theoremA]{Remark}

\renewcommand{\thetheoremA}{\Alph{section}.\arabic{theoremA}}

\newpage

\section{Related Work} \label{app:related_work}

\subsection{Existing methods for thermal state preparation}

This section outlines strategies for thermal state preparation, focusing on purification-based methods and dissipative engineering. This review is non-exhaustive and does not mention methods designed to sample from the thermal state rather than preparing it explicitly. Table~\ref{tab:thermalStatePrepMethods} presents the methods possessing scalings as well as the \texttt{DB-TFD} for ease of comparison.

\subsubsection{Approximating imaginary-time evolution with maximally entangled initial states}

We first consider methods analogous to the present approach, which prepare the thermal state by applying an approximation of imaginary-time evolution to a maximally entangled state.

Ref.~\cite{chowdhury2016quantumalgorithmsgibbssampling} approximates the exponential operator via a Linear Combination of Unitaries (LCU). For Hamiltonians decomposable into linear combinations of Pauli strings, the computational cost scales as
\begin{equation}
    \mathcal{O}\left(\sqrt{\beta} e^{\beta \|H\|_{\mathrm{op}}} \mathrm{polylog}\left(\sqrt{\beta} e^{\beta \|H\|_{\mathrm{op}}}\log(1/\varepsilon)\right)\right).
\end{equation}
This implementation also necessitates an ancilla register scaling as $\mathcal{O}(\log(N))$.

Alternatively, Ref.~\cite{Gily_n_2019} employs a block encoding technique based on polynomial approximations of the exponential. While the query complexity for the unitary scales as $\mathcal{O}(\sqrt{\beta \|H\|_{\mathrm{op}}}\log(1/\varepsilon))$, achieving a constant success probability requires $\mathcal{O}(\sqrt{D/Z(\beta)})$ amplification rounds, where $\| \cdot \|_{\mathrm{op}}$ denotes the operator norm. Consequently, the overall scaling is
\begin{equation}
    \mathcal{O}\left(\sqrt{\frac{\beta\|H\|_{\mathrm{op}} D}{Z(\beta)}} \log(1/\varepsilon)\right).
\end{equation}
Since the partition function is bounded by $Z(\beta)\leq e^{\beta\|H\|_{\mathrm{op}}}$, this scaling remains exponential in $\beta$. Furthermore, implementing the polynomial via block encoding requires $n+2$ ancilla qubits (where $n$ is the number of qubits required to block encode $H$) and access to controlled-$H$ unitaries.

Finally, Ref.~\cite{Motta_2019} approximates small imaginary time steps using real-time evolution combined with a classical minimization routine. This method remains efficient provided that correlations in the evolved state remain limited. In the worst case, the number of measurements required is $\exp(\mathcal{O}(C^d))$, where $C$ is the correlation length and $d$ is the lattice dimension; a similar scaling applies to the classical minimization. As the correlation length typically diverges at low temperatures, this approach is primarily efficient in the high-temperature regime.

\subsubsection{Lindbladian Approaches}

Another class of methods utilizes specific Lindbladians, known as Davies generators~\cite{Davies1974}, to drive the system toward the thermal state~\cite{chen2023quantumthermalstatepreparation, Rall_2023, ding2025endtoendefficientquantumthermal}. By mimicking the dynamics of an open system coupled to a thermal bath, these methods naturally reproduce physical thermalization. The primary challenges are the implementation of non-unitary Lindbladian evolution and the mixing time $t_{\mathrm{mix}}$, the duration required for convergence, which can be exponentially long at low temperatures for many systems.

The approach in Ref.~\cite{chen2023quantumthermalstatepreparation} proposes block-encoding of the Lindbladian $\mathcal{L}$. However, the complex structure of Davies generators necessitates resource-intensive operations, such as controlled-Hamiltonian evolution and quantum Fourier transforms, limiting the suitability of this method for early fault-tolerant architectures.

Conversely, Ref.~\cite{ding2025endtoendefficientquantumthermal} proposes a less hardware-demanding approach where the system-reservoir coupling is simulated using a single ancilla. This reduces the requirement to a simple Hamiltonian evolution. The overall scaling is $\mathcal{O}(t_{\mathrm{mix}}^4\beta^2 / \lambda_{\mathrm{gap}}^2\varepsilon^3)$, where $\lambda_{\mathrm{gap}}$ is the spectral gap of the Hamiltonian. While efficient when the mixing time is polynomial, this algorithm may require longer simulation times compared to the approaches based on block-encoding, which scale as $\mathcal{O}(\beta t^2_{\mathrm{mix}}/\varepsilon)$, albeit with significantly reduced gate complexity.

\subsubsection{Variational Approaches}

Variational approaches are a popular choice for noisy intermediate-scale quantum devices and have also been applied to thermal state preparation~\cite{wu2019variational, Zhu_2020, sewell2022thermalmultiscaleentanglementrenormalization, Sagastizabal_2021, Zoufal_2021}.
Most existing works focus on variationally preparing the thermofield double state~\cite{wu2019variational, Zhu_2020, Sagastizabal_2021, Zoufal_2021}.
These approaches typically rely on minimizing the free energy, a cost function that requires the evaluation of entanglement entropy~\cite{Islam_2015}, a task more demanding than standard energy estimation.

Although shallow circuits are often employed to accommodate the NISQ regime, several limitations remain. In particular, rigorous guarantees on scaling with system size or temperature are absent. 
As with other variational methods, optimization is further complicated by unfavorable landscape features, most notably barren plateaus~\cite{Larocca_2025}, which can severely impede convergence in larger systems.

VarQITE potentially provides a notable exception~\cite{McArdle_2019}.
Rather than minimizing a free-energy-based objective, VarQITE updates circuit parameters so that the evolution approximates imaginary-time evolution to first order. 
The McLachlan variational principle governs the update rule, thereby eliminating the need for explicit free energy estimation. 
This strategy has been applied to thermal state preparation in the context of quantum Boltzmann machines~\cite{Zoufal_2021}.

\subsubsection{Other Approaches}

Purification approaches have also been exploited via Quantum Phase Estimation (QPE). Ref.~\cite{Poulin_2009} utilizes QPE to estimate the eigenenergies of the system, followed by controlled rotations on an ancilla qubit proportional to the Boltzmann weights and amplitude amplification. While rigorous, this method requires deep circuits for phase estimation and an amplification overhead scaling with $\sqrt{1/Z(\beta)}$, facing similar exponential bottlenecks at low temperatures as approaches based on block encoding.

Finally, Ref.~\cite{Holmes_2022} introduces a perturbative framework where the thermal state of a Hamiltonian $H = H_0 + V$ is prepared starting from the thermal state of $H_0$ at the same inverse temperature $\beta$. The complexity of the algorithm scales as $\mathcal{O}(\exp{(\beta \|V\|_{\rm{op}})})$. The method relies on implementing the non-unitary operator $e^{-\frac{\beta}{2} W}$, where $W = H\otimes \mathds{1} - \mathds{1}\otimes H_0^\star$, via an LCU implementation using ancilla qubits. Notably, this operator can be directly implemented using double-bracket iterations, suggesting that the methods proposed herein may be extended to perturbative regimes.

\begin{table}[btp]   
\centering
\caption{Comparison table of quantum algorithms for preparing thermal states. Here $f = \tilde{\mathcal{O}}(g)$ indicates that $f = \mathcal{O}\left( g \, \mathrm{poly log}(g)\right)$.}
\begin{adjustbox}{angle=0}
\begin{tabular}{|p{3cm}|p{3.2cm}|p{4cm}|l|p{1.5cm}|p{4cm}|}
\hline
Algorithm& Number of calls to subroutines & Complexity per subroutine   & Classical costs&Additional costs& Total Depth Assumptions/Remarks \\
\hline
LCU approximation of the exponential \cite{chowdhury2016quantumalgorithmsgibbssampling}& $\mathcal{O}\left( \sqrt{\frac{D}{Z(\beta)}}\right)$&$\mathcal{O}\left( \sqrt{\beta} \text{polylog} \left(\frac{\sqrt{\beta}}{\varepsilon} \right)\right)$  & -&-& Requires access to $\sqrt{H}$.\\ \hline
Perturbation method \cite{Holmes_2022}&$\tilde {O} \left(e^{\beta \|V\|_{\rm{op}}}\right) $ &$O(\log(1/\varepsilon) + \text{poly}(\beta))$  & -&-& If $H = H_0 + V$ and we start from the thermal state of $H_0$ at temperature $\beta$ \\ \hline
Classically assisted QITE \cite{Motta_2019}& $O\left(\frac{\beta n}{\varepsilon} \right)$ &$\mathcal{O}(\exp(C(\beta)^d))$& $\mathcal{O}(\exp(C(\beta)^d))$&-& Depends on the correlations $C$ between the sites which increases as the state becomes colder.\\ \hline
QSVT implementation of the exponential \cite{Gily_n_2019}& $O\left(\sqrt{\frac{D}{Z(\beta)}}\right)$ &$\mathcal{O}(\sqrt{\beta}\log\frac{1}{\varepsilon})$ & -&-& -\\ \hline
Lindbladian (Operator FT) \cite{chen2023quantumthermalstatepreparation}& $\tilde{\mathcal{O}}(t_{\mathrm{mix}})$ & $\tilde{\mathcal{O}}\left(\frac{\beta t_{\mathrm{mix}}}{\varepsilon}\right)$ & -&-& -\\ \hline
Lindblandian with one resettable ancilla~\cite{ding2025endtoendefficientquantumthermal}& $\tilde{\mathcal{O}}\left(\frac{t_{\mathrm{mix}}^4\beta^2} {\varepsilon^3}\right)$& $\tilde{\mathcal{O}}\left(\frac{\beta t_{\mathrm{mix}}}{\varepsilon}\right)$ & -&-& Requires only forward Hamiltonian evolution and a single reusable ancilla.\\\hline
\texttt{vanilla} \texttt{DB-TFD}& $O\left( \frac{e^{\beta \|H\|_{\rm{op}}}}{\varepsilon}\right)$&  $\exp\left(\tilde{\mathcal{O}}\left(e^\beta/\varepsilon\right)\right)$ in the practical regime where $A(k)\in \mathcal{O}(k)$& -& -&No additional ancilla nor post-selection. Assumes exact commutator\\\hline
 \texttt{poly} \texttt{DB-TFD}& $O \left(\sqrt{\beta\log \left( \frac{1}{\varepsilon}\sqrt{\frac{D}{Z}} \right)} \right) $& $\exp\left(\tilde{\mathcal{O}}\left(\beta, \; \log(1/\varepsilon)\right)\right)$ in the practical regime where $A(k)\in \mathcal{O}(k)$& -& Shots for measuring $E$ and $V$&Assumes exact commutator.\\\hline
\end{tabular}
\end{adjustbox}

\label{tab:thermalStatePrepMethods}

\end{table}

\section{Properties of thermal state preparation via thermofield double state simulation}

In this section, we show how applying imaginary-time evolution to a thermofield double state after the partial trace leads to convergence toward a thermal state.

We begin by demonstrating that imaginary-time evolution applied to a thermofield double state produces another thermofield double state at a different temperature.
Recall that the imaginary-time evolution state $|\Phi(\tau)\rangle$ is defined as 
\begin{equation}
    | \Phi(\tau)\rangle = \frac{e^{-\tau H} |\Phi_0\rangle}{\| e^{-\tau H} |\Phi_0\rangle \|} \; ,
    \label{eq:ITE-state}
\end{equation}
with $H = H_s \otimes \mathds{1}_a$ and the initial state $| \Phi_0\rangle$.
Then, when the initial state is chosen to be a thermofield double state $| \psi_0\rangle= \ket{{\rm TFD}(\beta_{0})}$ at a given inverse temperature $\beta_{0}$, the resulting imaginary-time evolution state can be computed explicitly. 
For clarity, we evaluate the numerator and denominator separately.
The numerator can be written as follows; 
\begin{align}
    e^{- \tau H} | \Phi_0\rangle &= e^{-\tau H} | {\rm TFD}(\beta_0)\rangle \nonumber\\  &= e^{-\tau H} \frac{1}{\sqrt{Z_{\beta_0}}}\sum_{n} e^{-\beta_0 E_n/2} |n\rangle_{\rm sys} | \tilde{n}\rangle_{a} \nonumber \\
    & = \frac{1}{\sqrt{Z_{\beta_0}}} \sum_{n} e^{-(\beta_0 + 2\tau)E_n/2} |n\rangle_{\rm sys} | \tilde{n}\rangle_{a}  \nonumber ,
\end{align}
with $Z_{\beta_0} = \mathrm{Tr}[e^{-\beta_0 H_s}]$. 
Also, the normalization factor is given by
\begin{align}
     \| e^{- \tau H} | \Phi_0\rangle \| &= \sqrt{\langle {\rm TFD}(\beta_0) | e^{-2\tau H} | {\rm TFD}(\beta_0)\rangle} \nonumber  \\
     & = \sqrt{\frac{1}{Z_{\beta_0}} \sum_{n} e^{-(\beta_0 + 2\tau)E_n}} \nonumber\\ & = \sqrt{\frac{Z_{\beta_0 + 2\tau}}{Z_{\beta_0}}}.  \nonumber
\end{align}
Combining these expressions yields
\begin{align}
    | \Phi(\tau) \rangle = \frac{e^{- \tau H} | {\rm TFD}(\beta_0)\rangle}{\| e^{- \tau H} | {\rm TFD}(\beta_0)\rangle \|} = \frac{1}{\sqrt{Z_{\beta_0 + 2\tau}}}\sum_{n} e^{-(\beta_0 + 2\tau)E_n/2} |n\rangle_{\rm sys} | \tilde{n}\rangle_{a} = | {\rm TFD}(\beta_0+2\tau)\rangle . 
\end{align}
This result shows that imaginary-time evolution maps a thermofield double state to another thermofield double state with the same Hamiltonian but at inverse temperature $\beta_{0}+2\tau$.
Consequently, the reduced state of the system register,  $\psi_{\rm sys}(\tau)$, is exactly the thermal state $\rho_{\beta_0 + 2\tau}$ associated with the same Hamiltonian $H_s$.

\medskip
\medskip
Next, we examine how free energy can be used to quantify convergence toward the target thermal state. 
A commonly used measure is the quantum relative entropy, which quantifies the distinguishability between two quantum states, with smaller values indicating greater similarity.
However, evaluating this quantity requires explicit access to the target thermal state $\rho_{\beta}$, which may not be available in practical settings.
For this reason, we check whether the free energy can serve as an alternative indicator of convergence instead of the relative entropy, analogous to the variational principle in ground-state preparation, where lower energy corresponds to higher state fidelity.

Let $\rho_{\beta}$ denote the target thermal state at inverse temperature $\beta=1/T$ associated with the Hamiltonian $H_{s}$.
The free energy evaluated on the system register after discarding the auxiliary system is then defined as
\begin{equation}
    F(\rho_{\rm sys}) = E(\rho_{\rm sys}) - T S(\rho_{\rm sys}) = \mathrm{Tr}[\rho_{\rm sys} H_s] + T \mathrm{Tr}[\rho_{\rm sys} \, \log\rho_{\rm sys}] \;,
    \label{eq:free-energy-definition}
\end{equation}
where $\rho_{\rm sys}=\mathrm{Tr}_{\rm a}[\rho]$ denotes the reduced state on the system register obtained by tracing out the auxiliary subsystem of the composite state $\rho$. 
Actually, the quantum relative entropy between $\rho_{\rm sys}$ and $\rho_{\beta}$ can be expressed directly in terms of the free energy: 
\begin{align}
    S(\rho_{\rm sys} \| \rho_{\beta}) &= \mathrm{Tr}[\rho_{\rm sys} \left( \log \rho_{\rm sys} - \log \rho_{\beta}\right)] \nonumber \\&=  \mathrm{Tr}[\rho_{\rm sys} \log \rho_{\rm sys}] - \mathrm{Tr}[\rho_{\rm sys} \log \rho_{\beta}] \nonumber \\
    &= - S(\rho_{\rm sys}) - \mathrm{Tr}[\rho_{\rm sys} \left( - \beta H_s - \log(\mathrm{Tr}[e^{- \beta H_s}])\right)] \nonumber \\
    & = - S(\rho_{\rm sys}) - \beta \mathrm{Tr}[\rho_{\rm sys} H_s] + \log(\mathrm{Tr}[e^{- \beta H_s}]) \nonumber \\
    & = \beta \left( E(\rho_{\rm sys}) - T S(\rho_{\rm sys})\right) + \log(\mathrm{Tr}[e^{- \beta H_s}]) \nonumber \\
    & = \beta \left( F(\rho_{\rm sys}) - F(\rho_{\beta})\right) \; , \label{eq:free-energy-and-relative-entropy}
\end{align}
where in the second line we use the definition of the thermal state $\rho_{\beta} = e^{-\beta H_s}/Z$ with $Z = \mathrm{Tr}[e^{-\beta H_s}]$, and we have $F(\rho_{\beta}) = -T \log Z = - T \log(\mathrm{Tr}[e^{-\beta H_s}])$ in the last line. 
This relation shows that minimizing the free-energy gap is equivalent to minimizing the quantum relative entropy with respect to the target thermal state.
More precisely, since $ F(\rho_{\beta})$ is independent of the optimization over $\rho_{\rm sys}$, minimizing $F(\rho_{\rm sys})$ is equivalent to minimizing $S(\rho_{\rm sys} \| \rho_{\beta})$, up to an additive constant.
Consequently, the free energy provides a valid cost function for certifying convergence toward the thermal state.

In particular, when imaginary-time evolution is applied to thermofield double states, the thermal state is the unique minimum of the free energy along the imaginary-time evolution trajectory.
The free energy of the system register for the imaginary-time evolution state, after the partial trace over the auxiliary subsystem, is given by
\begin{align}
    F(\Phi_{\rm sys}(\tau)) & = \mathrm{Tr}[\Phi_{\rm sys}(\tau) H_s] + \frac{1}{\beta} \mathrm{Tr}[\Phi_{\rm sys}(\tau) \log \Phi_{\rm sys}(\tau)] \nonumber \\
    & = \mathrm{Tr}[\Phi(\tau) H] + \frac{1}{\beta} \mathrm{Tr}[\rho_{\beta_0 + 2\tau} \log \rho_{\beta_0 + 2\tau}] \nonumber \\
    & = E(\Phi(\tau)) - \frac{\beta_0 + 2\tau}{\beta} \mathrm{Tr}[\rho_{\beta_0 + 2\tau} H_s] - \frac{1}{\beta} \log Z_{\beta_0 + 2\tau}  \nonumber\\
    &= \left( 1 - \frac{\beta_0 + 2\tau}{\beta}\right) E(\Phi_{\rm sys}(\tau)) - \frac{1}{\beta} \log Z_{\beta_0 + 2\tau}
    \label{eq:ITE-free-energy-definition}
\end{align}
where $\Phi_{\rm sys}(\tau) = \mathrm{Tr}_{\rm a}[\ket{\Phi_{\rm sys}(\tau)}\bra{\Phi_{\rm sys}(\tau)}]$ denotes the reduced state of the system register, $E(\Phi_{\rm sys}(\tau)) = \mathrm{Tr}[\Phi_{\rm sys}(\tau) H_s]$ and $Z_{\beta_0+2\tau} = \mathrm{Tr}[e^{- (\beta_0+2\tau) H_s}]$. 
With this notation, the derivative of the free energy can be written as
\begin{align}
    \partial_{\tau} F(\Phi_{\rm sys}(\tau)) 
    & = \left( 1 - \frac{\beta_0 + 2\tau}{\beta}\right) \partial_\tau E(\Phi(\tau)) - \frac{2}{\beta} E(\Phi(\tau)) - \frac{1}{\beta}  \partial_\tau \log Z_{\beta_0 + 2\tau} ,
    \label{eq:ITE-free-energy-derivative-TFD-initial}
\end{align}
where $E(\Phi(\tau)) =\mathrm{Tr}[ \Phi(\tau) H]$.
Note that, since $H = H_s \otimes \mathds{1}_a$, the energy can equivalently be expressed as
\begin{align}
    E(\Phi_{\rm sys}(\tau)) &= \mathrm{Tr}_{\rm sys}[\Phi_{\rm sys}(\tau) H_s] = \mathrm{Tr}[ \Phi(\tau) H_s \otimes \mathds{1}_a] = \mathrm{Tr}[ \Phi(\tau) H] = E(\Phi(\tau)).
\end{align}
Similarly, for the variance defined as $V(\Phi_{\rm sys}(\tau)) = \mathrm{Tr}[\Phi_{\rm sys}(\tau) H_s^2] - \mathrm{Tr}[\Phi_{\rm sys}(\tau) H_s]^2$, we get
\begin{align}
    V(\Phi_{\rm sys}(\tau)) &= \mathrm{Tr}_{\rm sys}[\Phi_{\rm sys}(\tau) H_s^2] - \mathrm{Tr}_{\rm sys}[\Phi_{\rm sys}(\tau) H_s]^2 = \mathrm{Tr}[ \Phi(\tau) H^2] - \mathrm{Tr}[ \Phi(\tau) H]^2 =: V(\Phi(\tau))\;.
\end{align}
To simplify the derivative, we should compute $\partial_\tau E(\Phi(\tau))$ and $\partial_\tau \log Z_{\beta_0 + 2\tau}$ explicitly. 
The first term can be written as
\begin{align}
    \partial_\tau E(\Phi(\tau)) &= \partial_\tau \mathrm{Tr}[\Phi(\tau)H]\\& = \mathrm{Tr}[\partial_\tau\Phi(\tau)H] \\ 
    &= \mathrm{Tr}\big[ [[\Phi(\tau), H], \Phi(\tau)] H\big] \\
    & = -2 \left(\mathrm{Tr}[\Phi(\tau) H^2] - \mathrm{Tr}[\Phi(\tau)H]^2\right) \\& = -2 V(\Phi(\tau)) \;, 
\end{align}
where, in the second line, we use the differential equation form of the imaginary-time evolution, $\partial_\tau\Phi(\tau) = [[\Phi(\tau), H], \Phi(\tau)]$.
The derivative of the partition function can also be simplified as
\begin{align}
    \partial_\tau \log Z_{\beta_0 + 2\tau} = \frac{\partial_\tau Z_{\beta_0 + 2\tau}}{Z_{\beta_0 + 2\tau}} = \frac{1}{Z_{\beta_0 + 2\tau}} \mathrm{Tr}[-2 H_s e^{-(\beta_0+2\tau)H_s}] = -2 \mathrm{Tr}[H_s \rho_{\beta_0+2\tau}] = - 2 E(\Phi(\tau)) \;.
\end{align}
Substituting these expressions yields
\begin{equation}
    \partial_\tau F(\Phi_{\rm sys}(\tau)) = -2 \left( 1 - \frac{\beta_0 + 2\tau}{\beta}\right) V(\Phi(\tau)) \;. 
\end{equation}
This result shows that the free energy is minimized when $\beta_0 + 2\tau = \beta$.
Equivalently,imaginary-time evolution can drive the reduced state $\Phi_{\rm sys}(\tau)$ toward the thermal state at temperature $\beta$, thereby minimizing the relative entropy with respect to the target thermal state.

However, even when imaginary-time evolution is applied, the minimizer of the free energy does not generally correspond to the target thermal state if the initial state is not a thermofield double state.
Consider an initial state of the form 
\begin{equation}
    | \psi_0 \rangle = \frac{1}{\sqrt{Z'_{\beta_0}}} \sum_{n} e^{- \beta_0 E'_n/2} | n\rangle_{\rm sys} | \tilde{n} \rangle_{a},
\end{equation}
which corresponds to the thermal state $ \mathrm{Tr}_{\rm a}[\Phi_0] = e^{-\beta_0 H'_s}/Z'_{\beta_0}$ for a Hamiltonian $H_s^\prime$ at some temperature $\beta_0$ after tracing out the ancilla.
We assume here that $H_s^\prime$ commutes with the target Hamiltonian $H_s$, i.e., $[H_s, H_s^\prime] = 0$, and hence this does not represent an arbitrary initial state.
Here, $H_s^\prime = \sum_{n} E_n^\prime | n\rangle\langle n |_{\rm sys}$ and $Z'_{\beta_0} = \mathrm{Tr}[e^{- \beta_0 H_s^\prime}]$ is its partition function.

Performing imaginary-time evolution on this initial state, we find 
\begin{align}
    | \Phi'(\tau) \rangle = \frac{e^{- \tau H} | \Phi_0\rangle}{\| e^{- \tau H} |\Phi_0\rangle \|} = \frac{1}{\sum_{n} e^{-2\tau E_n - \beta_0 E'_n}}\sum_{n} e^{-\tau E_n - \beta_0 E'_n/2} |n\rangle_{\rm sys} | \tilde{n}\rangle_{a} = \frac{1}{\sqrt{Z_{\rm eff}}}\sum_{n} e^{-\beta_{\rm eff} E_n^{\rm eff} / 2} |n\rangle_{\rm sys} | \tilde{n}\rangle_{a} . 
\end{align}
where $Z_{\rm eff} = \mathrm{Tr}[e^{-\beta_{\rm eff} H_{\rm eff}}]$ with $-\beta_{\rm eff} H_{\rm eff} = -2\tau H - \beta_0 H'$. 
In this case, the imaginary-time evolution state corresponds to a modified effective thermal state determined by both $H$ and $H'$; 
\begin{equation}
    \Phi'_{\rm sys}(\tau) = \frac{e^{-\beta_{\rm eff} H_{\rm eff}}}{Z_{\rm eff}}.
\end{equation}
Using this expression, we can simply find the free energy of this thermal state as 
\begin{align}
    F(\Phi'_{\rm sys}(\tau)) & = \mathrm{Tr}[\Phi'_{\rm sys}(\tau) H_s] + \frac{1}{\beta} \mathrm{Tr}[\Phi'_{\rm sys}(\tau) \log \Phi'_{\rm sys}(\tau)] \\
    & = \mathrm{Tr}[\Phi'(\tau) H] + \frac{1}{\beta} \mathrm{Tr}\left[\Phi'_{\rm sys}(\tau) \log \left( \frac{e^{-\beta_{\rm eff}H_{\rm eff}}}{Z_{\rm eff}} \right) \right] \\
    & = \mathrm{Tr}[\Phi'(\tau) H] - \frac{1}{\beta} \mathrm{Tr}[\Phi'_{\rm sys}(\tau) \beta_{\rm eff} H_{\rm eff}] - \frac{1}{\beta} \log Z_{\rm eff} \\
    & = \left( 1 - \frac{2\tau}{\beta}\right) \mathrm{Tr}[\Phi'(\tau) H] - \frac{\beta_0}{\beta} \mathrm{Tr}[\Phi'(\tau) H'] - \frac{1}{\beta} \log Z_{\rm eff} .
\end{align}
Accordingly, we can compute the derivative of the free energy with respect to $\tau$ as
\begin{align}
    \partial_\tau F(\Phi'_{\rm sys}(\tau)) & = \left( 1 - \frac{2\tau}{\beta}\right) \mathrm{Tr}[\partial_\tau\Phi'(\tau) H] - \frac{2}{\beta} \mathrm{Tr}[\Phi'(\tau) H] - \frac{\beta_0}{\beta} \mathrm{Tr}[\partial_\tau\Phi'(\tau) H'] - \frac{1}{\beta} \partial_\tau\log Z_{\rm eff} 
\end{align}
To simplify this equation, we have to compute three terms: $\mathrm{Tr}[\partial_\tau\Phi'(\tau)H]$, $\mathrm{Tr}[\partial_\tau\Phi'(\tau)H']$ and $\partial_\tau\log Z_{\rm eff}$. 
The first term can be simplified as
\begin{align}
    \mathrm{Tr}[\partial_\tau\Phi'(\tau)H] & = \mathrm{Tr}\big[ [[\Phi'(\tau), H], \Phi'(\tau)] H\big] \\
    & = -2 \left(\mathrm{Tr}[\Phi'(\tau) H^2] - \mathrm{Tr}[\Phi'(\tau)H]^2\right) \\ 
    & = -2 V(\Phi'(\tau)) \;. 
\end{align}
As for the second term, we obtain 
\begin{align}
    \mathrm{Tr}[\partial_\tau\Phi'(\tau)H'] & = \mathrm{Tr}\big[ [[\Phi'(\tau), H], \Phi'(\tau)] H'\big] \\
    & = -2 \left(\mathrm{Tr}[\Phi'(\tau) H H^\prime] - \mathrm{Tr}[\Phi'(\tau)H] \, \mathrm{Tr}[\Phi'(\tau)H^\prime]\right) \;. 
\end{align}
Lastly, the third term reads
\begin{align}
    \partial_\tau\log Z_{\rm eff} = \frac{\partial_\tau Z_{\rm eff}}{Z_{\rm eff}} = \frac{1}{Z_{\rm eff}} \mathrm{Tr}[-2 H_s e^{-\beta_{\rm eff}H_{\rm eff}}] = -2 \mathrm{Tr}[H_s \Phi'_{\rm sys}(\tau)] = - 2 E(\Phi'(\tau)) \;.
\end{align}
Therefore, combining those three terms, the derivative of the free energy is given by
\begin{align}
    \partial_\tau F(\Phi'_{\rm sys}(\tau)) & = -2\left( 1 - \frac{2\tau}{\beta}\right) V\left( \Phi'(\tau)\right) - \frac{2}{\beta} E \left( \Phi'(\tau) \right) + \frac{2\beta_0}{\beta} \left(\mathrm{Tr}[\Phi(\tau) H H^\prime] - \mathrm{Tr}[\Phi'(\tau)H] \, \mathrm{Tr}[\Phi'(\tau)H^\prime]\right) + \frac{2}{\beta} E(\Phi'(\tau)) \\
    & = -2\left( 1 - \frac{2\tau}{\beta}\right) V\left( \Phi'(\tau)\right) + \frac{2\beta_0}{\beta} \left(\mathrm{Tr}[\Phi'(\tau) H H^\prime] - \mathrm{Tr}[\Phi'(\tau)H] \, \mathrm{Tr}[\Phi'(\tau)H^\prime]\right) \;. 
\end{align}

\section{A quick overview on double-bracket quantum algorithms} \label{app:dbalgorithms}

Double-bracket flows~\cite{bloch1985completely,bloch1990steepest,bloch1992completely,Brockett1991DBF,BLOCH1985103,moore1994numerical,BROCKETT1989761,deift1983ordinary,Chu_iterations,wegner1994flow,wegner2006flow,hastings2022lieb,GlazekWilson,GlazekWilson2,kehrein_flow,smith1993geometric,optimization2012,brockett2005smooth} are matrix-valued ordinary differential equations that have been widely used for matrix diagonalization, eigenvalue sorting, and QR decomposition~\cite{optimization2012,Brockett1991DBF,deift1983ordinary,Chu_iterations}.
Originally, this class of flows was introduced in the context of minimizing a least-squares cost between two matrices using steepest descent methods.

For a symmetric matrix $A$ and $B$, a double-bracket flow takes the following form of a differential equation:
\begin{equation} 
    \frac{dA(t)}{dt} = \big[[A(t),B], A(t)\big]\;,
\end{equation} \label{app_eq:dbf_in_general}
where the evolution preserves the spectrum of $A(t)$.

A first-order discretization of this flow is given by

\begin{equation} \label{app_eq:discritize_dbf_general}
    A_{k+1}(s) = e^{s[A_k, B] }A_k e^{-s[A_k,B]},
\end{equation}
with step size $s\in\mathbb{R}$. 
This update rule reproduces the continuous double-bracket flow in Eq.~\eqref{app_eq:dbf_in_general} to first order in $s$.
Explicitly, 
\begin{equation}
    \partial_s \left(e^{s[A_k, B] }A_k e^{-s[A_k,B]}\right)\bigg\vert_{s=0} = \left([A_k, B]e^{s[A_k, B] }A_k e^{-s[A_k,B]} - e^{s[A_k, B] }A_k [A_k, B] e^{-s[A_k,B]}\right)\bigg\vert_{s=0} = \big[[A_k,B], A_{k}\big].
\end{equation}
When $A$ is chosen to represent a pure quantum state and $B$ is taken to be a Hamiltonian, the resulting double-bracket flow coincides with imaginary-time evolution~\cite{gluza2024double}.

\subsection{Double-bracket quantum imaginary-time evolution (DB-QITE)} \label{App:DB-QITE}

This section summarizes results from Ref.~\cite{gluza2024double}, which show that imaginary-time evolution is a solution to the double-bracket flow and introduce a quantum algorithm termed the Double-Bracket Quantum Imaginary-Time Evolution (\texttt{DB-QITE}) algorithm. 

Imaginary-time evolution is defined as 
\begin{equation}
    |\Phi (\tau) \rangle = \frac{e^{-\tau H }|\Phi_0\rangle}{\|e^{-\tau H }|\Phi_0\rangle\|}\;,
    \label{eq:ITE-state-def}
\end{equation}
where $H$ denotes the target Hamiltonian, $\tau > 0$ and $\Phi_0$ is the initial state.
The norm satisfies $\||\Phi\rangle\| = \sqrt{\langle \Phi | \Phi \rangle}$ for any state $| \Phi\rangle$.

A key idea underlying the \texttt{DB-QITE} is the equivalence between imaginary-time evolution and the double-bracket flow.
This relationship is formalized in the following result.
\begin{propositionA}(Imaginary-time evolution is a solution to double-bracket flow~\cite{gluza2024double})
    The imaginary-time evolution state $\Psi(\tau) = | \Psi(\tau) \rangle \langle \Psi(\tau)|$ is a unique solution to Brockett’s double-bracket flow equation, i.e.
    \begin{equation}
        \frac{\partial \Phi(\tau)}{\partial \tau} = \big[[\Phi(\tau), H], \Phi(\tau)\big]\;.
        \label{eq:ITE-DBF}
    \end{equation}
\end{propositionA}

This equivalence implies that implementing imaginary-time evolution can be recast as implementing the double-bracket flow in Eq.~\eqref{eq:ITE-DBF}. 
Following the discretization strategy introduced in Eq.~\eqref{app_eq:discritize_dbf_general}, the double-bracket flow is approximated by the discrete update rule
\begin{equation}
    | \sigma_{k+1} \rangle = e^{s_k [\sigma_k, H]} | \sigma_k \rangle\;,
    \label{eq:QITE-DBI}
\end{equation}
where $\sigma_k = \ket{\sigma_k}\bra{\sigma_k}$ denotes the state at $k$-th iteration and  $s_k$ is the corresponding step size.
The notation $\sigma_{k} $  is introduced to distinguish the discrete iteration from the imaginary-time evolution state defined in Eq.~\eqref{eq:ITE-state-def}.

Direct implementation of Eq.~\eqref{eq:QITE-DBI} on quantum hardware remains challenging, since the unitary operator $\exp(s[\sigma_k, H])$  depends explicitly on the state at the $k$-th iteration.
To overcome this difficulty, a product-formula approximation to the exponential of a commutator is employed.
In particular, the group commutator,
\begin{equation}
    G_s(\rho) = e^{i \sqrt{s}H} e^{i \sqrt{s} \rho} e^{-i \sqrt{s}H} e^{-i \sqrt{s} \rho} \;,
\end{equation}
provides a first-order approximation
\begin{equation}
    G_s(\rho) = e^{s [\rho, H]} + \mathcal{O}(s^{3/2}).
\end{equation}
This approximation leads to an alternative discrete evolution based on group commutators. Introducing an initial state $\ket{\omega_{0}}$, the update rule takes the form
\begin{equation}
    | \omega_{k+1} \rangle = e^{i \sqrt{s_k}H} e^{i \sqrt{s_k} \omega_k} e^{-i \sqrt{s_k}H} |\omega_k \rangle ,
    \label{eq:DB-QITE}
\end{equation}
where $s_{k}$  denotes the step size at iteration $k$ and $\omega_{k}=\ket{\omega_{k}}\bra{\omega_{k}}$.
Eq.~\eqref{eq:DB-QITE} defines the \texttt{DB-QITE} algorithm.
Within this framework, the quantum circuit is synthesized recursively using Hamiltonian evolutions of the form $\exp(isH)$ together with reflection operations $\exp(is\omega_{0})$ with respect to the initial state $\omega_{0}$.

\subsection{Double-bracket quantum signal processing (DB-QSP)}
\label{app:Double-bracket quantum signal processing (DB-QSP)}

This section summarizes results from Ref.~\cite{suzuki2025doublebracketalgorithmquantumsignal}, which proposes a quantum algorithm for constructing matrix-valued polynomial functions without the use of ancilla qubits or post-selection. The approach relies on double-bracket quantum algorithms and is referred to as Double-Bracket Quantum Signal Processing (\texttt{DB-QSP}).

Within the \texttt{DB-QSP} framework, the goal is to implement a normalized polynomial function of a Hermitian matrix $H$.
A key observation is that the exponential of commutators applied to an initial state $\ket{\Psi}$ can be expressed as a linear combination of the state itself and the action of $H$ on the state;
\begin{equation}
    e^{s[\Psi,H]}|\Psi\rangle = \left(\frac{E_\Psi}{\sqrt{V_\Psi}} \cos (s\sqrt V_{\Psi})+\sin(s\sqrt{V_{\Psi}})\right)|\Psi\rangle -\frac{\sin(s\sqrt{V_{\Psi}})}{\sqrt{V_{\Psi}}} H |\Psi\rangle,
\end{equation}
where $E_\Psi = \langle \Psi|H|\Psi\rangle$ denotes the energy and $V_\Psi = \langle \Psi|(H-E_\Psi)^2|\Psi\rangle$  denotes the corresponding variance.

This expression can be matched to a normalized linear polynomial action of the form
\begin{equation}
    e^{s[\Psi,H]}|\Psi\rangle = \frac{(H-\tau I)|\Psi\rangle}{||(H-\tau I)|\Psi\rangle||}
\end{equation}
with $\tau\in\mathbb{R}$.
For a given real parameter $\tau$, one can choose the evolution time $s$ such that this equality holds.
The corresponding solution is given by
\begin{equation}
   s_{\tau}=  -\frac{1}{\sqrt{V_\Psi}} \arccos \left(\frac{E_\Psi-\tau}{\sqrt{V_\Psi+(E_\Psi-\tau)^2}} \right),
\end{equation}
implying that the knowing precise step size requires the energy and variance of the state.

This construction above applies only when $\tau$ is real.
For a complex parameter $z\in\mathbb{C}$, an additional reflection with respect to the state $\ket{\Psi}$ is required.
In this case, one obtains
\begin{equation}
    e^{i\theta \Psi} e^{s[\Psi,H]} |\psi\rangle = \frac{(H-zI)|\Psi\rangle}{||(H-zI)|\Psi\rangle||}, \, z\in \mathbb{C},
\end{equation}
where 
\begin{align}
        s_z &=  -\frac{1}{\sqrt{V_\Psi}} \arccos \left( \frac{|E_\Psi-z|}{\sqrt{V_\Psi+|E_\Psi-z|^2}} \right) \\
        \theta_z &= \arg\left(\frac{E_\Psi-z}{|E_\Psi-z|}\right).
\end{align}
The discussion above applies to linear polynomials in the Hamiltonian.
Higher-order polynomial transformations can then be implemented by iterating the operations.
In particular, a polynomial of degree $k$ is obtained by applying the procedure $k$ times.

\section{Polynomial approximation of exponential functions } \label{exponential_poly_approx}
In this section, we present how to construct a polynomial approximation of  $e^{-\tau x}$ with an order proportional to $\sqrt{\tau}$. 
This construction relies on the fact that monomials $x^k$ can be well approximated by polynomials of order $\sqrt{k}$ on the interval $x \in [-1,1]$ by the use of Chebyshev polynomials.
Building on the result, the Taylor series of the exponential can be approximated using a weighted sum of the monomials.
Since approximating the exponential requires $\mathcal{O}(\tau)$ terms of the Taylor series, using this Chebyshev construction requires $\mathcal{O}(\sqrt{\tau})$.
The results shown below are based on Ref.~\cite{sachdeva2013approximationtheorydesignfast}

Before proceeding, we recall relevant properties of Chebyshev polynomials.
The $n$-th Chebyshev polynomial, denoted by $T_n(x)$, satisfies the recursion  
\begin{equation}
    T_n(x) = 2xT_{n-1}(x) -T_{n-2}(x).
\end{equation}
with the initial conditions $T_0(x) = 1$ and $T_1(x) = x$.
This recursion directly implies 
\begin{equation}
    xT_n(x) = \frac{T_{n+1}(x)+T_{n-1}(x)}{2}.
    \label{eq:cheby_mean}
\end{equation}
Eq.~\eqref{eq:cheby_mean} plays a crucial role in the proof presented later.

\subsection{Approximating a monomial using Chebyshev polynomials}

We first show that a monomial $x^k$ can be approximated using a polynomial $p_{k,l}(x)$ of degree $l \in O(\sqrt{k})$ over the interval $[-1,1]$.
This analysis focuses on a uniform approximation, and therefore on bounding 
\begin{equation} \label{eq:bound_monomial}
    \sup_{x\in[-1,1]} |x^k-p_{k,l}(x)|.
\end{equation}

Let $Y$ be a random binary variable taking the value 1 and $-1$ with equal probability, i.e., $\mathbb{P}(Y=1)=\mathbb{P}(Y=1)=1/2$.
Then, the identity in Eq.~\eqref{eq:cheby_mean} can be rewritten as 
\begin{equation}
    xT_n(x) = \mathbb{E}_Y[T_{n+Y}(x)].
\end{equation}
Moreover, by introducing the sum $n_k = |\sum_{i=1}^k Y_i|$ using $k$ i.i.d. random binary variables, $Y_1,Y_2,\dots, Y_k$, the monomial $x^k$ can be expressed as
\begin{equation} \label{eq:exact_monomial_rb}
    x^k = \mathbb{E}_{\bf{Y}} [T_{n_k}(x)].
\end{equation}
Here, $n_0 = 0$ by definition. See Lemma~3.1 of Ref.~\cite{sachdeva2013approximationtheorydesignfast} for the detailed proof.
Within this setting, we define our polynomial approximation as a truncation of Eq.~\eqref{eq:exact_monomial_rb} where Chebyshev polynomials for $n_k \geq l$ are ignored, i.e., 
\begin{equation} \label{eq:mono_app_def}
    p_{k,l}(x) = \mathbb{E}_{\bf{Y}} [T_{n_k}(x) \mathbbm{1}_{n_k \leq l}].
\end{equation}
The intuition of the approximation comes from the fact that the order of the Chebyshev polynomial also corresponds to the leading power; $p_{k,l}$ is of order $l$.

An advantage of the expression above is that the approximation error can be controlled using tools from probability theory.
In particular, a Chernoff bound shows that contributions from polynomial terms of degree larger than $\mathcal{O}(\sqrt{k})$ can be significantly suppressed.
First, we obtain an equality regarding the variable $n_k$;
\begin{equation}
    \mathbb{P}_{\mathbf{Y}}[n_k > l] = \mathbb{P_{\mathbf{Y}}}\left[\left|\sum_i^k Y_i\right| >l \right] = 2\mathbb{P_{\mathbf{Y}}}\left[\sum_i^k Y_i >l\right],
\end{equation}
where the last equality follows from the symmetry between positive and negative values.
Then, using the Chernoff bound, we get
\begin{equation}
    \mathbb{P_{\mathbf{Y}}}\left[\sum_i^k Y_i >l\right] \leq \inf_{t>0} \frac{\mathbb{E}[e^{t\sum_i^k Y_i}]}{e^{tl}} = \inf_{t>0} \frac{\mathbb{E}[e^{tY}]^k}{e^{tl}},
\end{equation}
Here, we also utilize the fact that $Y$s are i.i.d. random variables.
The moment generating function of $Y$, i.e., $\mathbb{E}[e^{tY}]$, is given by $\cosh(t)$, which can be upper bounded by $e^{t^2/2}$.
Hence, we have 
\begin{equation}
    \mathbb{P_{\mathbf{Y}}}\left[\sum_i^k Y_i >l\right]  \leq \inf_{t>0}e^{\frac{kt^2}{2}-tl}.
\end{equation}
Optimizing over $t$, we also get $t = \frac{l}{k}$, implying $\inf_{t>0}e^{\frac{kt^2}{2}-tl} = e^{-\frac{l^2}{2k}}$. 
Then the final probability is upper bounded by:
\begin{equation} \label{eq:chernoff_for_monomial_construction}
    \mathbb{P_{\mathbf{Y}}}\left[\left|\sum_i^k Y_i\right| >l\right]  \leq 2e^{-\frac{l^2}{2k}}.
\end{equation}
We are now in a good position to bound the error coming from approximating $x^k$ with $p_{k,l}(x)$:

\begin{equation}
\begin{split}
    \sup_{x\in[-1,1]} |x^k -p_{k,l}(x)| &= \sup_{x\in[-1,1]} |\mathbb{E}_{\bf{Y}} [T_{n_k}(x)] -\mathbb{E}_{\bf{Y}} [T_{n_k}(x)\mathbbm{1}_{n_k \leq l}]| \\
    & = \sup_{x\in[-1,1]}|\mathbb{E}_{\bf{Y}} [T_{n_k}(x)\mathbbm{1}_{n_k >l}]|\\
    & \leq \sup_{x\in[-1,1]} \mathbb{E}_{\mathbf{Y}}[|T_{n_k}(x)| \mathbbm{1}_{n_k > l}] \\
    & \leq \mathbb{E}_{\mathbf{Y}} [\sup_{x\in[-1,1]}|T_{n_k}(x)| \mathbbm{1}_{n_k > l}] \\
    & \leq \mathbb{E}_{\mathbf{Y}}[\mathbbm{1}_{n_k > l}] \\
    & = \mathbb{P}_{\mathbf{Y}}[\mathbbm{1}_{n_k > l}] \\
    & \leq 2e^{-\frac{l^2}{2k}},
\end{split}
\end{equation}
where we used the fact that the $\sup_{x\in[-1,1]}|T_n(x)| \leq 1$ and $
\mathbb{E}[\mathbbm{1}_S] = \mathbb{P}[S]$.
Also, Eq.~\eqref{eq:chernoff_for_monomial_construction} is used for the last inequality. 

Consequently, for Eq.~\eqref{eq:bound_monomial} to be less that $\epsilon$ uniformly over $x\in[-1,1]$, we need
\begin{equation}
    2e^{-\frac{l^2}{2k}} \leq \varepsilon \implies l \geq \sqrt{2k\log\left(\frac{2}{\varepsilon}\right)},
\end{equation}
which means that it is sufficient to have the order of the polynomial to be $\mathcal{O}(\sqrt{k})$.

\subsection{Approximating exponential functions using Chebyshev polynomials}

We now approximate $\exp({-x})$ via polynomials of order $\mathcal{O}(\sqrt{\beta})$ on the interval $x\in [0,\beta]$.
Specifically, we approximate the Taylor series of the exponential using the monomial approximation outlined above, which can reduce the required polynomial degree from $\mathcal{O}(\beta)$ to $\mathcal{O}(\sqrt{\beta})$.

First, the interval of approximation must be matched.
The monomial approximation is constructed for $x \in [-1,1]$, whereas the present setting involves  $x \in [0,\beta]$.
To reconcile this mismatch, we rescale the exponential $\exp(-x) \to \exp(-\lambda -\lambda x)$ with $\lambda=\beta/2$, so that the domain $x \in [-1,1]$ is mapped onto the target interval $[0,\beta]$ when appearing in the exponent.

The Taylor expansion of $e^{-\lambda-\lambda x}$ around 0 gives
\begin{equation} \label{eq:taylor_exponential}
    e^{-\lambda-\lambda x} = e^{-\lambda} \sum_{i=0}^{\infty} \frac{(-\lambda)^i}{i!} x^i.
\end{equation}
Thus, we construct a polynomial approximation $q_{\lambda, k,l}$ by truncating Eq.~\eqref{eq:taylor_exponential} up to the order $k$, i.e., 
\begin{equation}
    q_{\lambda, k, l}(x) = e^{-\lambda}\sum_{i=0}^{k} \frac{(-\lambda)^i}{i!}p_{i,l}(x),
\end{equation}
where $p_{i,l}(x)$ is a approximation of the monomial $x^{k}$ shown in Eq.~\eqref{eq:mono_app_def}. 

We now show the error bound on the approximation.
On the interval of $x\in[-1,1]$, we have
\begin{equation} \label{eq:bound_exp_1}
\begin{split}
        \sup_{x\in[-1,1]} |e^{-\lambda-\lambda x} - q_{\lambda, k, l}(x)| &= \sup_{x\in[-1,1]} \left| e^{-\lambda}\sum_{i=0}^{\infty} \frac{(-\lambda)^i}{i!} x^i - e^{-\lambda}\sum_{i=0}^{k} \frac{(-\lambda)^i}{i!} p_{i,l}(x)\right| \\
        &= \sup_{x\in[-1,1]}\left|e^{-\lambda} \sum_{i=0}^k\frac{(-\lambda)^i}{i!}(x^i -p_{i,l}(x)) + e^{-\lambda}\sum_{i=k+1}^{\infty} \frac{(-\lambda)^i}{i!}x^i\right| \\
        &\leq \sup_{x\in[-1,1]} \left|e^{-\lambda} \sum_{i=0}^k\frac{(-\lambda)^i}{i!}(x^i -p_{i,l}(x))\right| +\left|e^{-\lambda}\sum_{i=k+1}^{\infty} \frac{(-\lambda)^i}{i!}x^i\right|.
\end{split}
\end{equation}
We now give an upper bound of each term in Eq.~\eqref{eq:bound_exp_1}.
\begin{itemize}
    \item \textbf{First term --} 
    We use the bound for approximating monomials to obtain the first term:
    \begin{equation}
    \begin{split}
    \sup_{x \in [-1,1]} \left|e^{-\lambda} \sum_{i=0}^{k} \frac{(-\lambda)^i}{i!} (x^i-p_{i,l}(x))\right| & \leq e^{-\lambda} \sum_{i=0}^k \frac{\lambda^i}{i!} \sup_{x\in[-1,1]} |x^i -p_{i,l}(x)| \\
    & \leq e^{-\lambda} \sum_{i=0}^k \frac{\lambda^i}{i!} 2e^{-\frac{l^2}{2i}} \\
    & \leq 2e^{-\frac{l^2}{2k}} e^{-\lambda} \sum_{i=0}^{k} \frac{\lambda^i}{i!} \\
    & \leq 2e^{-\frac{l^2}{2k}} e^{-\lambda} \sum_{i=0}^{\infty} \frac{\lambda^i}{i} \\
    & = 2e^{-\frac{l^2}{2k}} e^{-\lambda+\lambda} \\
    & = 2e^{-\frac{l^2}{2k}}.
\end{split}
\end{equation}    
    \item \textbf{Second term --}
    The tail of the Taylor series can be bounded as
    \begin{equation}
    \sup_{x\in[-1,1]}\left|e^{-\lambda} \sum_{i=k+1}^\infty \frac{(-\lambda)^i}{i!}x^i\right| \leq e^{-\lambda} \sum_{i=k+1}^\infty \frac{\lambda^i}{i!}.
    \end{equation}
    Since $i! \geq (\frac{i}{e})^i$, we obtain
    \begin{equation}
e^{-\lambda} \sum_{i=k+1}^\infty \frac{\lambda^i}{i!} \leq e^{-\lambda} \sum_{i=k+1}^\infty \left(\frac{\lambda e}{i}\right)^i \leq e^{-\lambda} \sum_{i=k+1}^\infty \left(\frac{\lambda e}{k}\right)^i.
    \end{equation}
    We then assume  $\frac{\lambda e}{k} \leq \frac{1}{e}$ to get
    \begin{equation}
e^{-\lambda} \sum_{i=k+1}^\infty \left(\frac{\lambda e}{k}\right)^i \leq e^{-\lambda} \sum_{i=k+1}^\infty e^{-i}.
    \end{equation}
    Therefore, the geometric series can be bounded as
    \begin{equation}
    \sup_{x\in[-1,1]}\left|e^{-\lambda} \sum_{i=k+1}^\infty \frac{(-\lambda)^i}{i!}x^i\right| \leq e^{-(\lambda+k)}
    \end{equation}
\end{itemize}
By combining the bounds together, we finally get
\begin{equation}
    \sup_{x\in[-1,1]} |e^{-\lambda-\lambda x} - q_{\lambda, k, l}(x)| \leq 2e^{-\frac{l^2}{2k}} + e^{-(\lambda+k)}.
    \label{eq:chebyshev_error_bound}
\end{equation}
This indicates that
\begin{equation}
        2e^{-\frac{l^2}{2k}} \leq \frac{\varepsilon}{2} \implies l \geq \sqrt{2k\log\left(\frac{4}{\varepsilon}\right)}
\end{equation}
and
\begin{equation}
    e^{-(\lambda+k)} \leq \frac{\varepsilon}{2} \implies k \geq \log(1/\varepsilon),
\end{equation}
because $\lambda = \beta/2 >0$. 
However, we recall that we assume $\frac{\lambda e}{k} \leq \frac{1}{e}$ to derive the bound for the second term in Eq.~\eqref{eq:bound_exp_1}.
Thus, the bound satisfies $k \geq \max[\lambda e^2, \log(\frac{2}{\varepsilon})]$.
For instance, when $\lambda e^2$ is the maximum, we can see that the order $l$ required for approximation is given by $e\sqrt{2\lambda \log(\frac{4}{\varepsilon})} = e\sqrt{\beta \log(\frac{4}{\varepsilon})} \in \mathcal{O}(\sqrt{\beta})$.

\section{Recursion steps $k$ required for DB-TFD}

\subsection{Vanilla DB approach: Proof of Lemma ~\ref{theorem:naive_error}} \label{Naive_DBI_perfomance}

We show the number of steps $k$ required for the \texttt{vanilla} \texttt{DB-TFD} state $ |\psi_{k}^{\texttt{(vanilla)}}\rangle$, corresponding to a perfect implementation of the exponential of the commutator, to approximate the thermofield double state at inverse temperature $\beta/2$ with $\varepsilon $-precision.
To this end, we first quantify the error introduced by the approximation. 
Note that, although our bound is slightly tighter in terms of the prefactor, the overall structure of the proof follows that of Ref.~\cite{mcmahon2025equatingquantumimaginarytime}.

Recall the definition of the \texttt{vanilla} \texttt{DB-TFD} state at step $k$,
\begin{equation}
\begin{split} \label{app_eq:naive_DB_TFD}
    |\psi_{j}^{\texttt{(vanilla)} } \rangle &:= \prod_{i=0} ^{j-1} e^{s_i [\psi_i^{\texttt{(vanilla)}}, H]} |\rm{TFD}(0)\rangle,
\end{split}
\end{equation}
with the initial condition $\ket{\psi_0^{\texttt{(vanilla)}}} = |\rm{TFD}(0)\rangle$.
Here, the index $j$ runs from 0 to $k$.
Throughout the analysis, the step size is chosen to be constant, $s_i=\beta/2k=:s$, for all steps.
To compare the evolution with the dynamics of the thermofield double state at matching time intervals, we also define the segmented thermofield double states as
\begin{equation}
      |\psi_{j}^{\rm{(ITE)}}  \rangle :=  |{\rm{TFD}}(j\beta/2k)\rangle = \frac{e^{-\tau_{j-1} H} |\psi_{j-1}^{\rm{(ITE)}} \rangle}{\|e^{-\tau_{j-1} H} |\psi_{j-1}^{\rm{(ITE)}} \rangle\|},
\end{equation}
where $\tau_{j}=\beta/2k=:\tau$ for all $j$.
The goal is then to bound the total error after $k$ steps, i.e., 
\begin{equation}
    \Delta_{k}:=\| |\psi_{k}^{\texttt{(vanilla)}}\rangle - |\rm{TFD}(\beta)\rangle \|. 
\end{equation}

Firstly, the gap $\Delta_{j+1}=\|  |\psi_{j+1}^{\texttt{(vanilla)} } \rangle - |\psi_{j+1}^{\rm{(ITE)} } \rangle\|$ can be bounded using the triangle inequality as
\begin{equation}
\begin{split}
     \Delta_{j+1} &\leq \left\| |\psi_{j+1}^{\texttt{(vanilla)} }\rangle  - e^{s[\psi_{j}^{\texttt{(vanilla)} },H]} |\psi_{j}^{\rm{(ITE)} }  \rangle \right\| \\&+\left\|e^{s[\psi_{j}^{\texttt{(vanilla)} },H]} |\psi_{j}^{\rm{(ITE)} }  \rangle-e^{s[\psi_{j}^{\rm{(ITE)} },H]} |\psi_{j}^{\rm{(ITE)} }  \rangle\right\| \\
     &+ \left\| e^{s[\psi_{j}^{\rm{(ITE)} },H]} |\psi_{j}^{\rm{(ITE)} }  \rangle -|\psi_{j+1}^{\rm{(ITE)} } \rangle \right\| ,
     \label{eq:naive_triangle}
     \end{split}
\end{equation}
where the triangle inequality has been applied.
Each of the three terms on the right-hand side can be bounded separately as follows.

\begin{enumerate}
\item
The first term  can be directly bounded as
\begin{equation}
\begin{split}
      \left\| |\psi_{j+1}^{\texttt{(vanilla)} }\rangle  - e^{s[\psi_{j}^{\texttt{(vanilla)} },H]} |\psi_{j}^{\rm{(ITE)} }  \rangle \right\| & = \left\| e^{s[\psi_{j}^{\texttt{(vanilla)} },H]} |\psi_{j}^{\texttt{(vanilla)} }\rangle  - e^{s[\psi_{j}^{\texttt{(vanilla)} },H]} |\psi_{j}^{\rm{(ITE)} }  \rangle \right\| \\
      &\leq \left\|e^{s[\psi_{j}^{\texttt{(vanilla)} },H]}\right\|_{{\rm op}} \left\|  |\psi_{j}^{\texttt{(vanilla)} } \rangle - |\psi_{j}^{\rm{(ITE)} } \rangle \right\| \\
      &= \Delta_j,
\end{split}
\end{equation}
where we use the unitarity of the exponential of the commutator, implying $\|e^{s[\psi_{j}^{\texttt{(vanilla)} },H]}\|_{{\rm op}}=1 $.
\item As for the second term, we obtain
\begin{equation}
\begin{split}
    \left\|e^{s[\psi_{j}^{\texttt{(vanilla)} },H]} |\psi_{j}^{\rm{(ITE)} }  \rangle-e^{s[\psi_{j}^{\rm{(ITE)} },H]} |\psi_{j}^{\rm{(ITE)} }  \rangle\right\|  &\leq \left\|e^{s[\psi_{j}^{\texttt{(vanilla)} },H]} -e^{s[\psi_{j}^{\rm{(ITE)} },H]}\right\|_{{\rm op}} \left\||\psi_{j}^{\rm{(ITE)} }  \rangle\right\| \\
    &=  \left\|e^{s[\psi_{j}^{\texttt{(vanilla)} },H]} -e^{s[\psi_{j}^{\rm{(ITE)} },H]}\right\|_{{\rm op}},
\end{split}
\end{equation}
where $\||\psi_{j}^{\rm{(ITE)} }  \rangle\|=1$. 
Since  
\begin{equation}
    \|e^A-e^B\|_{{\rm op}} \leq \|A-B\|_{{\rm op}},
\end{equation}
in case $e ^A$ and $e^B$ are unitaries, we get
\begin{equation}
\begin{split}
    \left\|e^{s[\psi_{j}^{\texttt{(vanilla)} },H]} -e^{s[\psi_{j}^{\rm{(ITE)} },H]}\right\|_{{\rm op}} &\leq \left\|s [\psi_{j}^{\texttt{(vanilla)} }- \psi_{j}^{\rm{(ITE)} },H]\right\|_{{\rm op}}\\ 
    &\leq 2s \|H\|_{{\rm op}}\left\| \psi_{j}^{\texttt{(vanilla)} }- \psi_{j}^{\rm{(ITE)} }\right\|_{{\rm op}} \\ &
    \leq 2s ||H||_{op} \Delta_j,
\end{split}
\end{equation} 
where we use the inequality $\| \psi - \phi\|_{{\rm op}} \leq \||\psi\rangle-|\phi\rangle\|$ for pure states in the last line.
\item To simplify the last term, we exploit the Lagrange remainder theorem for both operations; namely, we have the expressions
\begin{equation}
    e^{s[\psi_{j}^{\rm{(ITE)} }  ,H]}|\psi_{j}^{\rm{(ITE)} }  \rangle = (\mathds{1} + \partial_s e^{s[\psi_{j}^{\rm{(ITE)} }  ,H]} \rvert_{s=0} + 
    R_{\text{DB}}) |\psi_{j}^{\rm{(ITE)} }  \rangle,
\end{equation}
and 
\begin{equation}
    |\psi_{j+1}^{\rm{(ITE)} } \rangle = \frac{e^{-sH}|\psi_{j}^{\rm{(ITE)} }  \rangle}{\left\|e^{-sH}|\psi_{j}^{\rm{(ITE)} }  \rangle\right\|} = \left(\mathds{1} +\partial_s \left(\frac{e^{-sH}}{\left\|e^{-sH}|\psi_{j}^{\rm{(ITE)} }  \rangle\right\|} \right) \Bigg\rvert_{s=0} + R_{\text{ITE}}\right)|\psi_{j}^{\rm{(ITE)} } \rangle.
\end{equation}
The first-order derivatives for both cases coincide at $s=0$, 
\begin{equation}
    \partial_s \left(e^{s[\psi_{j}^{\rm{(ITE)} },H]} \right) \Big\rvert_{s=0} =  \partial_s \left( \frac{e^{-sH}|\psi_{j}^{\rm{(ITE)} }\rangle}{\left\|e^{-sH}|\psi_{j}^{\rm{(ITE)} }\rangle\right\|} \right) \Bigg\rvert_{s=0} = \left(E_{j}^{\rm{(ITE)} } -H\right) |\psi_{j}^{\rm{(ITE)} }\rangle,
\end{equation}
where $E_{j}^{\rm{(ITE)} } = \braket{\psi_{j}^{\rm{(ITE)} }|H|\psi_{j}^{\rm{(ITE)} }}$.
This thus implies that 
\begin{equation}
    \|e^{s[\psi_{j}^{\rm{(ITE)} } ,H]} |\psi_{j}^{\rm{(ITE)} }  \rangle -|\psi_{j}^{\rm{(ITE)} }\rangle \| = \|(R_{\text{DB}}+R_{\text{ITE}})|\psi_{j}^{\rm{(ITE)} }  \rangle \| \leq \|R_{\text{DB}}|\psi_{j}^{\rm{(ITE)} }  \rangle \|+\|R_{\text{ITE}}|\psi_{j}^{\rm{(ITE)} }  \rangle \|.
\end{equation}

The remainder from the exponential of the commutator can be rewritten as
\begin{equation}
    R_{\text{DB}}  = \frac{s^2}{2}[\psi_{j}^{\rm{(ITE)}},H]^2 e^{\xi[\psi_{j}^{\rm{(ITE)}},H]}
\end{equation}
for $\xi\in[0,s]$.
Hence, we obtain the bound
\begin{equation}
\begin{split}
     \|R_{\text{DB}} \ket{\psi_{j}^{\rm{(ITE)}}}  \| 
     &=  \frac{s^2}{2}\left\| [\psi_{j}^{\rm{(ITE)}},H]^2e^{\xi[\psi_{j}^{\rm{(ITE)}},H]}\ket{\psi_{j}^{\rm{(ITE)}}}\right\| \\
     &\le  \frac{s^2}{2}\left\| [\psi_{j}^{\rm{(ITE)}},H]^2\right\|_{{\rm op}}  \\
     & \le 2s^2 \|H\|_{{\rm op}}^2,
\end{split}
\end{equation}
where we use $\| [\psi_{j}^{\rm{(ITE)}},H]^2\|_{{\rm op}} = \| [\psi_{j}^{\rm{(ITE)}},H]\|_{{\rm op}}^2 \le  (2\|H\|_{{\rm op}})^2 $ using $\|\psi_{j}^{\rm{(ITE)}}\|_{{\rm op}}=1$.

Similarly, the remainder for imaginary-time evolution dynamics can be expressed as
\begin{equation}
\begin{split}
    \partial_s^2 (|\psi_{j+1}^{\rm{(ITE)} }\rangle )
    &= \partial_{s} \left(  \left(E_{j+1}^{\rm{(ITE)} } -H\right) |\psi_{j+1}^{\rm{(ITE)} }\rangle \right) \\
    &= \partial_{s} \left(E_{j+1}^{\rm{(ITE)} }\right) |\psi_{j+1}^{\rm{(ITE)} }\rangle +  \left(E_{j+1}^{\rm{(ITE)} } -H\right) \partial_{s} \left(|\psi_{j+1}^{\rm{(ITE)} }\rangle\right) \\
    & = \left( -2 \left(\braket{\psi_{j+1}^{\rm{(ITE)} }|H^2|\psi_{j+1}^{\rm{(ITE)} }} -{E_{j+1}^{\rm{(ITE)} }}^2\right) + \left(E_{j+1}^{\rm{(ITE)} } -H\right)^2\right)  |\psi_{j+1}^{\rm{(ITE)} }\rangle \\
    & = \left(-2V_{j+1}^{\rm{(ITE)}  } + \left(E_{j+1}^{\rm{(ITE)} } -H\right)^2 \right) |\psi_{j+1}^{\rm{(ITE)} }\rangle, 
\end{split}
\end{equation}
with $V_{j+1}^{\rm{(ITE)}  } = \braket{\psi_{j+1}^{\rm{(ITE)} }|H^2|\psi_{j+1}^{\rm{(ITE)} }} -{E_{j+1}^{\rm{(ITE)} }}^2 $.
In the last equality, we utilize
\begin{equation}
\begin{split}
    \partial_{s} \left(E_{j+1}^{\rm{(ITE)} }\right) &= \partial_{s} \left(\bra{\psi_{j+1}^{\rm{(ITE)} }}\right) H \ket{\psi_{j+1}^{\rm{(ITE)} }} + \bra{\psi_{j+1}^{\rm{(ITE)} }}H \partial_s \left(\psi_{j+1}^{\rm{(ITE)} }\right) \\
    &= -2  \bra{\psi_{j+1}^{\rm{(ITE)} }}\left(H^2 -  E_{j+1}^{\rm{(ITE)} } H\right)  \ket{\psi_{j+1}^{\rm{(ITE)} }} \\
    &= -2 V_{j+1}^{\rm{(ITE)}  } .
\end{split}
\end{equation}
Therefore, we obtain
\begin{equation}
\begin{split}
    \|R_{\text{ITE}}|\psi_{j}^{\rm{(ITE)} }  \rangle \| &\le \frac{s^2}{2} \left\|   \partial_s^2 (|\psi_{j+1}^{\rm{(ITE)} }\rangle )\Big\rvert_{s=\xi}  \right\| \\
    & \le \frac{s^2}{2} \left\| \left(-2V_{j+1}^{\rm{(ITE)}  } + \left(E_{j+1}^{\rm{(ITE)} } -H\right)^2 \right) \Bigg\rvert_{s=\xi}  \right\|_{{\rm op}} \\
     & \le \frac{s^2}{2}  \left(2 \|V_{j+1}^{\rm{(ITE)} }\rvert_{s=\xi}\|_{{\rm op}} + \left\|\left(H^2 -2 E_{j+1}^{\rm{(ITE)} } H + {E_{j+1}^{\rm{(ITE)} }}^2 \right)\bigg\rvert_{s=\xi}\right\|_{{\rm op}} \right) \\
     & \le  \frac{s^2}{2}  \cdot 8 \|H\|_{{\rm op}}^2 \\
     &= 4s^2 \|H\|_{{\rm op}}^2,
\end{split}
\end{equation}
where we use $E_{j+1}^{\rm{(ITE)} }\rvert_{s=\xi} \le \|H\|_{{\rm op}}$ and $V_{j+1}^{\rm{(ITE)} }\rvert_{s=\xi} \le \|H\|_{{\rm op}}^2 + \|E_{j+1}^{\rm{(ITE)} }\rvert_{s=\xi} \|_{{\rm op}} \|H\|_{{\rm op}} \le 2 \|H\|_{{\rm op}}^2 $ in the fourth line.
Consequently, the third term is bounded as
\begin{equation}
     \left\| e^{s[\psi_{j}^{\rm{(ITE)} },H]} |\psi_{j}^{\rm{(ITE)} }  \rangle -|\psi_{j+1}^{\rm{(ITE)} } \rangle \right\| \le 2s^2 \|H\|_{{\rm op}}^2 + 4s^2 \|H\|_{{\rm op}}^2 = 6s^2 \|H\|_{{\rm op}}^2.
\end{equation}
\end{enumerate}
By combining three terms, Eq.~\eqref{eq:naive_triangle} reads
\begin{equation} \label{app_eq:sequence_error_naive}
    \Delta_{j+1} \le (1+2s ||H||_{op} )\Delta_{j} +  6s^2 \|H\|_{{\rm op}}^2.
\end{equation}
Since Eq.~\eqref{app_eq:sequence_error_naive} is of the form $\Delta_j \leq a \Delta_{j-1}+b$, this can be solved as
\begin{equation}
\begin{split}
    \Delta_j &\leq a^j \Delta_0+ (1+a+a^2+\ldots+a^{j-1})b\\
    &= a^j \Delta_0 + \frac{1-a^j}{1-a} b  \\
    &= (1+2s\|H\|_{{\rm op}})^j \Delta_0 + \frac{1-(1+2s\|H\|_{{\rm op}})^j}{1-(1+2s\|H\|_{op})}  6s^2 \|H\|^2_{{\rm op}} \\
    & =(1+2s\|H\|_{{\rm op}})^j \Delta_0 + 3s\|H\|_{{\rm op}}((1+2s\|H\|_{{\rm op}})^j -1).
    \end{split}
\end{equation}
Without loss of generality, we can assume  $\Delta_0 =0$ since both cases start with $\ket{\text{TFD}(0)}$.
In addition, since $(1+x)^j \leq e^{jx}$ and $(e^{jx}-1) \leq e^{jx}$ for $x\in \mathbb{R}_{+}$, we have
\begin{equation}
    2x((1+x)^j-1) \leq 2x(e^{jx} -1) \leq 2xe^{jx} = 3s\|H\|_{{\rm op}} e^{2sj\|H\|_{{\rm op}}} .
\end{equation}
Therefore, for $j=k$ and $s=\beta/2k$, we get
\begin{equation} \label{eq:noise_bound_naive}
    \Delta_k \leq \frac{3\beta \|H\|_{{\rm op}}}{2k} e^{\beta \|H\|_{{\rm op}}}.
\end{equation}

Finally, we show the number of recursive steps $k$ required to achieve 
\begin{equation}
    \| |\psi_{k}^{\texttt{(vanilla)}}\rangle - |\rm{TFD}(\beta)\rangle \| \leq \varepsilon . 
\end{equation}
Using Eq.~\eqref{eq:noise_bound_naive}, we have $\Delta_{k}\le \varepsilon$, indicating that 
\begin{equation}
    k \geq \frac{3\beta \|H\|_{{\rm op}}}{2\varepsilon} e^{\beta \|H\|_{{\rm op}}}.
\end{equation}

\subsection{Poly DB approach: Proof of Lemma~\ref{theorem:poly_error}} \label{QSP_perfomance}

We show the polynomial degree $k$ required for the \texttt{poly} \texttt{DB-TFD} state $ |\psi_{k}^{\rm{(poly)}}\rangle$, corresponding to a perfect implementation of the exponential of the commutator, to approximate the thermofield double state at inverse temperature $\beta/2$ with $\varepsilon $-precision.
Similar to the case for \texttt{vanilla} \texttt{DB-TFD}, we first estimate the error incurred by the approximation.

To evaluate the error, we again consider the norm of the difference between two states
\begin{equation}
    \left\| |\psi_{k}^{\rm{(poly)}}\rangle - |{\rm{TFD}}(\beta/2)\rangle \right\| = 
    \left\|\frac{p_{k}(H) |{\rm{TFD}}(0)\rangle}{\|p_{k}(H) |{\rm{TFD}}(0)\rangle\|} - \frac{e^{-\tau H}|{\rm{TFD}}(0)\rangle}{||e^{-\tau H}|{\rm{TFD}}(0)\rangle||}\right\|,
\end{equation}
where $p(H)=c \prod_{k=1}^{K} (H-z_{k}I)$ denotes a degree-$k$ polynomial for $c,z_{k} \in \mathbb{C}$.
We then observe that, for operators  $A$ and $B$ and an arbitrary pure state $\ket{\psi}$,
\begin{equation}
    \begin{split}
        \left\|\frac{A|\psi\rangle}{||A |\psi\rangle||} - \frac{B|\psi\rangle}{||B|\psi\rangle||}\right\|
        & = \left\|\frac{A|\psi\rangle}{||A |\psi\rangle||} -\frac{B|\psi\rangle}{||A|\psi\rangle||}+\frac{B|\psi\rangle}{||A|\psi\rangle||}- \frac{B|\psi\rangle}{||B|\psi\rangle||}\right\|\\
        & \leq \left\|\frac{A|\psi\rangle}{||A |\psi\rangle||} -\frac{B|\psi\rangle}{||A|\psi\rangle||}\right\|+\left\|\frac{B|\psi\rangle}{||A|\psi\rangle||}- \frac{B|\psi\rangle}{||B|\psi\rangle||}\right\| \\
        &= \frac{||(A-B)|\psi\rangle ||}{||A|\psi\rangle||}+ ||B|\psi\rangle|| \left\|\frac{||B|\psi\rangle ||-||A|\psi\rangle||||}{||A|\psi\rangle ||||B|\psi\rangle|| } \right\| \\
        &\leq 2\frac{||(A-B)|\psi\rangle||}{||A|\psi\rangle||}.
    \end{split}
\end{equation}
This result shows that, to bound the error, it suffices to consider the operator difference $A-B$ and the norm of one of the operators, e.g., $\|A\ket{\psi}\|$.

\begin{itemize}
    \item \textbf{Operator difference:}
    To bound the operator difference, we employ the exponential approximation $p_{k}(H) =q_{\lambda,k,l}(H)$ presented in App.~\ref{exponential_poly_approx}.
    The Chebyshev polynomial approximation requires that the input lie within the interval $x \in [-1,1]$.
    To satisfy this condition, we shift and rescale the Hamiltonian as $H \rightarrow aH- bI = H'$, with  $a = \frac{2}{|E_{max}-E_0|} = \frac{2}{\Delta E}$ and $b = \frac{E_{max}+E_0}{\Delta E}$, so that the spectrum of $H'$ lies in $[-1,1]$.

    We note that this approach requires knowledge of the maximum and minimum energies, which may not be feasible in practice. In practical situations, upper and lower bounds on the spectrum can be used to ensure that the polynomial approximation remains valid.

    \item \textbf{Norm of the state under operator action: }
   We consider the norm of the imaginary-time evolution operator applied to the initial state, $ \|e^{-\tau H} |\text{TFD(0)}\rangle \|$.
   Since applying imaginary-time evolution to a thermofield double state produces another unnormalized thermofield double state, we have
   \begin{equation}
    \|e^{-\tau H} |\text{TFD(0)}\rangle \| = \frac{1}{\sqrt{D} } \left\|\sum_n e^{-\tau E_n} |n\rangle_{sys}|\tilde{n}\rangle_{a}\right\| = \sqrt{\frac{Z(2\tau)}{D}} = \sqrt{\frac{Z(\beta)}{D}},
    \end{equation}
    where $Z(\beta)$ denotes the partition function and $D=2^{n}$ is the Hilbert space dimension for a system of $n$ qubits.
\end{itemize}

Combining these results together, we can obtain the bound
\begin{equation}
    \left\|\frac{q_{\beta/4,k,l}(H') |\text{TFD(0)}\rangle}{\|q_{\beta/4,k,l}(H')  |\text{TFD(0)}\rangle\|} - \frac{e^{-\beta( H'+I)/2}|\text{TFD(0)}\rangle}{\|e^{-\beta (H'+I)/2}|\text{TFD(0)}\rangle\|}\right\| \leq 2\sqrt{\frac{D}{Z'(\beta)}}\left(2e^{-\frac{l^2}{2k}}+e^{-\left(\frac{\beta}{4}+k\right)}\right),
\end{equation}
where we assume $k\ge \beta e^2/4$ as in App.~\ref{exponential_poly_approx}.
Here, $Z'(\beta)$ denotes the partition function corresponding to $H'+I$.

By repeating the analysis in App.~\ref{exponential_poly_approx}, we can derive the necessary conditions for the order of the polynomial $k$.
To achieve a target precision $\varepsilon$, it suffices to choose
\begin{equation}
    k \geq \max \left( \frac{\beta e^2}{4}, \log\left(\frac{4}{\varepsilon} \sqrt{\frac{D}{Z'(\beta)}}\right) \right), \,\, l \geq\sqrt{2k \log\left(\frac{8}{\varepsilon} \sqrt{\frac{D}{Z'(\beta)}}\right)}.
\end{equation}
Moreover, assuming $\beta e^2/4 \geq  \log\left(4/\varepsilon \sqrt{D/Z'(\beta)}\right)$ for the bound of $k$, we can further simplify $l$ as $$l \geq e\sqrt{\frac{\beta}{2}\log\left(\frac{8}{\varepsilon} \sqrt{\frac{D}{Z'(\beta)}}\right)}.$$

We note that a thermal state of the rescaled Hamiltonian $H'$ at inverse temperature $\beta$ corresponds to a thermal state of the original Hamiltonian $H$ at inverse temperature $a\beta$.
Therefore, to prepare a thermal state of $H$, the evolution under $H'$ must reach the effective temperature $\beta/a$.
Since $a$ is typically less than 1, this requires a longer evolution time.
If the estimated bounds on the energies $E_{max}$ and $E_0$ are too loose, $a$ becomes small, resulting in an excessively long evolution time.

Specifically, assuming exact normalization, we have $\beta/a = \beta \Delta E/2$, which leads to the following lower bound on the number of iterations
\begin{equation}
    l \geq  \frac{e}{2}\sqrt{\beta \Delta E\log\left(\frac{8}{\varepsilon} \sqrt{\frac{D}{Z'\left(\frac{\beta\Delta E}{2}\right)}}\right)}.
\end{equation}

\section{Costs for realistic implementation of DB-TFD}

Here, we restrict our attention to the \texttt{vanilla} \texttt{DB-TFD} states, as the extension to the \texttt{poly} \texttt{DB-TFD} follows from the same calculation with additional reflections.

\subsection{Proof of Lemma ~\ref{lemma:PF_error}} \label{proof_PF_error_general}

We bound the error between the exact implementation of the exponential of commutators and its approximation obtained from an order-$m$ product formula over multiple iterations.

Define two sequences of states.
The first corresponds to the exact implementation of the exponential of the commutator,
\begin{equation} \label{app_eq:exact_dbi_k_general}
    \ket{\psi_{j+1}^{(\texttt{vanilla})}} = \prod_{i=0}^{j} e^{s_i [\psi_{i}^{(\texttt{vanilla})}, H]} \ket{\psi_{0}}.
\end{equation}
The second one is generated by replacing each exponential of the commutator in Eq.~\eqref{app_eq:exact_dbi_k_general} with its order-$m$ product-formula approximation,
\begin{equation} \label{app_eq:pf_k_}
    \ket{\phi_{j+1}^{(\texttt{vanilla})}} = \prod_{i=0}^{j} U_{\rm{PF}, i}^{({\texttt{vanilla}})}  \ket{\psi_{0}},
\end{equation}
where $U_{\rm{PF}, i}^{({\texttt{vanilla}})}$ denotes the product-formula approximation of order $m$.

Our objective is then to bound the accumulated error between the two states, quantified by the 2-norm, i.e., 
\begin{equation}
    \Delta_{k} = \| \ket{\psi_{k}^{(\texttt{vanilla})}}  - \ket{\phi_{k}^{(\texttt{vanilla})}}  \|.
\end{equation}

For an arbitrary index $j$, the gap $\Delta_{j+1}$ can be bounded using the triangle inequality as
\begin{equation}
\begin{split}
    \Delta_{j+1} &=  \|  \ket{\psi_{j+1}^{(\texttt{vanilla})}}  - \ket{\phi_{j+1}^{(\texttt{vanilla})}}    \| \\
    &\leq  \|  \ket{\psi_{j+1}^{(\texttt{vanilla})}}  - e^{s_j [\psi_{j}^{(\texttt{vanilla})}, H]}  \ket{\phi_{j}^{(\texttt{vanilla})}}    \| \\
    & \quad + \|  e^{s_j [\psi_{j}^{(\texttt{vanilla})}, H]}  \ket{\phi_{j}^{(\texttt{vanilla})}} - e^{s_j [\phi_{j}^{(\texttt{vanilla})}, H]}  \ket{\phi_{j}^{(\texttt{vanilla})}}    \| \\
    & \quad +  \| e^{s_j [\phi_{j}^{(\texttt{vanilla})}, H]}  \ket{\phi_{j}^{(\texttt{vanilla})}}  -   \ket{\phi_{j+1}^{(\texttt{vanilla})}}  \| .
\end{split}
\end{equation}

We study the three terms on the right-hand side in the following.

\begin{enumerate}
    \item We introduce an upper bound for the first two terms as
    \begin{equation}
    \begin{split}
    \|  \ket{\psi_{j+1}^{(\texttt{vanilla})}}  - e^{s_j [\psi_{j}^{(\texttt{vanilla})}, H]}  \ket{\phi_{j}^{(\texttt{vanilla})}}    \| 
      + \|  e^{s_j [\psi_{j}^{(\texttt{vanilla})}, H]}  \ket{\phi_{j}^{(\texttt{vanilla})}} - e^{s_j [\phi_{j}^{(\texttt{vanilla})}, H]}  \ket{\phi_{j}^{(\texttt{vanilla})}} = a_j \Delta_j,
    \end{split}
    \end{equation}
    where $a_j \geq 0$ characterizes the effective contraction.

    \item For the third term, we start with the simplification;
    \begin{equation}
        \| e^{s_j [\phi_{j}^{(\texttt{vanilla})}, H]}  \ket{\phi_{j}^{(\texttt{vanilla})}}  -   \ket{\phi_{j+1}^{(\texttt{vanilla})}}  \| \le \|  e^{s_j [\phi_{j}^{(\texttt{vanilla})}, H]} - U_{\rm{PF}, j}^{({\texttt{vanilla}})}  \|_{{\rm op}} \|\ket{\phi_{j}^{(\texttt{vanilla})}}\| =  \|  e^{s_j [\phi_{j}^{(\texttt{vanilla})}, H]} - U_{\rm{PF},j}^{({\texttt{vanilla}})}  \|_{{\rm op}},
    \end{equation}
    which allows us to straightforwardly use the error bound for product-formula approximation.
    That is, we have
    \begin{equation}
         \|  e^{s_j [\phi_{j}^{(\texttt{vanilla})}, H]} - U_{\rm{PF}, j}^{({\texttt{vanilla}})}  \|_{{\rm op}} \le C s_{j}^{\frac{m+1}{2}}.
    \end{equation}
\end{enumerate}

Now, combining together, we obtain
\begin{equation}
    \Delta_{j+1} \leq a_j \Delta_{j} +  C s_{j}^{\frac{m+1}{2}}.
\end{equation}
Without loss of generality, we define $s_{max}= \max_{j} (s_j)$, then the sequence is simplified as
\begin{equation}
    \Delta_{j+1} \le a_j \Delta_{j} +  b
\end{equation}
with $b= C s_{max}^{\frac{m+1}{2}}$.
Thus, by solving the above equation for $j=k$, we obtain
\begin{equation}
\begin{split}
    \Delta_{k} &= a_{k-1} a_{k-2} \dots a_2 a_1 \Delta_{0} +  b  \sum_{i = 1}^{k-1} \left( \prod_{j = i}^{k-1} a_j \right)\\
    &=a_{k-1} a_{k-2} \dots a_2 a_1 \Delta_{0}+ C s_{max}^{\frac{m+1}{2}} \sum_{i = 1}^{k-1} \left( \prod_{j = i}^{k-1} a_j \right).
\end{split}
\end{equation}
As $ \Delta_{0}=0$ without loss of generality, we obtain
\begin{equation}
    \Delta_{k} \le C s_{max}^{\frac{m+1}{2}} \sum_{i = 1}^{k-1} \left( \prod_{j = i }^{k-1} a_j \right)  .
\end{equation}
Using the definition of contraction coefficients in Def.~\ref{def:contraction-coeff} to simplify, we finally obtain
\begin{equation} \label{eq:general_PF_error}
    \Delta_{k} \le C s_{max}^{\frac{m+1}{2}} \sum_{i = 1}^{k-1} A_i = C s_{max}^{\frac{m+1}{2}} A(k) .
\end{equation}

\subsection{Setting product formula order} \label{app:product_formula_order_general}
We now determine the order of the product formula required to achieve the error bound:
\begin{equation}
    \left \| \ket{\psi_{k}^{(\texttt{vanilla})}}  - \ket{\phi_{k}^{(\texttt{vanilla})}}  \right\| \le \varepsilon.
\end{equation}
For clarity, we assume throughout that the operator norm of the Hamiltonian satisfies $\|H\|_{\mathrm{op}} \le 1$.

Before proceeding, we specify the constant $C$ appearing in the error bound.
This constant originates from the leading-order term in the Baker-Campbell-Hausdorff expansion associated with the product-formula approximation~\cite{Casas_2025}.
More explicitly, $C$ is given by
\begin{equation}
    C= \frac{\|\sum_j \alpha_j E_{m+1,j}\|_{{\rm op}}  }{(m+1)!} \le \frac{\sum_j |\alpha_j| \|E_{m+1,j}\|_{{\rm op}}  }{(m+1)!},
\end{equation}
where the coefficients $\alpha_j$ depend on the specific product formula employed as well as on the target evolution being approximated.
The term $E_{m+1,j}$ denotes the nested commutator of order $m+1$ arising in the Baker-Campbell-Hausdorff expansion~\cite{Casas_2025}, given by
\begin{equation}
    E_{m+1, j} = [X_{i_{1}},[X_{i_{2}}, \ldots, [X_{i_{m}},X_{i_{m+1}}]]],
\end{equation}
where $X_{i}\in\{\psi, H\}$.
Here, $\psi$ denotes the operator associated with the pure state used in the implementation.

We now bound the norm of the nested commutator terms.
Using the standard operator-norm inequality, we have
\begin{equation}
    \|[A,B]\|_{{\rm op}} \le \|AB\|_{{\rm op}} + \|BA\|_{{\rm op}} \le 2 \|A\|_{{\rm op}} \|B\|_{{\rm op}}.
\end{equation}
By recursively applying it, we obtain

\begin{equation}
\begin{split}
    \|E_{m+1, j}\|_{{\rm op}} &= \|[X_{i_{1}},[X_{i_{2}}, \ldots, [X_{i_{m}},X_{i_{m+1}}]]]\|_{{\rm op}} \\
    & \le 2^{m+1} \prod_{j=1}^{m+1} \|X_{i_j}\|_{{\rm op}} \\
    & \le 2^{m+1} \|H\|_{{\rm op}} \\
    & \le 2^{m+1},
\end{split} 
\end{equation}
where we use that the operator norm of a pure-state projector equals one, and the final inequality follows from the assumption $\|H\|_{\mathrm{op}} \le 1$.  

Define $\alpha=\sum_{i}|\alpha_j|$.
Then, the constant $C$ can be bounded as
\begin{equation}
    C \le \frac{\alpha 2^{m+1}}{(m+1)!}.
\end{equation}
Since $m\ge 2$, one may further simplify this expression using the inequality $2^{m+1}/(m+1)! \le 2^{m-1}/(m-1)!$.

Using the bound derived in Eq.~\eqref{eq:general_PF_error}, we require the total accumulated error to be bounded by $\varepsilon$:
\begin{equation}
    A(k) \; C s_{max}^{\frac{m+1}{2}} \le \varepsilon.
\end{equation}

Substituting the upper bound for $C$ and applying the standard factorial lower bound $x! \ge (x/e)^x$, we obtain:
\begin{equation}
    \alpha \; A(k) \; \left(\frac{2e\sqrt{s_{max}}}{m+1}\right)^{m+1} \le \varepsilon \implies \left( \frac{m+1}{2e\sqrt{s_{max}}} \right)^{m+1} \ge \frac{\alpha \; A(k)}{\varepsilon}.
\end{equation}

Taking the logarithm of both sides yields:
\begin{equation}
    (m+1) \log \left( \frac{2e\sqrt{s_{max}}}{m+1} \right) \leq \log \left( \frac{\varepsilon}{\alpha \; A(k)} \right).
\end{equation}

Since  $\log(x) \leq x-1$, the left-hand side of the equation can be bounded as
\begin{equation}
      (m+1)\log\left(\frac{2e\sqrt{s_{max}}}{m+1}\right)   \leq (m+1) \left(\frac{2e\sqrt{s_{max}}}{m+1} -1\right) = 2e\sqrt{s_{max}} - m - 1.
\end{equation}
Rearranging then yields a sufficient condition on the order $m$ of the product formula:
\begin{equation}
    m \ge 2e\sqrt{s_{max}} + \log\left(\frac{\alpha \; A(k)}{\varepsilon}\right)  - 1.
\end{equation}

\subsection{Setting product formula order using the general bound on contraction coefficient} \label{app:product_formula_order}
\begin{lemma}[Error incurred by an $m$-th order product formula]
Let $U_{DB,k}^{\rm{(v)}}$ denote the $k$-step \texttt{DB-TFD} using the exact exponential of the commutators, and let $U_{PF, k}^{\rm{(v)}}$ denote the corresponding unitary using an $m$-th order product formula with $N_m$ factors, for $\rm{v}=\{\texttt{vanilla},\,\texttt{poly}\}$.
Both constructions use the same Hamiltonian $H$ and employ an identical sequence of time steps.
Let $s_{(\rm max)}$ denote the largest step size. Then, for any initial state $|\psi_{0}\rangle$, the error after $k$ iterations is bounded by
\begin{equation} \label{app_eq:pf_bound}
    \Delta_{k+1} \leq   \frac{(1+2s_{(\rm{max})}\|H\|_{\rm{op}})^k}{2\|H\|_{\rm{op}}} Cs_{(\rm{max})}^{\frac{m-1}{2}},
\end{equation}
where $C$ is a constant dependent on the choice of product formula.
\label{lemma:PF_error_appendix}
\end{lemma}

\begin{proof}

We bound the error between the exact implementation of the exponential of commutators and its approximation obtained from an order-$m$ product formula over multiple iterations.

Define two sequences of states.
The first corresponds to the exact implementation of the exponential of the commutator,
\begin{equation} \label{app_eq:exact_dbi_k}
    \ket{\psi_{j+1}^{(\texttt{vanilla})}} = \prod_{i=0}^{j} e^{s_i [\psi_{i}^{(\texttt{vanilla})}, H]} \ket{\psi_{0}}.
\end{equation}
The second one is generated by replacing each exponential of the commutator in Eq.~\eqref{app_eq:exact_dbi_k} with its order-$m$ product-formula approximation,
\begin{equation} \label{app_eq:pf_k}
    \ket{\phi_{j+1}^{(\texttt{vanilla})}} = \prod_{i=0}^{j} U_{\rm{PF}, i}^{({\texttt{vanilla}})}  \ket{\psi_{0}},
\end{equation}
where $U_{\rm{PF}, i}^{({\texttt{vanilla}})}$ denotes the product-formula approximation of order $m$.

Our objective is then to bound the accumulated error between the two states, quantified by the 2-norm, i.e., 
\begin{equation}
    \Delta_{k} = \| \ket{\psi_{k}^{(\texttt{vanilla})}}  - \ket{\phi_{k}^{(\texttt{vanilla})}}  \|.
\end{equation}

For an arbitrary index $j$, the gap $\Delta_{j+1}$ can be bounded using the triangle inequality as
\begin{equation}
\begin{split}
    \Delta_{j+1} &=  \|  \ket{\psi_{j+1}^{(\texttt{vanilla})}}  - \ket{\phi_{j+1}^{(\texttt{vanilla})}}    \| \\
    &\le  \|  \ket{\psi_{j+1}^{(\texttt{vanilla})}}  - e^{s_j [\psi_{j}^{(\texttt{vanilla})}, H]}  \ket{\phi_{j}^{(\texttt{vanilla})}}    \| \\
    & \quad + \|  e^{s_j [\psi_{j}^{(\texttt{vanilla})}, H]}  \ket{\phi_{j}^{(\texttt{vanilla})}} - e^{s_j [\phi_{j}^{(\texttt{vanilla})}, H]}  \ket{\phi_{j}^{(\texttt{vanilla})}}    \| \\
    & \quad +  \| e^{s_j [\phi_{j}^{(\texttt{vanilla})}, H]}  \ket{\phi_{j}^{(\texttt{vanilla})}}  -   \ket{\phi_{j+1}^{(\texttt{vanilla})}}  \| .
\end{split}
\end{equation}

Each of the three terms on the right-hand side can be bounded separately as follows.

\begin{enumerate}
    \item The first term can be bounded as
    \begin{equation}
    \begin{split}
    \| \ket{\psi_{j+1}^{(\texttt{vanilla})}}  - e^{s_j [\psi_{j}^{(\texttt{vanilla})}, H]}  \ket{\phi_{j}^{(\texttt{vanilla})}}  \| &\le  \| e^{s_j [\psi_{j}^{(\texttt{vanilla})}, H]}\ket{\psi_{j}^{(\texttt{vanilla})}}  - e^{s_j [\psi_{j}^{(\texttt{vanilla})}, H]}  \ket{\phi_{j}^{(\texttt{vanilla})}}  \| \\
    & \le \|e^{s_j [\psi_{j}^{(\texttt{vanilla})}, H]}\|_{{\rm op}}  \| \ket{\psi_{j}^{(\texttt{vanilla})}}  -   \ket{\phi_{j}^{(\texttt{vanilla})}}  \| \\
    &=\Delta_{j},
    \end{split}
    \end{equation}
    where we use the fact that the norm of a unitary is one in the last equality.

    \item The second term is first simplified as
    \begin{equation}
    \begin{split}
         \|  e^{s_j [\psi_{j}^{(\texttt{vanilla})}, H]}  \ket{\phi_{j}^{(\texttt{vanilla})}} - e^{s_j [\phi_{j}^{(\texttt{vanilla})}, H]}  \ket{\phi_{j}^{(\texttt{vanilla})}}    \| & \le \| e^{s_j [\psi_{j}^{(\texttt{vanilla})}, H]}  -  e^{s_j [\phi_{j}^{(\texttt{vanilla})}, H]}    \|_{{\rm op}} \|\ket{\phi_{j}^{(\texttt{vanilla})}}\| \\
         &\le \| e^{s_j [\psi_{j}^{(\texttt{vanilla})}, H]}  -  e^{s_j [\phi_{j}^{(\texttt{vanilla})}, H]}    \|_{{\rm op}}.
    \end{split}
    \end{equation}
    In the second inequality we use $  \|\ket{\phi_{j}^{(\texttt{vanilla})}}\|=1$.
    Since \begin{equation}
    \|e^A-e^B\|_{{\rm op}} \leq \|A-B\|_{{\rm op}},
\end{equation}
in case $e ^A$ and $e^B$ are unitaries, we get 
\begin{equation}
\begin{split}
    \| e^{s_j [\psi_{j}^{(\texttt{vanilla})}, H]}  -  e^{s_j [\phi_{j}^{(\texttt{vanilla})}, H]}    \|_{{\rm op}} &\le \| s_j [\psi_{j}^{(\texttt{vanilla})}, H]  -  s_j [\phi_{j}^{(\texttt{vanilla})}, H]    \|_{{\rm op}} \\
    & \le 2s_j \|H\|_{{\rm op}}  \|\psi_{j}^{(\texttt{vanilla})}-\phi_{j}^{(\texttt{vanilla})}\|_{{\rm op}} \\
    &\le 2s_j \|H\|_{{\rm op}} \Delta_{j},
\end{split}
\end{equation}
where $\|\psi-\phi\|_{{\rm op}} \le \| \ket{\psi}- \ket{\phi} \| $ in the last line.

    \item For the third term, we start with the simplification;
    \begin{equation}
        \| e^{s_j [\phi_{j}^{(\texttt{vanilla})}, H]}  \ket{\phi_{j}^{(\texttt{vanilla})}}  -   \ket{\phi_{j+1}^{(\texttt{vanilla})}}  \| \le \|  e^{s_j [\phi_{j}^{(\texttt{vanilla})}, H]} - U_{\rm{PF}, j}^{({\texttt{vanilla}})}  \|_{{\rm op}} \|\ket{\phi_{j}^{(\texttt{vanilla})}}\| =  \|  e^{s_j [\phi_{j}^{(\texttt{vanilla})}, H]} - U_{\rm{PF},j}^{({\texttt{vanilla}})}  \|_{{\rm op}},
    \end{equation}
    which allows us to straightforwardly use the error bound for product-formula approximation.
    That is, we have
    \begin{equation}
         \|  e^{s_j [\phi_{j}^{(\texttt{vanilla})}, H]} - U_{\rm{PF}, j}^{({\texttt{vanilla}})}  \|_{{\rm op}} \le C s_{j}^{\frac{m+1}{2}}.
    \end{equation}
\end{enumerate}

Now, combining together, we obtain
\begin{equation}
    \Delta_{j+1} \le (1+2s_{j}\|H\|_{{\rm op}}) \Delta_{j} +  C s_{j}^{\frac{m+1}{2}}.
\end{equation}
Without loss of generality, we define $s_{max}= \max_{j} (s_j)$, then the sequence is simplified as
\begin{equation}
    \Delta_{j+1} \le a \Delta_{j} +  b
\end{equation}
with $a = 1+2s_{max}\|H\|_{{\rm op}}$ and $b= C s_{max}^{\frac{m+1}{2}}$.
Thus, by solving the above equation for $j=k$, we obtain
\begin{equation}
\begin{split}
    \Delta_{k} &= a^{k} \Delta_{0} + \frac{1-a^{k}}{1-a} b \\
    &=( 1+2s_{max}\|H\|_{{\rm op}})  \Delta_{0} + \frac{( 1+2s_{max}\|H\|_{{\rm op}})^{k} - 1}{2s_{max}\|H\|_{{\rm op}}} C s_{max}^{\frac{m+1}{2}}.
\end{split}
\end{equation}
As $ \Delta_{0}=0$ without loss of generality and $a^{k}-1 \le a^k$, we finally obtain
\begin{equation}
    \Delta_{k} \le \frac{( 1+2s_{max}\|H\|_{{\rm op}})^{k} - 1}{2\|H\|_{{\rm op}}} C s_{max}^{\frac{m-1}{2}}.
\end{equation}
If we assume $\|H\|_{{\rm op}}\le 1$, this further simplifies to
\begin{equation} \label{app_eq:bound_for_pf_general}
    \Delta_{k} \le \frac{ \left( 1+2s_{max}\right)^k}{2} C s_{max}^{\frac{m-1}{2}}.
\end{equation}

\end{proof}

Then, we determine the product formula order $m$ to satisfy the error bound $\Delta_k \le \varepsilon$.

From App.~\ref{app:product_formula_order_general}, the constant $C$ appearing in the error bound is given by 
\begin{equation}
    C \le \frac{\alpha 2^{m+1}}{(m+1)!}.
\end{equation}
Using the bound derived in Eq.~\eqref{app_eq:bound_for_pf_general}, we require the total accumulated error to be bounded by $\varepsilon$:

\begin{equation}
\begin{split}
    &\left( \frac{2e\sqrt{s_{max}}}{m-1} \right)^{m-1} \leq \frac{2\varepsilon}{\alpha (1+2s_{max})^k} \\
    \implies & (m-1)\log\left(\frac{2es_{max}}{m-1}\right) \leq \log \left(\frac{2\varepsilon}{\alpha (1+2s_{max})^k}\right).
\end{split}
\end{equation}
Since  $\log(x) \leq x-1$, the left-hand side of the equation can be bounded as
\begin{equation}
      (m-1)\log\left(\frac{2e\sqrt{s_{max}}}{m-1}\right)   \leq (m-1) \left(\frac{2e\sqrt{s_{max}}}{m-1} -1\right) = 2e\sqrt{s_{max}} -(m-1).
\end{equation}
Rearranging then yields a sufficient condition on the order $m$ of the product formula:
\begin{equation}
    m \ge 2e\sqrt{s_{max}} + k \log(1+2s_{max}) + \log\left(\frac{\alpha}{2\epsilon}\right) + 1.
\end{equation}
We note that the coefficient $\alpha$ may scale at most as $\mathcal{O}(c^{m+1})$~\cite{childs2021theory,low2019well,watson2025exponentially}, but this does not affect the overall scaling derived here.

\subsection{Lower bound on the total number of operations} \label{app_number_operations}
We finally assess the total number of operations required to achieve $\varepsilon$-precision in implementing the exact commutators using the product formula of order $m$.

\begin{theorem}[Oracle calls to Hamiltonian evolutions and reflection operators required to prepare $| \rm{TFD}(\beta)\rangle$ within error-precision $\varepsilon$]
\label{prop:total-query-naive-and-poly}
    To achieve an $\varepsilon$-precision approximation of the thermofield double state at inverse temperature $\beta$ using the \texttt{vanilla} or the \texttt{poly} \texttt{DB-TFD}, at least $N_{\rm{tot}}$ oracle calls to Hamiltonian evolutions and reflection operators are required. The scalings of the \texttt{vanilla} and the \texttt{poly} approaches are, respectively, given by
    \begin{itemize}
        \item \texttt{vanilla} \texttt{DB-TFD}: 
        \begin{equation}
            N_{\rm{tot}} \in \exp\left(\mathcal{O}\left(({e^\beta / \varepsilon})^2\right)\right),
        \end{equation}
        \item \texttt{poly} \texttt{DB-TFD}: 
        \begin{equation}
        \begin{split}
             &N_{\rm{tot}}  \in \exp\left(\mathcal{O}\left(\beta\log \left( \frac{8}{\varepsilon}\sqrt\frac{D}{Z(\beta)}\right)\log(1+2s_{max})\right)\right),
        \end{split}
        \label{eq:poly-DB-TFD-with-PF-error}
        \end{equation}
        where $D$ denotes the system dimension and $Z(\beta)$ is the partition function associated with the shifted Hamiltonian $H + \mathds{1}$. 
    \end{itemize}
    \end{theorem}
\begin{proof}

Let $N_{m}$ denote the total number of operations for product-formula approximation at a single step.
Then, $k$ recursion steps (or equivalently a degree-$k$ polynomial) require $\mathcal{O}(N_{m}^{k})$ queries in total.
We note that it is known theoretically that an $m$-th order product formula requires at least $\Omega(2^m/m)$ queries~\cite{Casas_2025} to approximate the exponential of commutators accurately.
Consequently, we obtain the lower bound
\begin{equation}
    N_m^k \geq  \left(\frac{2^{m}}{m}\right)^{k} = \frac{2^{km}}{e^{k\log(m)}} .
\end{equation}
This can be further bounded as
\begin{equation}
    \frac{2^{km}}{e^{k\log(m)}}  \leq \frac{e^{km}}{e^{k\log(m)}} = e^{k(m-\log(m))}.
\end{equation}
Since $m -\log(m)\in\mathcal{O}(m)$, the total query complexity scales as $e^{\mathcal{O}(km)}$.
Moreover, plugging the condition of $m$ derived in App.~\ref{app:product_formula_order}, we obtain
\begin{equation}
    e^{km} \geq  e^{k^2 \log(1+2s_{max})+ 2k(\sqrt{s_{max}}e+1)  +k\log\left(\frac{\alpha}{2 \varepsilon}\right)},
\end{equation}
which scales as $e^{\mathcal{O}(k^2)}$.

\medskip
\medskip
In what follows, we show the query complexity for both the \texttt{vanilla} \texttt{DB-TFD} and the \texttt{poly} \texttt{DB-TFD}.

\medskip
\medskip
\paragraph*{Number of operations for the \texttt{vanilla} \texttt{DB-TFD}} --
For the \texttt{vanilla} \texttt{DB-TFD} approach, Lemma~\ref{theorem:naive_error} suggests that the number of recursion step $k$ must scale as $\mathcal{O}(\frac{e^{\beta \|H\|_{\rm{op}}}}{\varepsilon})$ to achieve precision $\varepsilon$.
Consequently, the total number of queries, which scales as $e^{\mathcal{O}(k^2)}$, grows doubly exponentially with $\beta \|H\|_{\rm{op}}$.

\medskip
\medskip

\paragraph*{Number of operations for the \texttt{poly} \texttt{DB-TFD}} --
Lemma~\ref{theorem:poly_error} shows that the \texttt{poly} \texttt{DB-TFD} approach requires the polynomial degree $$k \geq  \frac{e}{2}\sqrt{\beta \Delta E\log\left(\frac{4}{\varepsilon} \sqrt{\frac{D}{Z'\left(\frac{\beta\Delta E}{2}\right)}}\right)}$$ to achieve $\varepsilon$-precision.
Thus, the total number of queries scaling as $e^{\mathcal{O}(k^2)}$ has the leading order
\begin{equation}
    N_{\rm{tot}} \in \exp\left(\frac{e^2}{4}\beta \Delta E \log\left(\frac{8}{\varepsilon}\sqrt{\frac{D}{Z\left(\frac{\beta\Delta E}{2}\right)}}\right)\log(1+2s_{max})\right).
\end{equation}

Recall that $Z'(\beta)$ is the partition function of $H'+I$, where $H'$ is the normalized Hamiltonian with spectrum in $[-1,1]$ and depends on the inverse temperature $\beta$.
Hence, it is misleading to claim that the overall scaling is merely $e^{\mathcal{O}(\beta)}$.
In particular, for large $\beta$, the partition function is dominated by the ground-state contribution, $e^{-\beta(E_0+1)} \geq e^{-\beta}= 1$, where $E_0 \in [-1,0]$ due to the normalization.
Consequently, in the large-$\beta$ limit, the scaling of the exponent for the total queries is $k^2=e^2\beta^2/4 \log(4\sqrt{D}/\varepsilon)\log(1+s_{\max})$.
\end{proof}

For the \texttt{vanilla} \texttt{DB-TFD}, the total query complexity $N_{\rm{tot}}$ scales super-exponentially in the inverse temperature, since the number of recursive steps $k$ itself grows exponentially with $\beta$, as shown in Lemma~\ref{theorem:naive_error}.
This implies unfavorable scaling even away from the low-temperature regime.

By contrast, the \texttt{poly} \texttt{DB-TFD} scheme exhibits a substantially improved scaling, as given in Eq.~\eqref{eq:poly-DB-TFD-with-PF-error}.
This scaling can be further simplified in standard settings. 
First, for a normalized Hamiltonian, the partition function admits the lower bound $Z(\beta)\ge e^{-\beta}$.
Second, considering a worst-case distribution of polynomial roots, the step size $s$ can be upper-bounded uniformly.
Incorporating these considerations yields an overall query complexity scaling as $N_{\rm{tot}} \in \exp(\mathcal{O}(\beta^2))$.

\subsection{Proof of Lemma ~\ref{lemma:PF_error_linear}} \label{app:proof_PF_error_linear}

Using Lemma ~\ref{lemma:PF_error}, as stated in Eq.~\eqref{eq:contraction_conj}, 
we assume that the contraction coefficient can be bounded by a polynomial in $k$ as follows:
\begin{equation}
A(k) = \sum_{i=1}^{k-1}\prod_{j = i}^{k} a_j \in \mathcal{O}(k^c)
\end{equation}
for some constant $c \ge 0$.
Thus, the total error accumulation is bounded strictly by a polynomial bound:
\begin{equation} \label{app_eq:bound_for_pf}
\Delta_{k} \le \mathcal{O}\left(k^{c} \, C s_{max}^{\frac{m+1}{2}}\right).
\end{equation}

\subsection{Setting product formula order in the practical regime} \label{app:product_formula_order_efficient}

We determine the product formula order $m$ to satisfy the error bound $\Delta_k \le \varepsilon$.

From App.~\ref{app:product_formula_order_general}, the constant $C$ appearing in the error bound is given by 
\begin{equation}
    C \le \frac{\alpha 2^{m+1}}{(m+1)!}.
\end{equation}
Using the bound derived in Eq.~\eqref{app_eq:bound_for_pf}, and introducing $\gamma$ as the constant factor from the $\mathcal{O}(k^{c})$ bound, we require the total accumulated error to be bounded by $\varepsilon$:
\begin{equation}
    \gamma k^{c} C s_{max}^{\frac{m+1}{2}} \le \varepsilon.
\end{equation}

Substituting the upper bound for $C$ and applying the standard factorial lower bound $x! \ge (x/e)^x$, we obtain:
\begin{equation}
    \alpha \gamma k^{c+1} \left(\frac{2e\sqrt{s_{max}}}{m+1}\right)^{m+1} \le \varepsilon \implies \left( \frac{m+1}{2e\sqrt{s_{max}}} \right)^{m+1} \ge \frac{\alpha \gamma k^{c}}{\varepsilon}.
\end{equation}

Taking the logarithm of both sides yields:
\begin{equation}
    (m+1) \log \left( \frac{2e\sqrt{s_{max}}}{m+1} \right) \leq \log \left( \frac{\varepsilon}{\alpha \gamma k^{c}} \right).
\end{equation}

Since  $\log(x) \leq x-1$, the left-hand side of the equation can be bounded as
\begin{equation}
      (m+1)\log\left(\frac{2e\sqrt{s_{max}}}{m+1}\right)   \leq (m+1) \left(\frac{2e\sqrt{s_{max}}}{m+1} -1\right) = 2e\sqrt{s_{max}} - m - 1.
\end{equation}
Rearranging then yields a sufficient condition on the order $m$ of the product formula:
\begin{equation}
    m \ge 2e\sqrt{s_{max}} + \log\left(\frac{\alpha \gamma}{\varepsilon}\right) + c \log(k) - 1.
\end{equation}

\subsection{Lower bound on the total number of operations for efficient cases} \label{app_number_operations_efficient}
We finally assess the total number of operations required to achieve $\epsilon$-precision in implementing the exact commutators using the product formula of order $m$.
From App.~\ref{app_number_operations}, the total query complexity scales as $e^{\mathcal{O}(km)}$. Then, using the product formula order found in App.~\ref{app:product_formula_order_efficient}, we have
\begin{equation}
    e^{km} \geq  e^{2ek\sqrt{s_{max}} + k\log\left(\frac{\alpha \gamma}{\varepsilon}\right) + c k\log(k) - k},
\end{equation}
which scales as $e^{\mathcal{O}(k \log k)}$.

\medskip
\medskip
In what follows, we show the query complexity for both the \texttt{vanilla} \texttt{DB-TFD} and the \texttt{poly} \texttt{DB-TFD} for these efficient cases.

\medskip
\medskip

\paragraph*{Number of operations for the \texttt{vanilla} \texttt{DB-TFD}} --
For the \texttt{vanilla} \texttt{DB-TFD} approach, Lemma~\ref{theorem:naive_error} suggests that the number of recursion step $k$ must scale as $\mathcal{O}(\frac{e^{\beta \|H\|_{\rm{op}}}}{\varepsilon})$ to achieve precision $\varepsilon$.
Consequently, the total number of queries $N_{\rm{tot}}$, which scales as $e^{\mathcal{O}(k \log k)}$, grows doubly exponentially with $\beta \|H\|_{\rm{op}}$.

\medskip
\medskip

\paragraph*{Number of operations for the \texttt{poly} \texttt{DB-TFD}} --
Lemma~\ref{theorem:poly_error} shows that the \texttt{poly} \texttt{DB-TFD} approach requires the polynomial degree $$k \geq  \frac{e}{2}\sqrt{\beta \Delta E\log\left(\frac{8}{\varepsilon} \sqrt{\frac{D}{Z'\left(\frac{\beta\Delta E}{2}\right)}}\right)}$$ to achieve $\varepsilon$-precision. 
Thus, the total number of queries scaling as $e^{\mathcal{O}(k \log k)}$ has the leading order
\begin{equation}
    N_{\rm{tot}} \in \exp \left( \mathcal{O} \left( \sqrt{\beta \Delta E \log \left( \frac{1}{\varepsilon} \sqrt{\frac{D}{Z'(\frac{\beta \Delta E}{2})}}\right)} \log \left( \beta \Delta E \log \left( \frac{1}{\varepsilon} \sqrt{\frac{D}{Z'(\frac{\beta \Delta E}{2})}}\right) \right)\right) \right).
\end{equation}

\subsection{Error propagation in parameters due to statistical noise} \label{app_statistical_noise_propagation}
In this section, we analyze the error propagation of the statistical noise for parameters $s$ and $\theta$ in \texttt{poly} \texttt{DB-TFD}.
As shown in App.~\ref{app:Double-bracket quantum signal processing (DB-QSP)}, each iteration step $j$ requires the estimation of 
\begin{align}
        s_j &=  -\frac{1}{\sqrt{V_j}} \arccos \left( \frac{|E_j-z|}{\sqrt{V_j+|E_j-z|^2}} \right) \\
        \theta_j &= \arg\left(\frac{E_j-z}{|E_j-z|}\right),
\end{align}
where $E_j$ and $V_j$ are the energy and variance with respect to the state at $j$-th iteration.
In practice, however, these quantities cannot be evaluated exactly due to finite sampling and imperfections in state preparation. Consequently, only estimates of the energy and variance, denoted $\hat{E}_j$ and $\hat{V}_j$ are available.
Let $\hat{s}_j$ and $\hat{\theta}_j$ denote the corresponding estimates of $s_j$ and $\theta_j$.
These quantities are given by
\begin{align}
        \hat{s}_j &=  -\frac{1}{\sqrt{\hat{V}_j}} \arccos \left( \frac{|\hat{E}_j-z|}{\sqrt{\hat{V}_j+|\hat{E}_j-z|^2}} \right) \\
        \hat{\theta}_j &= \arg\left(\frac{\hat{E}_j-z}{|\hat{E}_j-z|}\right).
\end{align}
Our goal is to quantify how errors in the estimated energy and variance propagate to the parameters $s_j$ and $\theta_j$.
To this end, we define $\delta E_j = |E_j - \hat{E}_j|$ and $\delta V_j = |V_j - \hat{V}_j|$, and seek upper bounds on
$\delta s_j = |s_j - \hat{s}_j|$ and $\delta \theta_j = |e^{i\theta_j} - e^{i\hat{\theta}_j}|$ expressed as a function of $\delta E_j $ and $\delta V_j $.

\paragraph*{Error propagation for $\delta \theta_j$ --}
First, we study $\delta \theta_j$. Using the triangle inequality, we can write
\begin{equation}
    \begin{split}
        \delta \theta_j & = |e^{i\theta_j} - e^{i\hat{\theta}_j}| \\
        & = \left | \frac{\hat{E}_j - z}{|\hat{E}_j - z|} - \frac{E_j - z}{|E_j - z|}\right | \\
        &\leq |\hat{E}_j - z| \left | \frac{1}{|\hat{E}_j - z|} - \frac{1}{|E_j - z|}\right |  + \frac{|\hat{E}_j - E_j|}{|{E}_j - z|} \\
        & \leq \eta_{E_j} \left ||{E}_j - z| - |\hat{E}_j - z| \right | + \eta_{E_j} |\hat{E}_j - {E}_j| \\
        &\leq 2 \eta_{E_j} |\hat{E}_j - {E}_j| \\
        &= 2 \eta_{E_j} \, \delta E_j, 
    \end{split}
    \label{app-eq-delta-theta-error}
\end{equation}
where $\eta_{E_j} = \max\{ \frac{1}{|{E}_j - z|} , \frac{1}{|\hat{E}_j - z|}\}$.

\paragraph*{Error propagation for $\delta s_j$ --}

Let $E_j, \hat{E}_j \in \mathbb{R}$ and $V_j, \hat{V}_j > 0$. For a fixed parameter $z \in \mathbb{C}$, we have the following errors:
\begin{equation}
    \begin{split}
        \delta E_j &= |E_j - \hat{E}_j|\\
        \delta V_j &= |V_j - \hat{V}_j|
    \end{split}
\end{equation}

Let $\eta_{V_j} = \max\{1/\sqrt{V_j},1/ \sqrt{\hat{V}_j}\}$. We assume the distances to $z$ are bounded strictly away from zero, and define the maximum inverse distance $\eta_E$ as 
\begin{equation}
 \eta_E = \max_j\left\{ \eta_{E_j}\right\}_j = \max_j \left\{ \frac{1}{|E_j - z|} , \frac{1}{|\hat{E}_j - z|} \right\}_j.
\end{equation}

Moreover, let $d_j = |E_j - z|$ and $\hat{d}_j = |\hat{E}_j - z|$. Using these notations, we can write $s_j$ and its estimated value $\hat{s}_j$ as follows. 
\begin{equation}
    \begin{split}
        s_j &= \frac{1}{\sqrt{V_j}}\arccos\left(\frac{d_j}{\sqrt{V_j+d_j^2}}\right)\\
        \hat{s}_j &= \frac{1}{\sqrt{\hat{V}_j}}\arccos\left(\frac{\hat{d}_j}{\sqrt{\hat{V}_j+\hat{d}_j^2}}\right)
    \end{split}
\end{equation}

By applying the identity $\arccos\left(\frac{x}{\sqrt{a+x^2}}\right) = \arctan\left(\frac{\sqrt{a}}{x}\right)$ for $a>0, x>0$, we have an expression
\begin{equation}
    \begin{split}
        s_j &= \frac{1}{\sqrt{V_j}}\arctan\left(\frac{\sqrt{V_j}}{d_j}\right)\\
        \hat{s}_j &= \frac{1}{\sqrt{\hat{V}_j}}\arctan\left(\frac{\sqrt{\hat{V}_j}}{\hat{d}_j}\right).
    \end{split}
\end{equation}

Let us define the general function for $s$ as
\begin{equation}
    s(V, d) := \frac{1}{\sqrt{V}}\arctan\left(\frac{\sqrt{V}}{d}\right). 
\end{equation}

To derive an upper bound on $\delta s_j = |s_j - \hat{s}_j|$, we employ the triangle inequality to decompose the error into contributions originating from the perturbations in $d$ (which depends on $E$) and $V$:
\begin{equation}
    |s_j - \hat{s}_j| = |s(V_j, d_j) - s(\hat{V}_j, \hat{d}_j)| \leq |s(V_j, d_j) - s(\hat{V}_j, d_j)| + |s(\hat{V}_j, d_j) - s(\hat{V}_j, \hat{d}_j)|.
\end{equation}
Here, the first term $|s(V_j, d_j) - s(\hat{V}_j, d_j)|$ captures the contribution arising from the perturbation in $V$, whereas the second term $ |s(\hat{V}_j, d_j) - s(\hat{V}_j, \hat{d}_j)|$ quantifies the sensitivity of $s$ in $E$ through its dependence on $d$.

Now, we will use the Mean Value Theorem to bound each term separately.
\begin{enumerate}
    \item \textbf{Bounding the first term:} 
    First, we need to compute the partial derivative of $s(V,d)$ with respect to $V$. Defining $u = \frac{\sqrt{V}}{d}$ for simplicity, we obtain
    \begin{equation}
        \begin{split}
            \frac{\partial s(V,d)}{\partial V} &= -\frac{1}{2V^{3/2}}\arctan(u) + \frac{1}{\sqrt{V}}\left(\frac{1}{1+u^2}\right)\left(\frac{1}{2d\sqrt{V}}\right) \\
            &= -\frac{1}{2V^{3/2}}\arctan(u) + \frac{1}{2V(d^2+V)}d \\
            &= -\frac{1}{2V^{3/2}} \left( \arctan(u) - \frac{d\sqrt{V}}{d^2+V} \right) \\&= -\frac{1}{2V^{3/2}} \left( \arctan(u) - \frac{u}{1+u^2} \right)
        \end{split}
    \end{equation}

    Since $\arctan(u) \geq \frac{u}{1+u^2}$ for all $u \geq 0$, the absolute value of the derivative is given by 
    \begin{equation}
        \left|\frac{\partial s(V,d)}{\partial V}\right| = \frac{1}{2V^{3/2}} \left( \arctan(u) - \frac{u}{1+u^2} \right)
        \label{app-eq:derivative-s-v}
    \end{equation}

    To further bound the quantity, we analyze the function $g(u) = \arctan(u) - \frac{u}{1+u^2}$. The derivative of $g(u)$ is given by 
    \begin{equation}
        \frac{dg(u)}{du} = \frac{1}{1+u^2} - \frac{1(1+u^2) - 2u^2}{(1+u^2)^2} = \frac{2u^2}{(1+u^2)^2}
    \end{equation}
    Since ${dg(u)}/{du} > 0$ for all $u > 0$, the function $g(u)$ is strictly increasing. Hence, the supremum of $g(u)$ occurs as $u \to \infty$ (which corresponds to $d \to 0$), i.e., 
    \begin{equation}
        \lim_{u \to \infty} \left[ \arctan(u) - \frac{u}{1+u^2} \right] = \frac{\pi}{2} - 0 = \frac{\pi}{2}
    \end{equation}

    Substituting this supremum into Eq.~\eqref{app-eq:derivative-s-v} gives the following bound; 
    \begin{equation}
        \left|\frac{\partial s(V,d)}{\partial V}\right| < \frac{1}{2V^{3/2}} \left( \frac{\pi}{2} \right) = \frac{\pi}{4V^{3/2}}.
    \end{equation}

    Finally, using the Mean Value Theorem, we obtain an upper bound for the first term as follows;
    \begin{equation}
        |s(V_j, d_j) - s(\hat{V}_j, d_j)| \leq \frac{\pi \, \eta_{V_j}^{3}}{4} \delta V_j.
    \end{equation}

    \item \textbf{Bounding the second term:} 
    Similar to the first term, we start with computing the partial derivative of $s(V,d)$ with respect to $d$ and obtain
    \begin{equation}
        \begin{split}
            \frac{\partial s(V,d)}{\partial d} = \frac{1}{\sqrt{V}} \left( \frac{1}{1 + (\sqrt{V}/d)^2} \right) \left( -\frac{\sqrt{V}}{d^2} \right) = -\frac{1}{V + d^2}.
        \end{split}
    \end{equation}

    By the Mean Value Theorem, there exists some intermediate $d_j^*$ between $d_j$ and $\hat{d}_j$, such that  
    \begin{equation}
        |s(\hat{V}_j, d_j) - s(\hat{V}_j, \hat{d}_j)| \leq \left| \frac{\partial s(V,d)}{\partial d}\bigg|_{(\hat{V}_j, d_j^*)} \right| |d_j - \hat{d}_j| = \frac{|d_j - \hat{d}_j|}{\hat{V}_j + d_j^{*2}}.
    \end{equation}
    Using the reverse triangle inequality, we can rewrite $|d_j - \hat{d}_j|$ as 
    \begin{equation}
         |d_j - \hat{d}_j| = \big| |E_j - z| - |\hat{E}_j - z| \big| \leq |E_j - \hat{E}_j| = \delta E_j.
    \end{equation}
    Since $\hat{V}_j + d_j^{*2} \geq \hat{V}_j \geq 1/\eta^2_{V_j}$, we obtain the following bound for the second term. 
    \begin{equation}
         |s(\hat{V}_j, d_j) - s(\hat{V}_j, \hat{d}_j)| \leq \eta^2_{V_j}\delta E_j
    \end{equation}
\end{enumerate}

Combining the upper bounds for the first and second terms yields the inequality for $|\delta s_j|$ as follows;
\begin{equation}
    \delta s_j = |s_j - \hat{s}_i| \leq \eta_{V_j}^2 \delta E_j + \frac{\pi \, }{4}\eta_{V_j}^{3}\,\delta V_j.
\end{equation}

Applying the square root to both sides gives the final bound as
\begin{equation}
    \sqrt{\delta s_j} \leq \sqrt{\eta_{V_j}^2 \delta E_j + \eta_{V_j}^{3}\,\delta V_j}, 
\end{equation}
where by the sub additivity of the square root, $\sqrt{a+b} \leq \sqrt{a} + \sqrt{b}$.
Hence a further simplification leads to 
\begin{equation}
    \sqrt{\delta s_j} \leq \eta_{V_j} \sqrt{\delta E_j} + \eta_{V_j}^{3/2} \sqrt{\delta V_j}
    \label{app-eq-delta-s-error}.
\end{equation}

\subsection{Error due to the statistical noise}\label{app_proof_prop_1}
Here, we bound the error incurred by replacing the exact energies and variances appearing in an $m$-order product formula with their estimated values. 
Since the \texttt{vanilla} approach requires no intermediate measurements, we only study this error for the \texttt{poly} \texttt{DB-TFD}.

For the \hiddenlink{poly-DBTFD}{\texttt{poly} \texttt{DB-TFD}}, the parameters $\{(\theta_{i},s_{i})\}$ depend explicitly on the energy and the variance.
This indicates that the empirical estimates $(\hat{E},\hat{V})$ of the true values $(E,V)$ leads to the deviation between the implemented evolution $| \varphi_k^{(\texttt{poly})} \rangle$ and the ideal one $| \phi_k^{(\texttt{poly})} \rangle$ at each step.
Assuming a standard unbiased empirical estimator, we obtain directly from Chebyshev's inequality that when using $N_{\rm{shots}}$ measurement shots to estimate the energy $\hat{E}_i$, 
\begin{equation}
    \Pr\left( |E_i - \hat{E}_i| \leq \epsilon_{\rm{shots}}\right) \geq 1 - \frac{\langle \varphi^{(\texttt{poly})}_{i}|(H - E_i)^2|\varphi^{(\texttt{poly})}_{i} \rangle}{N_{\rm{shots}} \epsilon^2_{\rm{shots}}}.
\end{equation}
Moreover, in order to have $|E_i - \hat{E}_i| \leq \epsilon_{\rm{shots}}$ with probability at least $1 - \delta_{p}$, we need $N_{\rm{shots}} \geq \langle \varphi^{(\texttt{poly})}_{i}|(H - E_i)^2|\varphi^{(\texttt{poly})}_{i} \rangle/\epsilon^2_{\rm{shots}} \delta_p$ measurement shots. 
Similarly, we can apply the same statistical analysis for estimated variance $\hat{V}_i$. 
Using $N_{\rm{shots}}$ shots to estimate $\hat{V}_i$, the probability that the error satisfies $|V_i - \hat{V}_i| \leq \epsilon_{\rm{shots}}$ is lower bounded by 
\begin{equation}
    \Pr\left( |V_i - \hat{V}_i| \leq \epsilon_{\rm{shots}}\right) \geq 1 - \frac{\langle \varphi^{(\texttt{poly})}_{i}|(H - E_i)^4|\varphi^{(\texttt{poly})}_{i} \rangle}{N_{\rm{shots}} \epsilon^2_{\rm{shots}}}.  
\end{equation}
Hence, in order to have $|V_i - \hat{V}_i| \leq \epsilon_{\rm{shots}}$ with probability at least $1 - \delta_{p}$, we need $N_{\rm{shots}} \geq \langle \varphi^{(\texttt{poly})}_{i}|(H - E_i)^4|\varphi^{(\texttt{poly})}_{i} \rangle/\epsilon^2_{\rm{shots}} \delta_p$ measurement shots. 
With these in mind, the following proposition provides the error due to the statistical noise for the energy and the variance estimation.

\begin{prop}[Error incurred by finite number of measurement shots to estimate the energy and the variance] \label{prop:statistical_error}
    Define $\delta E = \max_k |\hat{E_k} -E_k|$ as the maximum error between the estimated and true energies over all iterations.
    Similarly, define $\delta V = \max_k |\hat{V_k}-V_k|$ for the variance. 
    Denote by $z_k$ the roots of the corresponding polynomial, with $|z_k| \leq |z|$,
    Consider the state $|\phi_k^{(\texttt{poly})} \rangle$ generated using parameters computed from the exact energies and variances after $k$ iterations, and the state $| \varphi_k^{(\texttt{poly})} \rangle$ generated using the corresponding empirical estimates.
    Set $\delta_{E,V}=\max(\delta E,\delta V)$ and define $\eta_{E,V} = \rm{max}_k \left( \frac{1}{\sqrt{V_k}}, \frac{1}{\sqrt{\hat{V}_k}}, \frac{1}{|E_k - z|}, \frac{1}{|\hat{E}_k - z|}\right)$. Then, the following bound holds:
    \begin{equation}
        \|| \varphi_k^{(\texttt{poly})} \rangle-| \phi_k^{(\texttt{poly})} \rangle\| \leq \frac{2   \gamma_m }{13} \left( 14 + 2 \gamma_m \eta_{E,V}\right)^{k} \delta.
    \end{equation}
    where $\gamma_m$ is a constant depending on the choice of product formula, and $\delta = \max\left(2\eta_{E,V} \delta_{E,V}, \eta_{E,V} \sqrt{\delta_{E,V}}(1 + \sqrt{\eta_{E,V}})\right)$. 
\end{prop}

\begin{proof}
 
We start with defining two states. The first state corresponds to the implementation of $m$-order product formula with exact energies and variances, given by
\begin{equation}
    | \phi^{(\texttt{poly})}_{j+1} \rangle = \prod_{l = 0}^{j} e^{i \theta_l \phi^{(\texttt{poly})}_{l}} e^{i\alpha_1 \sqrt{s_l}H} e^{i\alpha_2 \sqrt{s_l} \phi^{(\texttt{poly})}_{l}} e^{i\alpha_3 \sqrt{s_l}H} e^{i\alpha_4 \sqrt{s_l} \phi^{(\texttt{poly})}_{l}} \cdots e^{i\alpha_{N_m - 1} \sqrt{s_l}H} e^{i\alpha_{N_m} \sqrt{s_l} \phi^{(\texttt{poly})}_{l}} | \rm{TFD}(0) \rangle,
\end{equation}
where the number of factor $N_m$, and the phases $\{ \alpha_i \}_{i=1}^{N_m}$ depend on the specific product formula employed.

The second one is generated by replacing the parameters $\{ \theta_i, s_i \}_i$ with their approximated ones $\{ \hat{\theta}_i, \hat{s}_i \}_i$, 
\begin{equation}
    | \varphi^{(\texttt{poly})}_{j+1} \rangle = \prod_{l = 0}^{j} e^{i \hat{\theta}_l \varphi^{(\texttt{poly})}_{l}} e^{i\alpha_1 \sqrt{\hat{s}_l}H} e^{i\alpha_2 \sqrt{\hat{s}_l} \varphi^{(\texttt{poly})}_{l}} e^{i\alpha_3 \sqrt{\hat{s}_l}H} e^{i\alpha_4 \sqrt{\hat{s}_l} \varphi^{(\texttt{poly})}_{l}} \cdots e^{i\alpha_{N_m - 1} \sqrt{\hat{s}_l}H} e^{i\alpha_{N_m} \sqrt{\hat{s}_l} \varphi^{(\texttt{poly})}_{l}} | \rm{TFD}(0) \rangle.
\end{equation}
Note that throughout this section, we drop the absolute value notation for brevity; i.e., $s_i$ and $\hat{s}_i$ represent $|s_i|$ and $|\hat{s}_i|$, respectively.
Then we aim to bound the accumulated error between the two states as a function of $\delta s_i $ and $\delta \theta_i$, with $\delta s_i = |s_i - \hat{s}_i|$ and $\delta \theta_i = |e^{i\theta_i} - e^{i\hat{\theta}_i}|$.
To this end, we define the error between the two states using the 2-norm;
\begin{equation}
    \hat{\Delta}_{j+1} = \| | \phi^{(\texttt{poly})}_{j+1} \rangle - | \varphi^{(\texttt{poly})}_{j+1} \rangle\|.
    \label{app:eq-error-noise}
\end{equation}

Before proceeding, let us define the following notation to facilitate the proof; 
\begin{align}
    U_{\theta}^{l} &= e^{i {\theta}_l \phi^{(\texttt{poly})}_{l}}, \\
    \hat{U}_{\theta}^{l} &= e^{i {\theta}_l \varphi^{(\texttt{poly})}_{l}}, \\
    \hat{U}_{\hat{\theta}}^{l} &= e^{i \hat{\theta}_l \varphi^{(\texttt{poly})}_{l}}, \\
    U_{s}^{l} &= e^{i\alpha_1 \sqrt{{s}_l}H} e^{i\alpha_2 \sqrt{{s}_l} \phi^{(\texttt{poly})}_{l}} e^{i\alpha_3 \sqrt{{s}_l}H} e^{i\alpha_4 \sqrt{{s}_l} \phi^{(\texttt{poly})}_{l}} \cdots e^{i\alpha_{N_m - 1} \sqrt{{s}_l}H} e^{i\alpha_{N_m} \sqrt{{s}_l} \phi^{(\texttt{poly})}_{l}} \label{eq:u-s-def},\\
    \hat{U}_{s}^{l} &= e^{i\alpha_1 \sqrt{{s}_l}H} e^{i\alpha_2 \sqrt{{s}_l} \varphi^{(\texttt{poly})}_{l}} e^{i\alpha_3 \sqrt{{s}_l}H} e^{i\alpha_4 \sqrt{{s}_l} \varphi^{(\texttt{poly})}_{l}} \cdots e^{i\alpha_{N_m - 1} \sqrt{{s}_l}H} e^{i\alpha_{N_m} \sqrt{{s}_l} \varphi^{(\texttt{poly})}_{l}} \label{eq:hat-u-s-def},\\
    \hat{U}_{\hat{s}}^{l} &= e^{i\alpha_1 \sqrt{\hat{s}_l}H} e^{i\alpha_2 \sqrt{\hat{s}_l} \varphi^{(\texttt{poly})}_{l}} e^{i\alpha_3 \sqrt{\hat{s}_l}H} e^{i\alpha_4 \sqrt{\hat{s}_l} \varphi^{(\texttt{poly})}_{l}} \cdots e^{i\alpha_{N_m - 1} \sqrt{\hat{s}_l}H} e^{i\alpha_{N_m} \sqrt{\hat{s}_l} \varphi^{(\texttt{poly})}_{l}}. \label{eq:hat-u-hat-s-def}
\end{align}
Using these notations, the error in Eq.~\eqref{app:eq-error-noise} can be rewritten as 
\begin{equation}
    \hat{\Delta}_{j+1} = \| | \phi^{(\texttt{poly})}_{j+1} \rangle - | \varphi^{(\texttt{poly})}_{j+1} \rangle\| = \| U_{\theta}^{j} U_{s}^{j} |\phi^{(\texttt{poly})}_{j} \rangle - \hat{U}_{\hat{\theta}}^{j} \hat{U}_{\hat{s}}^{j} |\varphi^{(\texttt{poly})}_{j} \rangle \|.
\end{equation}
First, we apply the triangle inequality to obtain
\begin{equation}
    \begin{split}
        \| | \phi^{(\texttt{poly})}_{j+1} \rangle - | \varphi^{(\texttt{poly})}_{j+1} \rangle\| & \leq 
        \| U_{\theta}^{j} U_{s}^{j} |\phi^{(\texttt{poly})}_{j} \rangle - U_{\theta}^{j} U_{s}^{j} |\varphi^{(\texttt{poly})}_{j} \rangle\| \\
        &+ \| U_{\theta}^{j} U_{s}^{j} |\varphi^{(\texttt{poly})}_{j} \rangle - \hat{U}_{\theta}^{j} U_{s}^{j} |\varphi^{(\texttt{poly})}_{j} \rangle\| \\
        &+ \| \hat{U}_{\theta}^{j} U_{s}^{j} |\varphi^{(\texttt{poly})}_{j} \rangle - \hat{U}_{\theta}^{j} \hat{U}_{s}^{j} |\varphi^{(\texttt{poly})}_{j} \rangle\| \\
        &+ \| \hat{U}_{\theta}^{j} \hat{U}_{s}^{j} |\varphi^{(\texttt{poly})}_{j} \rangle - \hat{U}_{\hat{\theta}}^{j} \hat{U}_{s}^{j} |\varphi^{(\texttt{poly})}_{j} \rangle\| \\
        &+ \| \hat{U}_{\hat{\theta}}^{j} \hat{U}_{s}^{j} |\varphi^{(\texttt{poly})}_{j} \rangle - \hat{U}_{\hat{\theta}}^{j} \hat{U}_{\hat{s}}^{j} |\varphi^{(\texttt{poly})}_{j} \rangle\|.
    \end{split}
    \label{app-eq:noise-error}
\end{equation}
Then, we proceed to simplify each term separately. 
\begin{enumerate}
    \item For the first term, we obtain 
    \begin{equation}
        \begin{split}
        \| U_{\theta}^{j} U_{s}^{j} |\phi^{(\texttt{poly})}_{j} \rangle - U_{\theta}^{j} U_{s}^{j} |\varphi^{(\texttt{poly})}_{j} \rangle\| &= \| U_{\theta}^{j} U_{s}^{j} \left(|\phi^{(\texttt{poly})}_{j} \rangle - |\varphi^{(\texttt{poly})}_{j} \rangle\right)\| \\
        &= \| |\phi^{(\texttt{poly})}_{j} \rangle - |\varphi^{(\texttt{poly})}_{j} \rangle\|\\& = \hat{\Delta}_j,
        \end{split}
    \end{equation}
    where we use the unitary invariance in the second line. 
    \item For the second term, we have 
    \begin{equation}
        \begin{split}
        \| U_{\theta}^{j} U_{s}^{j} |\varphi^{(\texttt{poly})}_{j} \rangle - \hat{U}_{\theta}^{j} U_{s}^{j} |\varphi^{(\texttt{poly})}_{j} \rangle\| &= \| \left(U_{\theta}^{j} - \hat{U}_{\theta}^{j} \right) U_{s}^{j} |\varphi^{(\texttt{poly})}_{j} \rangle\| \\
        & \leq \| U_{\theta}^{j} - \hat{U}_{\theta}^{j} \|_{\rm{op}} \cdot \| U_{s}^{j} |\varphi^{(\texttt{poly})}_{j} \rangle\| \\
        & \leq \| U_{\theta}^{j} - \hat{U}_{\theta}^{j} \|_{\rm{op}} \\&= \| e^{i \theta_j \phi_j^{(\texttt{poly})}} - e^{i \theta_j \varphi_j^{(\texttt{poly})}}\|_{\rm{op}} \\
        &\leq  |\theta_j| \; \| \phi_j^{(\texttt{poly})} - \varphi_j^{(\texttt{poly})}\|_{\rm{op}} \\
        & \leq 2 |\theta_j| \; \| | \phi_j^{(\texttt{poly})} \rangle - | \varphi_j^{(\texttt{poly})} \rangle\| \\&= 2 |\theta_j| \; \hat{\Delta}_j, 
        \end{split}
    \end{equation}
    where we use the fact that $\| A B\| \leq \| A\|\, \| B\|$ in the second line, and the inequality $\| e^A - e^B\| \leq \| A - B\|$ for unitary operators in the forth line. Moreover, in the last line, we use $\| \phi_j^{(\texttt{poly})} - \varphi_j^{\texttt{poly}}\|_{\rm{op}} \leq 2 \| | \phi_j^{(\texttt{poly})} \rangle - | \varphi_j^{\texttt{poly}} \rangle\|$. 
    \item The third term is give by 
    \begin{equation}
        \begin{split}
        \| \hat{U}_{\theta}^{j} U_{s}^{j} |\varphi^{(\texttt{poly})}_{j} \rangle - \hat{U}_{\theta}^{j} \hat{U}_{s}^{j} |\varphi^{(\texttt{poly})}_{j} \rangle\| &= \| \hat{U}_{\theta}^{j} \left( U_{s}^{j} - \hat{U}_{s}^{j} \right)|\varphi^{(\texttt{poly})}_{j} \rangle\| \\
        & \leq \| \hat{U}_{\theta}^{j} \|_{\rm{op}} \, \| U_{s}^{j} - \hat{U}_{s}^{j} \|_{\rm{op}} \\
        &\leq \| U_{s}^{j} - \hat{U}_{s}^{j} \|_{\rm{op}} ,
        \end{split}
    \end{equation}
    where we use a normalised state assumption and the property $\| A B\| \leq \| A\|\, \| B\|$ in the second line. 
    Moreover, using the definition of the operators in Eq.~\eqref{eq:u-s-def} and Eq.~\eqref{eq:hat-u-s-def}, adding and subtracting terms leads to
    \begin{equation}
    \begin{split}
        &\| U_{s}^{j} - \hat{U}_{s}^{j} \|_{\rm{op}}\\ &= \| e^{i\alpha_1 \sqrt{{s}_j}H} e^{i\alpha_2 \sqrt{{s}_j} \phi^{(\texttt{poly})}_{j}}  \cdots e^{i\alpha_{N_m - 1} \sqrt{{s}_j}H} e^{i\alpha_{N_m} \sqrt{{s}_j} \phi^{(\texttt{poly})}_{j}} - e^{i\alpha_1 \sqrt{{s}_j}H} e^{i\alpha_2 \sqrt{{s}_j} \varphi^{(\texttt{poly})}_{j}}  \cdots e^{i\alpha_{N_m - 1} \sqrt{{s}_j}H} e^{i\alpha_{N_m} \sqrt{{s}_j} \varphi^{(\texttt{poly})}_{j}} \|_{\rm{op}} \\
        & \leq \| e^{i\alpha_1 \sqrt{{s}_j}H} \left( e^{i\alpha_2 \sqrt{{s}_j} \phi^{(\texttt{poly})}_{j}} - e^{i\alpha_2 \sqrt{{s}_j} \varphi^{(\texttt{poly})}_{j}} \right) e^{i\alpha_3 \sqrt{{s}_j}H} \cdots e^{i\alpha_{N_m - 1} \sqrt{{s}_j}H} e^{i\alpha_{N_m} \sqrt{{s}_j} \phi^{(\texttt{poly})}_{j}} \|_{\rm{op}} \\
        &+ \| e^{i\alpha_1 \sqrt{{s}_j}H} e^{i\alpha_2 \sqrt{{s}_j}\varphi_j^{(\texttt{poly})}}  e^{i\alpha_3 \sqrt{{s}_j}H} \left( e^{i\alpha_4 \sqrt{{s}_j}\phi_j^{(\texttt{poly})}} - e^{i\alpha_4 \sqrt{{s}_j}\varphi_j^{(\texttt{poly})}}\right) \cdots e^{i\alpha_{N_m - 1} \sqrt{{s}_j}H} e^{i\alpha_{N_m} \sqrt{{s}_j} \phi^{(\texttt{poly})}_{j}} \|_{\rm{op}} + \cdots \\
        & + \| e^{i\alpha_1 \sqrt{{s}_j}H} e^{i\alpha_2 \sqrt{{s}_j} \varphi^{(\texttt{poly})}_{j}} e^{i\alpha_3 \sqrt{{s}_j}H} e^{i\alpha_4 \sqrt{{s}_j}\varphi^{(\texttt{poly})}_{j}}  \cdots e^{i\alpha_{N_m - 1} \sqrt{{s}_j}H} \left( e^{i\alpha_{N_m} \sqrt{{s}_j} \phi^{(\texttt{poly})}_{j}} - e^{i\alpha_{N_m} \sqrt{{s}_j} \varphi^{(\texttt{poly})}_{j}}\right) \|_{\rm{op}}
    \end{split}
    \end{equation}
    Then by using $\| AB\| \leq \| A\|\| B\|$ we can rewrite it as 
    \begin{equation}
    \begin{split}
        \| U_{s}^{j} - \hat{U}_{s}^{j} \|_{\rm{op}} &\leq  \| e^{i\alpha_2 \sqrt{{s}_j} \phi^{(\texttt{poly})}_{j}} - e^{i\alpha_2 \sqrt{{s}_j} \varphi^{(\texttt{poly})}_{j}} \|_{\rm{op}} 
        + \|  e^{i\alpha_4 \sqrt{{s}_j}\phi_j^{(\texttt{poly})}} - e^{i\alpha_4 \sqrt{{s}_j}\varphi_j^{(\texttt{poly})}} \|_{\rm{op}} + \cdots \\
        & + \| e^{i\alpha_{N_m} \sqrt{{s}_j} \phi^{(\texttt{poly})}_{j}} - e^{i\alpha_{N_m} \sqrt{{s}_j} \varphi^{(\texttt{poly})}_{j}} \|_{\rm{op}} \\
        &\leq \sum_{l=1}^{N_m/2} \| e^{i\alpha_{2l} \sqrt{{s}_j} \phi^{(\texttt{poly})}_{j}} - e^{i\alpha_{2l} \sqrt{{s}_j} \varphi^{(\texttt{poly})}_{j}} \|_{\rm{op}} \\
        &\leq \sum_{l=1}^{N_m/2}  | \alpha_{2l} \sqrt{{s}_j} | \, \| \phi^{(\texttt{poly})}_{j} - \varphi^{(\texttt{poly})}_{j}\|_{\rm{op}} \\ 
        &\leq \sum_{l=1}^{N_m/2}  2| \alpha_{2l}| \,  \sqrt{{s}_j}  \, \| |\phi^{(\texttt{poly})}_{j}\rangle - |\varphi^{(\texttt{poly})}_{j} \rangle \| 
    \end{split}
    \end{equation}
    Defining $\alpha_{\rm{max}} = \max_{i} \{ |\alpha_i| \}_i$, we have
    \begin{equation}
        \| \hat{U}_{\theta}^{j} U_{s}^{j} |\varphi^{(\texttt{poly})}_{j} \rangle - \hat{U}_{\theta}^{j} \hat{U}_{s}^{j} |\varphi^{(\texttt{poly})}_{j} \rangle\| \leq N_m \alpha_{\rm{max}} \sqrt{s_j} \hat{\Delta}_j
    \end{equation}
    \item For the fourth term, similar to the second term, we have 
    \begin{align}
        \| \hat{U}_{\theta}^{j} \hat{U}_{s}^{j} |\varphi^{(\texttt{poly})}_{j} \rangle - \hat{U}_{\hat{\theta}}^{j} \hat{U}_{s}^{j} |\varphi^{(\texttt{poly})}_{j} \rangle\| & \leq \| \hat{U}_{\theta}^{j} - \hat{U}_{\hat{\theta}}^{j}\|_{\rm{op}}.
    \end{align}
    Then since $\hat{U}_{{\theta}}^{j} = e^{i {\theta}_j \varphi_{j}^{(\texttt{poly})}}$ and $\hat{U}_{\hat{\theta}}^{j} = e^{i \hat{\theta}_j \varphi_{j}^{(\texttt{poly})}}$, using the fact that $\varphi_{j}^{(\texttt{poly})}$ is a pure state, we have
    \begin{align}
        \| \hat{U}_{\theta}^{j} - \hat{U}_{\hat{\theta}}^{j}\|_{\rm{op}} &= \| \left( \mathds{1} + (e^{i \theta_j} - 1)\varphi_{j}^{(\texttt{poly})}\right) - \left( \mathds{1} + (e^{i \hat{\theta}_j} - 1)\varphi_{j}^{(\texttt{poly})}\right)\|_{\rm{op}} \\
        & = |e^{i {\theta_j}} - e^{i \hat{\theta}_j}|
    \end{align}
    Consequently, we obtain 
    \begin{align}
        \| \hat{U}_{\theta}^{j} \hat{U}_{s}^{j} |\varphi^{(\texttt{poly})}_{j} \rangle - \hat{U}_{\hat{\theta}}^{j} \hat{U}_{s}^{j} |\varphi^{(\texttt{poly})}_{j} \rangle\| & \leq |e^{i {\theta}_j} - e^{i \hat{\theta}_j}| = \delta \theta_j. 
    \end{align}
    \item Finally, for the fifth term, we have
    \begin{equation}
    \begin{split}
        \| \hat{U}_{\hat{\theta}}^{j} \hat{U}_{s}^{j} |\varphi^{(\texttt{poly})}_{j} \rangle - \hat{U}_{\hat{\theta}}^{j} \hat{U}_{\hat{s}}^{j} |\varphi^{(\texttt{poly})}_{j} \rangle\| &= \| \hat{U}_{\hat{\theta}}^{j} \left( \hat{U}_{s}^{j} - \hat{U}_{\hat{s}}^{j} \right) |\varphi^{(\texttt{poly})}_{j} \rangle\| \\
        & \leq \| \hat{U}_{s}^{j} - \hat{U}_{\hat{s}}^{j}\|_{\rm{op}}. 
    \end{split}
    \end{equation}
    Then, we can upper bound $\| \hat{U}_{s}^{j} - \hat{U}_{\hat{s}}^{j}\|$ similar to the third term. According to Eq.~\eqref{eq:hat-u-s-def} and Eq.~\eqref{eq:hat-u-hat-s-def}, we have
    \begin{equation}
    \begin{split}
        &\| \hat{U}_{s}^{j} - \hat{U}_{\hat{s}}^{j} \|_{\rm{op}} \\
        &= \| e^{i\alpha_1 \sqrt{{s}_j}H} e^{i\alpha_2 \sqrt{{s}_j} \varphi^{(\texttt{poly})}_{j}}  \cdots e^{i\alpha_{N_m - 1} \sqrt{{s}_j}H} e^{i\alpha_{N_m} \sqrt{{s}_j} \varphi^{(\texttt{poly})}_{j}} - e^{i\alpha_1 \sqrt{{\hat{s}}_j}H} e^{i\alpha_2 \sqrt{{\hat{s}}_j} \varphi^{(\texttt{poly})}_{j}}  \cdots e^{i\alpha_{N_m - 1} \sqrt{{\hat{s}}_j}H} e^{i\alpha_{N_m} \sqrt{{\hat{s}}_j} \varphi^{(\texttt{poly})}_{j}} \|_{\rm{op}} \\
        & \leq \| \left( e^{i\alpha_1 \sqrt{{s}_j}H} - e^{i\alpha_1 \sqrt{{\hat{s}}_j}H} \right)  e^{i\alpha_2 \sqrt{{s}_j} \varphi^{(\texttt{poly})}_{j}} e^{i\alpha_3 \sqrt{{s}_j}H} \cdots e^{i\alpha_{N_m - 1} \sqrt{{s}_j}H} e^{i\alpha_{N_m} \sqrt{{s}_j} \varphi^{(\texttt{poly})}_{j}} \|_{\rm{op}} \\
        &+ \| e^{i\alpha_1 \sqrt{{\hat{s}}_j}H} \left( e^{i\alpha_2 \sqrt{{s}_j} \varphi^{(\texttt{poly})}_{j}} - e^{i\alpha_2 \sqrt{{\hat{s}}_j} \varphi^{(\texttt{poly})}_{j}} \right) e^{i\alpha_3 \sqrt{{s}_j}H} \cdots e^{i\alpha_{N_m - 1} \sqrt{{s}_j}H} e^{i\alpha_{N_m} \sqrt{{s}_j} \varphi^{(\texttt{poly})}_{j}} \|_{\rm{op}} \\
        &+\|  e^{i\alpha_1 \sqrt{{\hat{s}}_j}H}  e^{i\alpha_2 \sqrt{{\hat{s}}_j} \varphi^{(\texttt{poly})}_{j}} \left( e^{i\alpha_3 \sqrt{{s}_j}H} - e^{i\alpha_3 \sqrt{{\hat{s}}_j}H}\right) e^{i\alpha_4 \sqrt{{s}_j}\varphi_j^{(\texttt{poly})}} \cdots e^{i\alpha_{N_m - 1} \sqrt{{s}_j}H} e^{i\alpha_{N_m} \sqrt{{s}_j} \varphi^{(\texttt{poly})}_{j}} \|_{\rm{op}} + \cdots \\
        & + \| e^{i\alpha_1 \sqrt{{\hat{s}}_j}H} e^{i\alpha_2 \sqrt{{\hat{s}}_j} \varphi^{(\texttt{poly})}_{j}} e^{i\alpha_3 \sqrt{{\hat{s}}_j}H} e^{i\alpha_4 \sqrt{{\hat{s}}_j}\varphi^{(\texttt{poly})}_{j}}  \cdots e^{i\alpha_{N_m - 1} \sqrt{{\hat{s}}_j}H} \left( e^{i\alpha_{N_m} \sqrt{{s}_j} \varphi^{(\texttt{poly})}_{j}} - e^{i\alpha_{N_m} \sqrt{{\hat{s}}_j} \varphi^{(\texttt{poly})}_{j}}\right) \|_{\rm{op}}. 
    \end{split}
    \end{equation}
    Again, using $\| A B\| \leq \|A \| \, \|B \|$, we have 
    \begin{equation}
    \begin{split}
        \| \hat{U}_{s}^{j} - \hat{U}_{\hat{s}}^{j} \|_{\rm{op}} &\leq \sum_{l = 1}^{N_m/2} \| e^{i\alpha_{2l - 1} \sqrt{{{s}}_j}H} - e^{i\alpha_{2l - 1} \sqrt{{\hat{s}}_j}H} \|_{\rm{op}} + \sum_{l = 1}^{N_m/2} \| e^{i\alpha_{2l} \sqrt{{{s}}_j}\varphi_j^{(\texttt{poly})}} - e^{i\alpha_{2l} \sqrt{{\hat{s}}_j}\varphi_j^{(\texttt{poly})}} \|_{\rm{op}} \\
        & \leq \sum_{l =1}^{N_m/2}  \left( |\alpha_{2l-1}| \, |\sqrt{s_j} - \sqrt{\hat{s}_j}| \, \|H \|_{\rm{op}} + |\alpha_{2l}| \, |\sqrt{s_j} - \sqrt{\hat{s}_j}| \, \|\varphi_j^{(\texttt{poly})} \|_{\rm{op}}\right) \\
        & \leq \sum_{l=1}^{N_m/2} 2\, |\sqrt{s_j} - \sqrt{\hat{s}_j}| \, \alpha_{\rm{max}} \\
        & = N_m  \alpha_{\rm{max}} |\sqrt{s_j} - \sqrt{\hat{s}_j}|,
    \end{split}
    \end{equation}
    where $\alpha_{\rm{max}} = \max_{i} \{ |\alpha_i | \}_i$ and we use $\| H\|\leq 1$ in the third line. Now we need to bound $|\sqrt{s_j} - \sqrt{\hat{s}_j}|$ with a function of $|{s_j} - {\hat{s}_j}| = \delta s_k$. Note that, assuming $a , b \geq 0$, then we can write 
    \begin{equation}
        a - b = (\sqrt{a} - \sqrt{b})(\sqrt{a} + \sqrt{b}).
        \label{app:eq-a-b}
    \end{equation}
    Taking the absolute value from both sides of Eq.~\eqref{app:eq-a-b}, assuming $a \neq b \neq 0$, we have
    \begin{equation}
        |\sqrt{a} - \sqrt{b}| = \frac{|a - b|}{|\sqrt{a} + \sqrt{b}|}.
        \label{app:eq-a-b-sqrt-abs}
    \end{equation}
    Now in order to upper bound $|\sqrt{a} - \sqrt{b}|$, we lower bound $|\sqrt{a} + \sqrt{b}|$. From triangle inequality, since $a , b\geq 0$, we have $|a - b| \leq a + b$. Adding the positive term $2\sqrt{ab}$ to $|a-b|$, we have 
    \begin{equation}
        |a - b| \leq a + b \leq a + b + 2 \sqrt{ab} = (\sqrt{a} + \sqrt{b})^2, 
    \end{equation}
    where by taking the square root of both sides, we have $\sqrt{|a - b|} \leq |\sqrt{a} + \sqrt{b}|$. Now, using this result in Eq.~\eqref{app:eq-a-b-sqrt-abs}, we can write
    \begin{equation}
        |\sqrt{a} - \sqrt{b}| \leq \sqrt{|a - b|}.
        \label{app:eq-a-b-sqrt-delta}
    \end{equation}
    According to equation Eq.~\eqref{app:eq-a-b-sqrt-delta}, by replacing $a$ and $b$ with $s_j$ and $\hat{s}_j$, respectively, we have
    \begin{equation}
        |\sqrt{s}_j - \sqrt{\hat{s}_j}| \leq \sqrt{|s_j - \hat{s}_j|} = \sqrt{\delta s_j}.
    \end{equation}
    Finally, we have the following upper bound for the fifth term. 
    \begin{equation}
        \| \hat{U}_{\hat{\theta}}^{j} \hat{U}_{s}^{j} |\varphi^{(\texttt{poly})}_{j} \rangle - \hat{U}_{\hat{\theta}}^{j} \hat{U}_{\hat{s}}^{j} |\varphi^{(\texttt{poly})}_{j} \rangle\| \leq N_m  \alpha_{\rm{max}} |\sqrt{s_j} - \sqrt{\hat{s}_j}| \leq N_m  \alpha_{\rm{max}} \sqrt{\delta s_j}
    \end{equation}
\end{enumerate}
Collecting all the terms, we can rewrite Eq.~\eqref{app-eq:noise-error} as 
\begin{equation}
\begin{split}
    \hat{\Delta}_{j+1} &= \| | \phi_{j+1}^{(\texttt{poly})}\rangle - | \varphi_{j+1}^{(\texttt{poly})}\rangle \|  \\
    & \leq \left( 1 + 2 |\theta_j| + N_m \alpha_{\rm{max}} \sqrt{s_j} \right) \hat{\Delta}_{j} + \delta \theta_j + N_m \alpha_{\rm{max}} \sqrt{\delta s_j}\\
    & \leq \left( 14 + \frac{2 N_m \alpha_{\rm{max}}}{\sqrt{V_j}} \right) \hat{\Delta}_{j} + \delta \theta_j + N_m \alpha_{\rm{max}} \sqrt{\delta s_j}
\end{split}
\label{app-eq:nose-error-delta-k}
\end{equation}
where we use $|\theta_j |\leq 2 \pi$ and $s_j \leq \pi/2\sqrt{V_j} \leq 2 / \sqrt{V_j}$ in the last line. 
Define $\eta = \max_j\{ 1/\sqrt{V_j}\}_k$ and $\delta = \max_j\{ \delta \theta_j , \sqrt{\delta s_j}\}_j$. Solving the iterative equation for $j=k$, we have 
\begin{equation}
    \begin{split}
        \hat{\Delta}_{k+1} &\leq (1 + N_m \alpha_{\rm{max}}) \delta\sum_{i=0}^{k} \left( 14 + 2 N_m \alpha_{\rm{max}}\eta\right)^i \\
        & = (1 + N_m \alpha_{\rm{max}})\delta \frac{1 - \left( 14 + 2 N_m \alpha_{\rm{max}}\eta \right)^{k+1}}{1 - \left( 14 + 2 N_m \alpha_{\rm{max}}\eta\right)} \\
        &\leq (1 + N_m \alpha_{\rm{max}})\delta \frac{\left( 14 + 2 N_m \alpha_{\rm{max}}\eta\right)^{k+1}}{ 13 + 2 N_m \alpha_{\rm{max}} \eta} \\
        &\leq \frac{2 N_m \alpha_{\rm{max}}\delta}{13} \left( 14+2 N_m \alpha_{\rm{max}} \eta\right)^{k+1}.
    \end{split}
\end{equation}
where we use $N_m \alpha_{\rm{max}} \geq 1$ in the last line. 

Finally, we have the error as a function of $\delta$ and $\eta$. However, to find the final bound, we need to rewrite the bound as a function of energy and variance and the errors in the estimation of energy and variance. To this aim, we use the bound from App.~\ref{app_statistical_noise_propagation}. Using Eq.~\eqref{app-eq-delta-s-error} and Eq.~\eqref{app-eq-delta-theta-error}, we can have
\begin{equation}
    \begin{split}
        \delta \theta_j &\leq 2 \eta_{E_j} \delta E_j, \\
        \sqrt{\delta s_j} &\leq  \eta_{V_j} \sqrt{\delta E_j} + \eta_{V_j}^{3/2} \sqrt{\delta V_j},
    \end{split}
\end{equation}
where $\eta_{E_j} = \max \{ \frac{1}{|E_j - z|}, \frac{1}{|\hat{E}_j - z|}\}$ and $\eta_{V_j} = \max\{ \frac{1}{\sqrt{V_j}}, \frac{1}{\sqrt{\hat{V}_j}}\}$. Moreover, we can define $\eta_E = \max_k\{\eta_{E_k}\}$ and $\eta_V = \max_k\{\eta_{V_k}\}$.
Since $\delta = \max_k\{ \delta \theta_k , \sqrt{\delta s_k}\}$, we can write it as a function of energy and variance by 
\begin{equation}
    \delta = \max_j\{ 2 \eta_{E_j} \delta E_j , \eta_{V_j} \sqrt{\delta E_j} + \eta_{V_j}^{3/2} \sqrt{\delta V_j}\} = \max\left\{2\eta_{(E,V)} \delta_{(E,V)}, \eta_{(E,V)} \sqrt{\delta_{(E,V)}}(1 + \sqrt{\eta_{(E,V)}})\right\},
\end{equation}
where $\eta_{(E,V)} = \max\{\eta_{E_k}, \eta_{V_k}\}$, $\delta_{(E,V)} = \max_k\{ \delta E_k, \delta V_k\}$. 

Finally, defining $\gamma_m = N_m \alpha_{\rm{max}}$, we have 
\begin{equation}
        \hat{\Delta}_k = \|| \varphi_k^{(\texttt{poly})} \rangle-| \phi_k^{(\texttt{poly})} \rangle\| \leq \frac{2   \gamma_m }{13} \left( 14 + 2 \gamma_m \eta_{(E,V)}\right)^{k} \delta.
\end{equation}
\end{proof}

\section{Details on numerical simulations} \label{app:numerics}

In this section, we describe the details of our numerical simulations. 
In particular, we provide (1) the numerical implementation of \texttt{DB-TFD}, which implements the exponential of exact commutators, (2) its product-formula approximation, (3) the setup for generative modeling tasks, and (4) the additional numerical results.

\subsection{Numerical implementation of DB-TFD with exponential of exact commutators}

Both the \texttt{vanilla} \texttt{DB-TFD} and the \texttt{poly} \texttt{DB-TFD} require the exponential of commutators, and hence we first elaborate on their numerical implementation.
To implement the exponential of commutators, we utilize the Python package SciPy~\cite{2020SciPy-NMeth}, taking advantage of its sparse matrix functionalities to significantly reduce computational cost.

As shown in Ref.~\cite{suzuki2025doublebracketalgorithmquantumsignal}, the exponential of commutators applied to an initial state $\ket{\psi}$ can be written as the corresponding linear polynomial, i.e., 
\begin{equation} \label{app_eq: dbi_linear_poly}
    e^{s[\psi,H]}|\psi \rangle = \left(\frac{E -H}{\sqrt{V}}\sin(s\sqrt{V}) +\cos(s\sqrt{V})\right) |\psi\rangle,
\end{equation}
where $E=\braket{\psi|H|\psi}$ and $V=\braket{\psi|(H-E)^2|\psi}$.
Thus, utilizing the expression, we can avoid direct exponentiation of the matrix and only require simple matrix-vector operations. In the case of a sparse Hamiltonian, this can lead to substantial speed-ups.

Explicitly, we define three vectors:
\begin{enumerate}
    \item $|\psi\rangle$
    \item $|h\rangle = H|\psi\rangle$
    \item $|v\rangle = |h\rangle -E |\psi\rangle$
\end{enumerate}
Then, the exponential of the commutator in Eq.~\eqref{app_eq: dbi_linear_poly} reads
\begin{equation}
    e^{s[\psi,H]}|\psi \rangle = \cos(s\sqrt{V})|\psi\rangle -\frac{\sin(s\sqrt{V})}{\sqrt{V}}|v\rangle.
\end{equation}
Thus, this computation requires only a single matrix-vector multiplication and two inner products. 

\medskip
For the \texttt{poly} \texttt{DB-TFD} approach, the roots of the polynomial are required to factorize it and determine the corresponding time steps. These can be efficiently computed using SciPy’s root-finding functions. If the roots are complex, an additional reflection operator $e^{i\theta \psi}$ must be included. This adds minimal computational overhead, since
\begin{equation}
    e^{i\theta \psi} = I +(e^{i\theta}-1)\psi.
    \label{eq:reflection_expansion}
\end{equation}
Applying this to the exponential of commutators in Eq.~\eqref{app_eq: dbi_linear_poly} gives
\begin{equation}
    e^{i\theta\psi}e^{s[\psi,H]}|\psi\rangle = e^{s[\psi,H]}|\psi\rangle +(e^{i\theta}-1) \psi e^{s[\psi,H]}|\psi\rangle.
\end{equation}
The second term reduces to $(e^{i\theta}+1) \braket{\psi|e^{s[\psi,H]}|\psi} |\psi\rangle$, where
\begin{equation}
    \braket{\psi|e^{s[\psi,H]}|\psi} = \left\langle \psi\Bigg| \frac{E -H}{\sqrt{V}}\sin(s\sqrt{V}) +\cos(s\sqrt{V}) \Bigg|\psi\right\rangle = \cos(s\sqrt{V}).
\end{equation}
Thus, the unitary for each step of the \texttt{poly} \texttt{DB-TFD} is 
\begin{equation}
    e^{i\theta\psi}e^{s[\psi,H]}|\psi\rangle = \cos(s\sqrt{V})|\psi\rangle -\frac{\sin(s\sqrt{V})}{\sqrt{V}}|v\rangle + (e^{i\theta}-1)\cos(s\sqrt{V})|\psi\rangle = e^{i\theta}\cos(s\sqrt{V})|\psi\rangle -\frac{\sin(s\sqrt{V})}{\sqrt{V}}|v\rangle,
\end{equation}
i.e., the multiplication of $\cos(s\sqrt{V})$ is modified by the phase $e^{i\theta}$.

\subsection{Numerical implementation of the product-formula approximation}

An $m$-th order product-formula approximation of the exponential of commutators is given by
\begin{equation}
    e^{s^2[A,B]} = e^{t_1A}e^{t_2B}e^{t_3A}e^{t_4B}... = e^{s^2[A,B]} +O(x^{m+1}),
\end{equation}
where $t_k = \alpha_k s$ with a specific coefficient $\alpha_k$.

Explicit formulas for second- and third-order product formulas, as well as recursive constructions for higher-order forms, are provided in~\cite{Chen_2022}. The second-order formula is given by
\begin{equation}
    e^{sA}e^{sB}e^{-sA}e^{-sB} = e^{s^2[A,B]} +O (s^3),
\end{equation}
while the third-order formula reads
\begin{equation}
     e^{\phi sA}e^{\phi sB}e^{-sA}e^{-(\phi+1)sB}e^{(1-\phi )sA}e^{sB} = e^{s^2[A,B]} +O (s^4),
\end{equation}
with $\phi=\frac{1+\sqrt{5}}{2}$.
Higher-order product formulas can be constructed recursively. Denoting an $n$-th order product formula by $f_n(s)$, the recursion reads
\begin{itemize}
    \item  For even $n$:
    \begin{equation}
        f_{n+1}(s) = f_{n}(s/\sqrt{2})f_{n}(-s/\sqrt{2}).
    \end{equation}
    \item For odd $n$:
    \begin{equation}
        f_{n+1}(s) = f_n(\nu s)^2f_n(\mu s)^{-1}f_n(\nu s)^2,
    \end{equation}
    where $\mu = \sqrt{4\sigma}$, $\nu =\sqrt{1/4+\sigma}$ and $\sigma =  4^{2/(n+1)}/(4(4-4^{2/(n+1)}))$.
\end{itemize}
We note that alternative recursive schemes exist that employ different numbers of copies and distinct constant factors, while preserving the formal order of the approximation.

For the \texttt{DB-TFD} implementation, $A$ and $B$ are replaced by $i\psi = i|\psi\rangle \langle \psi|$ and $iH$, respectively.
That is, for a pure state $\psi$ and an Hamiltonian $H$, the implementation becomes the form of
\begin{equation}
    e^{ia\psi}e^{ibH}|\psi\rangle
\end{equation}
for $a,b\in\mathbb{R}$.
The reflection gate $\exp(ia\psi)$ can be applied without explicitly exponentiating the density matrix, as shown in Eq.~\eqref{eq:reflection_expansion}, requiring only an inner product.
Moreover, due to the sparsity of the Hamiltonian, the exponential of $H$ can be implemented efficiently.

\subsection{Numerical setting for generative modeling tasks}

We provide the details of the numerical setups used for generative modeling tasks in the main text.
First, we briefly review quantum generative modeling using quantum Boltzmann machines.
Next, we describe the specific tasks considered in this study, followed by a short overview of Variational Quantum Imaginary Time Evolution (VarQITE)~\cite{McArdle_2019, Zoufal_2021}.
Finally, we present the setup employed for noisy simulations.

\subsubsection{Generative modeling tasks using quantum Boltzmann machines} \label{app:numerics-generative-models}

We here give a summary of the generative modeling task used for our numerical results. 
The goal of generative modeling tasks is to generate the target distribution or density matrix $\eta$ by training a parameterized quantum state. 
Specifically, the parametrized state is used to approximate the thermal state 
\begin{equation}
    \rho_{\bm{\theta}} = \frac{e^{- H_{\bm{\theta}}}}{Z} 
\end{equation}
of a given parametrized Hamiltonian $H_{\bm{\theta}} = \sum_{i=1}^{M} \theta_i G_i$, with Hermitian $G_i$ operators and parameters $\bm{\theta} = (\theta_1, \theta_2, \dots, \theta_M)$.
Here, $Z = \rm{Tr}{(e^{- H_{\bm{\theta}}})}$.

To obtain the target distribution, a common approach is to employ the gradient descent method, where the quantum relative entropy given by 
\begin{equation}
    S(\eta\| \rho_{\bm{\theta}}) = \rm{Tr}(\eta \log \eta) - \bm{Tr}(\eta \log \rho_{\bm{\theta}}),
\end{equation}
is used as a cost function.
In particular, the gradient can be written as expectation values of the Hermitian terms in the Hamiltonian;
\begin{equation}
    \frac{\partial S(\eta||\rho_\theta)}{\partial \theta_i} = \langle G_i \rangle_{\rho_{\bm{\theta}}} -\langle G_i \rangle_{\eta},
\end{equation}
where $\braket{G}_{\Psi}:=\mathrm{Tr}[G\Psi]$.
Consequently, parameter updates require estimating the expectation values of all generators $\{G_{i}\}$ with respect to the model state $\rho_{\bm{\theta}}$.

For all tasks used in this work, we parametrize the Hamiltonian $H_{\theta}$ using a mean-field term and a two-body fully connected term:
\begin{equation} \label{app_eq:parameterized_hamiltonian}
    H_{\theta} = \sum_{k=x,y,z} \left(\sum_{i, j}\lambda _{i,j} \sigma_i^k\sigma_j^k + \sum_i \gamma_i \sigma_i^k\right), 
\end{equation}
where $\sigma^k_i$ corresponds to one of the Pauli matrices at site $i$.
Here, $\lambda_{i,j}$ and $\gamma_i$ correspond to the parameters to be learned.

\subsubsection{A quick overview on Variational Quantum Imaginary Time Evolution (VarQITE)}

Variational Quantum Imaginary Time Evolution (VarQITE)~\cite{McArdle_2019, Zoufal_2021} is a variational approach for implementing imaginary-time evolution.
Let $\bm{\theta}(\tau)$ denote the parameters of a parametrized quantum circuit at step $\tau$ and let the corresponding state be $|\psi(\bm \theta(\tau))\rangle$.
Within this framework, the McLachlan variational principle~\cite{Yuan_2019} is employed to determine the updated parameters at $\tau+\delta \tau$ such that the dynamics generated by the parameters of a parametrized quantum circuit approximate imaginary-time evolution.

Concretely, this principle minimizes the norm of the residual between the infinitesimal imaginary-time evolution and the variational change induced by the parametrized quantum circuit, i.e.,
\begin{equation}
    \delta \left\|\left(\partial/\partial_\tau - (E_\tau-H)\right)|\psi(\bm\theta(\tau))\rangle\right\| = 0,
\end{equation}
where $(E_ \tau -H)$ corresponds to the first-order approximation of imaginary-time evolution.

By applying the McLachlan variational principle, one obtains a system of differential equations $A\bm{\theta}=C$ that governs the parameter dynamics~\cite{McArdle_2019}, where
\begin{equation}
    A_{ij} = \Re{ \left\langle \frac{\partial \psi(\bm{\theta})}{\partial \theta_i} \middle| \frac{\partial \psi(\bm{\theta})}{\partial \theta_j} \right \rangle},
    \label{eq:Amatrix}
\end{equation}
and the vector $C$ is defined as
\begin{equation}
    C_i = - \Re{ \left\langle \frac{\partial \psi(\bm{\theta})}{\partial \theta_i} \middle| H \middle| \psi(\bm{\theta})\right \rangle}.
    \label{eq:Cvector}
\end{equation}
Once the system is solved, the parameters can be updated using standard numerical ODE methods, such as Euler's rule:
\begin{equation}
    \bm{\theta}_{\tau +\Delta \tau} \approx \bm{\theta}_{\tau} + \Delta \tau \bm{A}^{-1} \bm{C}.
\end{equation}
In the limit of infinitesimal time steps, and provided that the parametrized quantum circuit is sufficiently expressive, this procedure can faithfully reproduce the dynamics of exact imaginary-time evolution. 
In practice, however, the matrix $A$ may be singular or ill-conditioned, necessitating the use of pseudo-inverses or regularization techniques to ensure numerical stability.

Given this procedure, we estimate the cost associated with implementing VarQITE. Specifically, we quantify three main components: (1) the circuit depth, (2) the number of repetitions required to execute the quantum circuits, and (3) the number of iterations $k$ needed to reach a target imaginary-time evolution at time $\beta$.
For comparison, we also summarize the corresponding costs for \texttt{DB-TFD}-based approaches.

\begin{itemize}
    \item \textbf{Circuit Depth:}
    VarQITE does not adaptively increase circuit depth during the optimization. Once the parametrized quantum circuit is fixed at the outset, the circuit depth remains constant throughout the training process.

    In contrast, for \texttt{DB-TFD} approaches, the circuit depth grows significantly with the number of iterations $k$. In particular, the depth scales doubly exponentially for the \texttt{vanilla} \texttt{DB-TFD} scheme, while it scales super-exponentially for the \texttt{poly} \texttt{DB-TFD} construction in the worst case.
    \item \textbf{Number of repetition to execute circuits:}
    Updating the parameters in VarQITE requires estimating all entries of the matrix $A$ and the vector $C$. 
    Estimating the matrix $A$ scales quadratically with the number of variational parameters $L$, while estimating $C$ scales linearly with both $L$ and the number of Hamiltonian terms $M$, assuming a decomposition $H = \sum_{i=1}^{M} \theta_i G_i $. Denoting the number of VarQITE iterations by $k^{({\rm V})}$, the total number of circuit repetitions scales as  $\mathcal{O}((L^2+LM)k^{({\rm V})})$.

    For the \texttt{vanilla} \texttt{DB-TFD} approach, thermal state preparation does not require repeated circuit executions, leading to a constant repetition cost $\mathcal{O}(1)$.

    For the \texttt{poly} \texttt{DB-TFD} approach, the implementation requires estimating both the energy and the variance of the Hamiltonian. 
    When the Hamiltonian contains $M$ terms, these estimations scale as $\mathcal{O}(M)$ and $\mathcal{O}(M^2)$, respectively. 
    With a polynomial degree $k^{(\rm poly)}$, the total number of circuit repetitions therefore scales as $\mathcal{O}((M^2+M)k^{(\rm poly)})$.
    \item \textbf{Number of iteration $k$:}
    VarQITE relies on a first-order approximation of imaginary-time evolution, where the error introduced at each step propagates into subsequent parameter updates. 
    To control this accumulated error, the step size must be chosen sufficiently small, which can lead to an exponentially large number of iterations in the worst case.

    For the \texttt{vanilla} \texttt{DB-TFD} approach, Lemma~\ref{theorem:naive_error} similarly implies that the required number of iterations grows exponentially with the target evolution time.

    In contrast, for the \texttt{poly} \texttt{DB-TFD} approach, Lemma~\ref{theorem:poly_error} shows that the number of iterations scales as $\mathcal{O}(\sqrt{\beta})$, which is smaller than the other approaches.
\end{itemize}

These comparisons highlight the trade-offs among the different approaches. When circuit depth is the primary constraint, VarQITE is the most suitable choice, as it maintains a fixed depth throughout the evolution. 
However, suppressing the error accumulated at each step may require an exponentially large number of circuit repetitions.
From the perspective of circuit repetitions, the \texttt{vanilla} \texttt{DB-TFD} approach is the most efficient, as it requires only a constant number of executions. This advantage, however, comes at the cost of a doubly exponential growth in circuit depth.

Regarding the number of iterations, the \texttt{poly} \texttt{DB-TFD} approach performs best.
This scaling may suggest robustness to noise compared to the other methods; that is, it requires significantly less circuit depth than \texttt{vanilla} \texttt{DB-TFD}, demanding fewer circuit repetitions than VarQITE.
In fact, numerical results further indicate that \texttt{poly} \texttt{DB-TFD} is more robust than VarQITE, supporting this interpretation.

Moreover, although the \hiddenlink{poly-DBTFD}{\texttt{poly} \texttt{DB-TFD}} formally depends on the polynomial degree as a bottleneck, practical implementations often require only a low degree. In particular, our observations show that a degree as small as four already yields accurate results.
Taken together, these findings suggest that \texttt{poly} \texttt{DB-TFD} is a promising candidate for near-term and early fault-tolerant quantum devices.

\subsubsection{Types of tasks considered in this study}

To benchmark the performance of our proposal, we consider three representative tasks: the XXZ model, the Bernoulli distribution, and the bar-and-stripes dataset of size $2\times 2$.

\begin{itemize}
    \item \textbf{XXZ model} --
    The XXZ model is a paradigmatic spin-$1/2$ lattice Hamiltonian with anisotropic nearest-neighbor interactions, widely used as a benchmark in quantum many-body physics.
    The Hamiltonian is written as
    \begin{equation}
        H= \sum_{i=1}^{N-1} J (X_iX_{i+1}+Y_iY_{i+1} ) + \Delta Z_iZ_{i+1} + h\sum_i^N Z_i,
    \end{equation}
    where $J$ and $\Delta$ correspond to nearest-neighbor interactions and $h$ an external magnetic field.
    For the target model, we set the values of these parameters to the same as in Ref.~\cite{Coopmans_2024}, that is, $J = -0.5$, $\Delta = -0.7$, $h = -0.8$.
    
    \item \textbf{The Bernoulli distribution} --
    The Bernoulli distribution corresponds to a parametrized two-outcome distribution where the probability to obtain $1$ or $0$ are given by $p$ and $1-p$, respectively. This is encoded by the following quantum state:
    \begin{equation}
        |\psi\rangle = \sqrt{1-p}|0\rangle + \sqrt{p}|1\rangle.
    \end{equation}
    We consider the task of learning a product-state vector in which each qubit is independently prepared according to a Bernoulli distribution with probability $p_i$. 
    For $n$ qubits, the target quantum state is given by
    \begin{equation}
        |\psi\rangle = \bigotimes_{i}^n \left(\sqrt{1-p_i}|0\rangle + \sqrt{p_i}|1\rangle \right).
    \end{equation}
    \item \textbf{Bar-and-stripes dataset} -- 
      Given a black and white picture composed of $n\times m$ pixels, we can encode the color of a given pixel by a variable that takes either $1$ or $0$. We then consider the set of all pictures that have horizontal stripes or vertical bars. The $2\times 2$ set is shown in Figure~$\ref{fig:bar_and_stripes}$. Each $n\times m$ image can be encoded as an $nm$ qubit state, and the goal is to learn the state corresponding to the Bar-and-stripes set.
    
\end{itemize}
\begin{figure}
        \centering
        \includegraphics[width=0.5\linewidth]{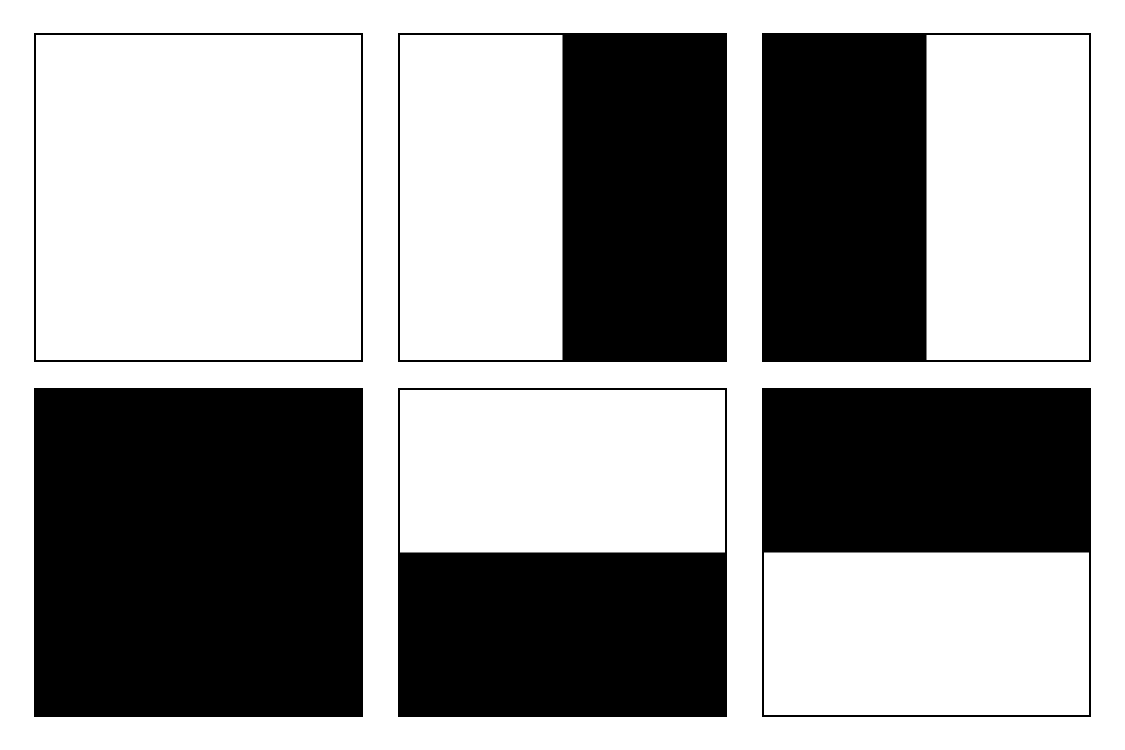}
        \caption{\textbf{The Bar and Stripes $2 \times 2$ data set}}
        \label{fig:bar_and_stripes}
\end{figure}

For all these tasks, we use the parametrized Hamiltonian shown in Eq.~\eqref{app_eq:parameterized_hamiltonian}.

In the numerical simulations, a hardware-efficient parameter of a parametrized quantum circuit is used for VarQITE separated into two subsystems $A$ and $B$ in order to prepare the thermofield double state. For system A, the circuit consists of a layer of single-qubit rotations with the Pauli $X$, $Y$, and $Z$ applied sequentially, followed by a layer of two-qubit entangling gates arranged in an alternating pattern using $XX$, $YY$, and $ZZ$ interactions between all qubits. This is basically a trotterized version of the Hamiltonian in $\eqref{app_eq:parameterized_hamiltonian}$. Then $XX$, $YY$, and $ZZ$ gates are applied to the corresponding pairs between the $A$ and $B$ subsystems.
Circuits with one to five such layers are examined, where we consider the previous construction as a single circuit layer.
In the main text, we show the best performance among the trials.
Indeed, two layers consistently provide the best results, and thus all numerical results presented in the main text correspond to the two-layer circuit.

The initial state is always chosen to be the maximally entangled state. To prepare this state, an additional set of parametrized gates is introduced. These gates are initialized with a fixed parameter choice that prepares the state $|\text{TFD}(0)\rangle$. The parameters of the other layers are set to zero at initialization.

To compute the entries of the matrix $A$ in Eq.~\eqref{eq:Amatrix} and the vector $C$ in Eq.~\eqref{eq:Cvector} for VarQITE efficiently and exactly, we employ the PennyLane~\cite{bergholm2022pennylaneautomaticdifferentiationhybrid} package, which is used for both circuit construction and execution.
For the regularization to invert the matrix $A$, we introduce a regularization parameter $\lambda$. 
This value is chosen from the interval $[10^{-4},10^{-1}]$ such that the resulting dynamics consistently approximate imaginary-time evolution.

\subsubsection{Noisy simulation setup}\label{App:shotNoise}
We detail the noise models employed in the numerical simulations for the generative modeling tasks.
In practical settings, beyond errors arising from product-formula approximations, additional imperfections stem from the finite number of measurement shots used to estimate observables and from intrinsic noise in quantum hardware.
To assess the practical performance of the proposed \hiddenlink{poly-DBTFD}{\texttt{poly} \texttt{DB-TFD}} method in comparison with VarQITE, we incorporate two sources of noise: (1) statistical noise due to finite measurement resources for estimating the energy and its variance, and (2) depolarizing noise associated with imperfect quantum gates.

\begin{itemize}
    \item \textbf{Statistical noise} --
    To model the statistical uncertainty arising from a finite number of measurement shots, we define $N_{\text{shots}}$ as the total measurement budget per iteration. We assume this total budget is uniformly allocated across all $M$ unique Pauli operators required by the respective algorithm. More specifically, $N_{\rm{Paulis}}\le M$ is the total number of Paulis needed to estimate the energy and variance for the DB-TFD approach, and the Paulis required to construct the matrix $A$ and vector $C$ for VarQITE. Consequently, each individual Pauli string is measured $N_P =N_{\rm{shots}}/N_{Paulis}$ times.
    \\
    For a given Hamiltonian $H=\sum_i c_i P_i$, decomposed into Pauli operators $\{P_i\}$, we sample a noisy estimator for the expectation value of each required Pauli string $P_i$ according to $\langle \hat{P_i}\rangle \sim\mathcal{N}(\langle P_i \rangle,\frac{1-\langle P_i\rangle ^2}{N_P}
    )$, where $\langle P_i\rangle=\braket{\phi|P_i|\phi}$ for the given state $\ket{\phi}$.
    The estimated energy is then obtained as $\hat{E}=\sum_i c_i\langle \hat{P_i} \rangle$.
    To estimate the variance, we apply the same sampling procedure to the expectation values of the Pauli terms appearing in the decomposition of $H^2$. 
    This yields an estimator $\hat{V} = \langle \hat{H}^2 \rangle - \hat{E}^2$.
    Due to statistical fluctuations, $\hat{V}$ may become negative even though the exact variance is non-negative.
    To avoid this unphysical behavior, we regularize the estimator according to $\hat{V}\rightarrow\max(\hat{V},c)$, thereby imposing a lower bound $c$ on the estimated variance.
    Throughout the numerical simulations, we set $c=10^{- 10}$.
    For modeling shot noise in the entries of the $A$ matrix appearing in VarQITE, we follow Ref.~\cite{McArdle_2019}, where each entry $\hat{A}_{i,j}$ is sampled from $\mathcal{N}(A_{i,j}, \frac{1}{N_P}(1/16-A_{i,j}^2))$.
    
    \item \textbf{Depolarizing noise} -- 
    We employ a depolarizing channel~\cite{nielsen2010quantum} to model quantum noise in actual hardware. 
    For a per-gate error rate $q$ and a circuit composed of $N_g$ gates, the resulting noisy state $\rho_{{\rm noisy}}$ is given by 
    \begin{equation}
        \rho_{{\rm noisy}} = (1-q)^{N_g} \rho +(1-(1-q)^{N_g})\frac{I}{N},
    \end{equation}
    where $\rho$ denotes the noiseless quantum state.

    For \texttt{DB-TFD}, we take the number of queries, namely Hamiltonian evolutions and reflection operations, to define the total gate count $N_g$, while for VarQITE we use the number of layers as the total gate count instead to better correspond with the queries of \texttt{DB-TFD}.
    
\end{itemize}
\end{document}